\documentclass[
paper=letter,%
numbers=noendperiod,%
captions=nooneline,%
DIV=10%
]{scrartcl}

\usepackage[T1]{fontenc}%
\usepackage{lmodern}%
\usepackage[american]{babel}%
\usepackage{microtype}%

\usepackage[includeheadfoot, top=0.7in, bottom=0.35in, left=0.8in, right=0.8in]{geometry}%

\usepackage[
hyperref,%
table%
]{xcolor}%

\usepackage{scrlayer-scrpage}%

\usepackage{calc}%
\usepackage{rotating}%
\usepackage{amsmath}%
\usepackage{mathtools}%
\usepackage{amssymb}%
\usepackage{amsthm}%
\usepackage{thmtools}%
\usepackage{etoolbox}%
\usepackage{xparse}%
\usepackage{bm}%
\usepackage{bbm}%
\usepackage{enumitem}%
\usepackage[calc]{datetime2}%
\usepackage{graphicx}%
\usepackage{grffile}%
\usepackage{tikz}%
\usepackage{wrapfig}%
\usepackage{tabularx}%
\usepackage{siunitx}%
\usepackage{booktabs}%
\usepackage{multirow}%
\usepackage{fancyvrb}%
\usepackage{textcomp}%
\usepackage{listings}%
\usepackage{csquotes}%
\usepackage[
style=authoryear,%
dashed=false,%
hyperref=true,%
useprefix=true,%
maxnames=2,%
uniquelist=minyear,%
maxbibnames=6,%
uniquename=false,%
seconds=true,%
date=iso,%
urldate=iso%
]{biblatex}%
\usepackage[
hypertexnames=false,%
setpagesize=false,%
pdfborder={0 0 0},%
pdfstartview=Fit,%
bookmarksopen=true,%
bookmarksnumbered=true%
]{hyperref}%
\usepackage[hyphens]{xurl}
\hypersetup{breaklinks=true}
\setkomafont{pageheadfoot}{\normalfont\normalcolor\sffamily}%
\setkomafont{pagenumber}{\normalfont\normalcolor\sffamily}%
\automark{section}%
\setkomafont{captionlabel}{\normalfont\normalcolor\sffamily\bfseries}%

\newcommand*{\mysquare}{\rule[0.18em]{0.36em}{0.36em}}
\newcommand*{\mytriangle}{\raisebox{0.12em}{\resizebox{0.48em}{0.48em}{$\blacktriangleright$}}}
\newcommand*{\mybar}{\rule[0.32em]{0.62em}{0.08em}}
\newcommand*{\mydot}{\raisebox{0.14em}{\resizebox{0.44em}{!}{$\bullet$}}}
\setlist{%
  align=left,%
  labelindent=0mm, %
  leftmargin=!,%
  itemindent=0mm, %
  listparindent=\parindent,%
  parsep=0mm,%
  topsep=1mm,%
  itemsep=1mm%
}
\setlist[itemize,1]{label={\mysquare\ }, labelwidth=\widthof{\mysquare\ }}%
\setlist[itemize,2]{label={\mytriangle\ }, labelwidth=\widthof{\mytriangle\ }}%
\setlist[itemize,3]{label={\mybar\ }, labelwidth=\widthof{\mybar\ }}%
\setlist[itemize,4]{label={\mydot\ }, labelwidth=\widthof{\mydot\ }}%
\setlist[enumerate,1]{label=\arabic*), labelwidth=\widthof{9)}}%
\setlist[enumerate,2]{label=\arabic{enumi}.\arabic*), labelwidth=\widthof{9.9)}}%
\setlist[enumerate,3]{label=\arabic{enumi}.\arabic{enumii}.\arabic*), labelwidth=\widthof{9.9.9)}}%
\setlist[enumerate,4]{label=\arabic{enumi}.\arabic{enumii}.\arabic{enumiii}.\arabic*), labelwidth=\widthof{9.9.9.9)}}%

\allowdisplaybreaks
\newcommand*{\abstractnoindent}{}%
\let\abstractnoindent\abstract
\renewcommand*{\abstract}{\let\quotation\quote\let\endquotation\endquote
  \abstractnoindent}
\deffootnote[1em]{1em}{1em}{\textsuperscript{\thefootnotemark}}%
\renewcommand*{\big}[1]{{\vcenter{\hbox{\scalebox{1.30}{\ensuremath#1}}}}}%

\definecolor{blue}{RGB}{58, 95, 205}%
\definecolor{red}{RGB}{205, 41, 144}%
\definecolor{orange}{RGB}{238, 118, 0}%
\definecolor{chocolate}{RGB}{205, 102, 29}%

\lstdefinestyle{input}{
  backgroundcolor=\color{black!12},%
  commentstyle=\itshape\color{black!50},%
  keywordstyle=\bfseries\color{black},%
  stringstyle=\color{black}%
}

\lstdefinestyle{output}{
  backgroundcolor=\color{black!6}%
}

\lstdefinestyle{codestyle}{
  language={},%
  keywords={},%
  otherkeywords={}%
}

\lstnewenvironment{codeinput}[1][]{%
  \lstset{style=input, style=codestyle}
  #1%
}{\vspace{-0.25\baselineskip}}%

\lstnewenvironment{codeoutput}[1][]{%
  \lstset{style=output, style=codestyle}
  #1%
}{\vspace{-0.25\baselineskip}}%

\expandafter\let\csname Sinput\endcsname\relax
\expandafter\let\csname endSinput\endcsname\relax
\expandafter\let\csname Soutput\endcsname\relax
\expandafter\let\csname endSoutput\endcsname\relax

\lstdefinestyle{Rstyle}{
  language=R,%
  keywords={},%
  otherkeywords={}%
}

\lstnewenvironment{Sinput}[1][]{%
  \lstset{style=input, style=Rstyle}
  #1%
}{\vspace{-0.25\baselineskip}}%

\lstnewenvironment{Soutput}[1][]{%
  \lstset{style=output, style=Rstyle}
  #1%
}{\vspace{-0.25\baselineskip}}%

\lstdefinestyle{Cstyle}{
  language=C,%
  keywords={},%
  otherkeywords={}%
}

\lstnewenvironment{Cinput}[1][]{%
  \lstset{style=input, style=Cstyle}
  #1%
}{\vspace{-0.25\baselineskip}}%

\lstnewenvironment{Coutput}[1][]{%
  \lstset{style=output, style=Cstyle}
  #1%
}{\vspace{-0.25\baselineskip}}%

\lstdefinestyle{Bashstyle}{
  language=bash,%
  keywords={},%
  otherkeywords={}%
}

\lstnewenvironment{Bashinput}[1][]{%
  \lstset{style=input, style=Bashstyle}
  #1%
}{\vspace{-0.25\baselineskip}}%

\lstnewenvironment{Bashoutput}[1][]{%
  \lstset{style=output, style=Bashstyle}
  #1%
}{\vspace{-0.25\baselineskip}}%

\lstdefinestyle{LaTeXstyle}{
  language=[LaTeX]TeX,%
  texcs={},%
  keywords={},%
  otherkeywords={}%
}

\lstnewenvironment{LaTeXinput}[1][]{%
  \lstset{style=input, style=LaTeXstyle}
  #1%
}{\vspace{-0.25\baselineskip}}%

\lstnewenvironment{LaTeXoutput}[1][]{%
  \lstset{style=output, style=LaTeXstyle}
  #1%
}{\vspace{-0.25\baselineskip}}%

\DeclareNameAlias{sortname}{family-given}%
\DefineBibliographyExtras{american}{\DeclareQuotePunctuation{}}%
\renewbibmacro*{volume+number+eid}{%
  \setunit*{\addcomma\space}%
  \printfield{volume}%
  \printfield{number}}
\DeclareFieldFormat*{number}{(#1)}
\DeclareFieldFormat*{title}{#1}%
\DeclareFieldFormat{doi}{%
  \ifhyperref
    {\href{http://dx.doi.org/#1}{\nolinkurl{doi:#1}}}%
    {\nolinkurl{doi:#1}}}%
\renewbibmacro*{in:}{}%
\DeclareFieldFormat{isbn}{ISBN #1}%
\DeclareFieldFormat{pages}{#1}%
\DeclareFieldFormat{url}{\url{#1}}%
\DeclareFieldFormat{urldate}{\mkbibparens{#1}}%
\renewcommand*{\cite}[2][]{\textcite[#1]{#2}}%

\newif\ifstarttheorem
\declaretheoremstyle[%
  spaceabove=0.5em,
  spacebelow=0.5em,
  headfont=\sffamily\bfseries\global\starttheoremtrue,
  notefont=\sffamily\bfseries,
  notebraces={(}{)},
  headpunct={},
  bodyfont=\normalfont,
  postheadspace=\newline%
]{myMainStyle}
\declaretheorem[style=myMainStyle, numberwithin=section]{definition}%
\declaretheorem[style=myMainStyle, sibling=definition]{proposition}
\declaretheorem[style=myMainStyle, sibling=definition]{lemma}
\declaretheorem[style=myMainStyle, sibling=definition]{theorem}
\declaretheorem[style=myMainStyle, sibling=definition]{corollary}
\declaretheorem[style=myMainStyle, sibling=definition]{remark}
\declaretheorem[style=myMainStyle, sibling=definition]{example}
\declaretheorem[style=myMainStyle, sibling=definition]{algorithm}

\makeatletter
\preto\itemize{%
  \if@inlabel
    \ifstarttheorem
      \mbox{}\par\nobreak\vskip\glueexpr-\parskip-\baselineskip+0.25em\relax\hrule\@height\z@
    \fi%
  \fi%
  \global\starttheoremfalse%
 \def\tempa{proof}%
 \ifx\tempa\mycurrenvir
    \ifstarttheorem
      \mbox{}\par\nobreak\vskip\glueexpr-\parskip-\baselineskip+0.25em\relax\hrule\@height\z@
    \fi%
 \fi%
 \global\starttheoremfalse%
}
\preto\enditemize{\global\starttheoremfalse}
\makeatother

\makeatletter
\preto\enumerate{%
  \if@inlabel
    \ifstarttheorem
      \mbox{}\par\nobreak\vskip\glueexpr-\parskip-\baselineskip+0.25em\relax\hrule\@height\z@
    \fi%
  \fi%
  \global\starttheoremfalse%
 \def\tempa{proof}%
 \ifx\tempa\mycurrenvir
    \ifstarttheorem
      \mbox{}\par\nobreak\vskip\glueexpr-\parskip-\baselineskip+0.25em\relax\hrule\@height\z@
    \fi%
 \fi%
 \global\starttheoremfalse%
}
\preto\endenumerate{\global\starttheoremfalse}
\makeatother

\newcommand{\mygrey}[1]{{\color{gray}#1}}
\makeatletter
\NewDocumentCommand{\tmb}{O{0.1mm} O{0.1mm} O{0.88} m m m}{%
  \mathrel{%
    \vbox{\offinterlineskip\m@th
      \ialign{%
        \hfil##\hfil\cr
        $\scriptscriptstyle\text{\scalebox{#3}{#4}}\mathstrut$\cr%
        \noalign{\vspace{#1}}%
        \vtop{%
          \ialign{%
            \hfil##\hfil\cr
            $#5$\cr\noalign{\vspace{#2}}%
            $\scriptscriptstyle\text{\scalebox{#3}{#6}}\mathstrut$\cr%
          }%
        }\cr
      }%
    }%
  }%
}
\makeatother

\makeatletter
\NewDocumentCommand{\tmbc}{O{0.1mm} O{0.1mm} O{0.88} m m m}{
  \mathrel{%
    \vbox{\offinterlineskip\m@th
      \ialign{%
        \hfil##\hfil\cr
        $\scriptscriptstyle\mathclap{\text{\scalebox{#3}{#4}}}\mathstrut$\cr%
        \noalign{\vspace{#1}}%
        \vtop{%
          \ialign{%
            \hfil##\hfil\cr
            $#5$\cr\noalign{\vspace{#2}}%
            $\scriptscriptstyle\mathclap{\text{\scalebox{#3}{#6}}}\mathstrut$\cr%
          }%
        }\cr
      }%
    }%
  }%
}
\makeatother

\newcommand*{\T}{^{\top}}

\newcommand*{\isim}{\tmb{\tiny{ind.}}{\sim}{}}

\newcommand*{\IN}{\mathbb{N}}

\newcommand*{\IR}{\mathbb{R}}

\newcommand*{\Exp}{\operatorname{Exp}}

\newcommand*{\Beta}{\operatorname{Beta}}

\newcommand*{\U}{\operatorname{Unif}}
\newcommand*{\B}{\operatorname{Ber}}
\newcommand*{\M}{\operatorname{M}}

\newcommand*{\I}{\mathbbm{1}}
\newcommand*{\rd}{\mathrm{d}}

\newcommand*{\LS}{\mathcal{LS}}
\newcommand*{\LSi}{\LS^{-1}}

\renewcommand*{\P}{\mathbb{P}}
\newcommand*{\E}{\mathbb{E}}

\newcommand*{\R}{\textsf{R}}
\newcommand*{\eps}{\varepsilon}

\def\calK{\mathcal{K}}

\newcommand*{\sort}[1]{\left[#1\right]}

\begin{document}
\thispagestyle{plain}
\begin{center}
  \sffamily
  {\bfseries\LARGE Index-mixed copulas\par}
  \bigskip\smallskip
  {\Large Klaus Herrmann\footnote{D{\'e}partement de
      Math{\'e}matiques, Universit{\'e} de Sherbrooke, \href{mailto:klaus.herrmann@usherbrooke.ca}{\nolinkurl{klaus.herrmann@usherbrooke.ca}}.},
    Marius Hofert\footnote{Department of Statistics and Actuarial Science, The University of
      Hong Kong, \href{mailto:mhofert@hku.hk}{\nolinkurl{mhofert@hku.hk}}.},
    Nahid Sadr\footnote{D{\'e}partement de
      Math{\'e}matiques, Universit{\'e} de Sherbrooke, \href{mailto:nahid.sadr@usherbrooke.ca}{\nolinkurl{nahid.sadr@usherbrooke.ca}}}
    \par
    \bigskip
    \today\par}
\end{center}
\par\smallskip
\begin{abstract}\textbf{Abstract}\newline\noindent
  The class of index-mixed copulas is introduced and its properties are
  investigated. Index-mixed copulas are constructed from given base copulas and
  a random index vector, and show a rather remarkable degree of analytical
  tractability. The analytical form of the copula and, if it exists, its density
  are derived. As the construction is based on a stochastic representation,
  sampling algorithms can be given. Properties investigated include bivariate
  and trivariate margins, mixtures of index-mixed copulas, symmetries such as
  radial symmetry and exchangeability, tail dependence, measures of concordance
  such as Blomqvist's beta, Spearman's rho or Kendall's tau and concordance
  orderings. Examples and illustrations are provided, and applications to the
  distribution of sums of dependent random variables as well as the stress
  testing of general dependence structures are given. A particularly interesting
  feature of index-mixed copulas is that they allow one to provide a revealing
  interpretation of the well-known family of Eyraud--Farlie--Gumbel--Morgenstern
  (EFGM) copulas. Through the lens of index-mixing, one can explain why EFGM
  copulas can only model a limited range of concordance and are tail
  independent, for example. Index-mixed copulas do not suffer from such
  restrictions while remaining analytically tractable.
\end{abstract}
\minisec{Keywords}
Copulas, index-mixing, stochastic representation, properties, mixtures,
Eyraud--Farlie--Gumbel--Morgenstern, sums of dependent random variables %
\minisec{MSC2010}
60E05, %
60G09, %
62E15, %
62H99, %
62H86, %
62H20%

\section{Introduction}\label{sec:intro}
The notion of copulas often provides a convenient tool for describing stochastic
dependence between random variables. For $d\ge 2$, a $d$-dimensional
\emph{copula} is a $d$-dimensional distribution function with standard uniform
univariate margins, restricted to the $d$-dimensional unit hypercube. By Sklar's
theorem, see \cite{sklar1959} or \cite{sklar1996}, any multivariate distribution
function can be decomposed into its margins and a copula, and the composition of
any copula with univariate marginal distribution functions leads to a valid
multivariate distribution function.  In applications under stochastic
dependence, it is typically required to be able to sample from a given copula,
for which stochastic representations are commonly used. A \emph{stochastic
  representation} is a representation of a random vector by transformations of
simpler or more fundamental random variables from well-known distributions (such
as uniform, exponential or normal; hence easy to simulate from). Stochastic
representations of random vectors are also important for understanding how
dependence arises among the components of a random vector and thus for modeling
purposes.

In this paper, we introduce the idea of index-mixing stochastic representations
to construct stochastic representations of new copulas termed index-mixed
copulas. By combining a number of initial base copulas into a
    new one, our approach falls into the wider literature of copula-to-copula
    transformations.  Such transformations factor into bivariate and fully
    multivariate approaches.  For example, bivariate approaches include the use
    of co-copulas \parencite{Girard2018,MANSTAVICIUS201948}, integral
    transformation approaches \parencite{popovic2023integral}, individual
    transformations of arguments \textcite{XIE201920}, binary aggregation
    operators \parencite{durante2009construction,kolesarova2013new} or weighted
    geometric means \parencite{cuadras2009constructing}. Approaches in the
    multivariate case include transformations based on univariate distortions
    \parencite{Morillas2005}, generalizations of Bernstein copula constructions
    \parencite{YangChenWangWang2015,XieFangYangBu2022} leading to multivariate
    composite copulas, generalizing product constructions
    \parencite{Khoudraji1995,Liebscher2008,MAZO2015363}, generalizations of
    distortions and product constructions \parencite{fischer2012constructing} or
    uniformity preserving transformations
    \parencite{McNeil2021modelling,hofert2025w}.

As the construction in this article is based on a stochastic representation,
efficient sampling algorithms can be provided. We investigate various
properties of index-mixed copulas. In particular, we derive the analytical form
of index-mixed copulas and of their densities (if the latter
exist). Furthermore, the analytical form of bivariate and trivariate margins can
be derived, as well as that of mixtures of index-mixed copulas. We also
investigate symmetries such as radial symmetry and exchangeability. In terms of
measures of association, tail dependence coefficients and the measures of
concordance Blomqvist's beta, Spearman's rho and Kendall's tau are addressed,
some even in the multivariate case. Further results include orthant dependence
and concordance orderings. Moreover, the distribution of the sum of
index-mixed-dependent exponential random variables can be identified. Examples
and illustrations are also provided, and an application to stress testing of
general dependence structures is outlined. Through the idea of index-mixing, we
can also revisit the construction of Eyraud--Farlie--Gumbel--Morgenstern (EFGM)
copulas from a different point of view, explaining why they can only model a
limited range of concordance and are tail independent, and why they have a
restricted parameter space -- none of these issues exist for index-mixed
copulas.

Section~\ref{sec:index:mixing} introduces the construction principle of
index-mixed copulas and provides insight into how they can be
interpreted. Section~\ref{sec:index:mixed:properties} presents various
properties of index-mixed copulas, including the aforementioned
ones. Section~\ref{sec:examples} provides example illustrations and outlines an
application. Finally, Section~\ref{sec:concl} contains concluding remarks and
avenues for future research on index-mixed copulas. The appendix provides
the aforementioned details about EFGM copulas.

\section{Construction of index-mixed copulas}\label{sec:index:mixing}
Let $e_k^K=(0,\dots,0,1,0,\dots,0)\T$, $k=1,\dots,K$, denote the $k$th standard
basis (column-)vector of $\IR^K$, $K\in\IN$. We now define the main
ingredients to construct and comprehend $d$-dimensional index-mixed copulas;
note that the symbol $\biguplus$ is used to denote a disjoint union throughout this work.

The eventual definition of index-mixed copulas and their properties are based on certain random matrices and related quantities which we now introduce.
To this end let $K\in\IN$ and $d\in\IN$, $d\ge 2$, and let
  $\bm{I}=(I_1,\ldots,I_d)\in\{1,\dots,K\}^d$ be the random \emph{index vector} with \emph{index
    distribution} $F_{\bm{I}}$.
  Furthermore, let $I^{\text{mat}}=(e_{I_1}^{K},\dots,e_{I_d}^{K})\T \in\{0,1\}^{d\times K}$
  be the associated random \emph{index matrix}, where the $j$th row of $I^{\text{mat}}$ is
  $(e^K_{I_j})\T=(0,\dots,0,1,0,\dots,0)$ (with the $1$ being located in position $I_j$
  of this vector of length $K$),
  $j=1,\dots,d$. %
  Associated with $I^{\text{mat}}$ comes the random \emph{index partition}
  $\biguplus_{k=1}^{K} J_k$ of $\{1,\dots,d\}$, with $k$th partition
  element $J_k=\{j\in\{1,\dots,d\}:I_{j,k}^{\text{mat}}=1\}$ being the set of all row indices
  $j\in\{1,\dots,d\}$ such that $I_{j,k}^{\text{mat}}=1$, and the \emph{size} of the $k$th partition
  element $J_k$ is $D_k=|J_k|=\sum_{j=1}^d I_{j,k}^{\text{mat}}$, $k=1,\dots,K$, with
  $\sum_{k=1}^KD_k=d$; note that $D_k$ is the number of $1$s in the $k$th
  column of $I^{\text{mat}}$ and also the number of times the index vector $\bm{I}$ contains $k$.
Even though some quantities (such as the index partition $\{J_1,\dots,J_K\}$) are not needed for Definition~\ref{def:index:mixed:copulas}, they will turn
out to be useful in proofs of statements about index-mixed copulas later.
We calculate the introduced  quantities explicitly in the following example which
is continued in Example~\ref{ex:realization:1.5} below and Example~\ref{ex:realization:2} in
Section~\ref{sec:analytical:form}.
\begin{example}[Index vector, index matrix and index partition]\label{ex:realization:1}
This example calculates the index matrix $I^{\text{mat}}$ and related quantities for a given realization of the index vector $\bm{I}$.
Assume $d=4$ and $K=5$ and that the realization of $\bm{I} = (I_1,I_2,I_3,I_4)$ is $(3,1,2,1)$.
In this case, the index matrix $I^{\text{mat}}$ is given by
\begin{align}
I^{\text{mat}} =
\begin{pmatrix}\label{eq:index:matrix}
0 & 0 & 1 & 0 & 0\\
1 & 0 & 0 & 0 & 0\\
0 & 1 & 0 & 0 & 0\\
1 & 0 & 0 & 0 & 0
\end{pmatrix},
\end{align}
where we can see that each row of $I^{\text{mat}}$ is a base vector of $\IR^5$.
Concerning the associated partition of $\{1,2,3,4\}$ we have
$J_1 = \{2,4\}$, $J_2 = \{3\}$, $J_3 = \{1\}$ and $J_4=J_5=\emptyset$, leading to $D_1 = 2$, $D_2=D_3=1$ and $D_4=D_5=0$.
\end{example}
The concept of index vectors now leads to the following copula construction.
\begin{definition}[Index-mixed copulas]\label{def:index:mixed:copulas}
For $K\in\IN$ and $d\in\IN$, $d\ge 2$, let $\bm{I}=(I_1,\ldots,I_d)\in\{1,\dots,K\}^d$ be a random index vector.
  Moreover, let $C_k$, $k=1,\dots,K$, be $d$-dimensional copulas referred to as
  \emph{base copulas}, and, for independent $\tilde{\bm{U}}_k\sim C_k$,
  $k=1,\dots,K$, let the $(d,K)$-matrix $\tilde{U}=(\tilde{\bm{U}}_k)_{k=1}^K$ be
  the random \emph{copula matrix} with $k$th column being $\tilde{\bm{U}}_k$ for
  all $k=1,\dots,K$.
  The $d$-dimensional \emph{index-mixed copula (of order $K$)} is defined as
  the $d$-dimensional distribution function of
  \begin{align}
    \bm{U}=\tilde{U}_{\bm{I}}:=\begin{pmatrix}\tilde{U}_{1,I_1}\\ \vdots \\ \tilde{U}_{d,I_d}\end{pmatrix},\label{U:index:mixed}
  \end{align}
  restricted to the unit hypercube.
\end{definition}
Starting with copula-dependent columns in the copula matrix provides a natural
way to introduce dependence. The index vector selects, in each dimension
$j=1,\dots,d$, precisely one element of the $j$th row of the copula matrix.
This guarantees that $U_j\sim\U(0,1)$ and thus that $\bm{U}$ indeed follows a
copula (the index-mixed copula). We therefore also have no restriction on the
index distribution, another advantage of the construction.

\begin{example}[Realization of index-mixed copulas; continuation of Example~\ref{ex:realization:1}]\label{ex:realization:1.5}
Continuing with the data in Example~\ref{ex:realization:1}, where $\bm{I}=(3,1,2,1)$, consider now a realization
\begin{align*}
\tilde{u} =
\begin{pmatrix}
\mygrey{0.3901599} & \mygrey{0.6312906} & 0.8859627 & \mygrey{0.4008808} & \mygrey{0.3877381}\\
0.4303661 & \mygrey{0.5132919} & \mygrey{0.6224755} & \mygrey{0.2833943} & \mygrey{0.3537322}\\
\mygrey{0.6463193} & 0.2398525 & \mygrey{0.9671771} & \mygrey{0.4003459} & \mygrey{0.3565862}\\
0.9121791 & \mygrey{0.3758194} & \mygrey{0.8885885} & \mygrey{0.7227110} & \mygrey{0.5034753}
\end{pmatrix}
\end{align*}
of the copula matrix
$\tilde{U}_{\bm{I}}\in[0,1]^{4\times 5}$. The realization of the index-mixed copula $\bm{U}$ is then given as
$(\tilde{u}_{1,3},\tilde{u}_{2,1},\tilde{u}_{3,2},\tilde{u}_{4,1})=(
0.8859627,
0.4303661,
0.2398525,
0.9121791
)$, where
$\tilde{u}_{1,3}$ is a realization of the first margin of $C_3$,
$(\tilde{u}_{2,1},\tilde{u}_{4,1})$ is a realization of the $(2,4)$-margin of
$C_1$ and $\tilde{u}_{3,2}$ is a realization of the third margin of $C_2$.
\end{example}

We conclude this section with an example where we illustrate index-mixing by means of scatter plots.
The necessary sampling algorithms are discussed in Section~\ref{sec_sampling} below, while the computation of the presented tail dependence coefficients is discussed in Section~\ref{sec_tail_dependence}.
More illustrations highlighting the construction of index-mixed copulas can be found in Figures~\ref{fig:base:cop:samples}--\ref{fig:IMC:2d:I:comon} in Section~\ref{sec:examples}.
\begin{example}[Illustrating the effect of index-mixing for $d=4$]\label{ex:ill:index:mixing:4d}
  Figure~\ref{fig:IMC:4d} shows the scatter-plot matrix of a sample of size
  10\,000 from the $4$-dimensional index-mixed copula of order $K=2$ with base
  copulas being the Gumbel copula $C^{\text{G}}$ and the (homogeneous) Gaussian copula $C^{\text{Ga}}$,
  with correlation parameter such that
  Kendall's tau equals $0.5$. The distribution of the index vector is specified
  via $p_{\bm{I}}(1,1,2,2)=1/2$, $p_{\bm{I}}(1,2,1,2)=1/3$,
  $p_{\bm{I}}(2,2,1,1)=1/6$.
  \begin{figure}[htbp]
    \hfill
    \includegraphics[width=0.8\textwidth]{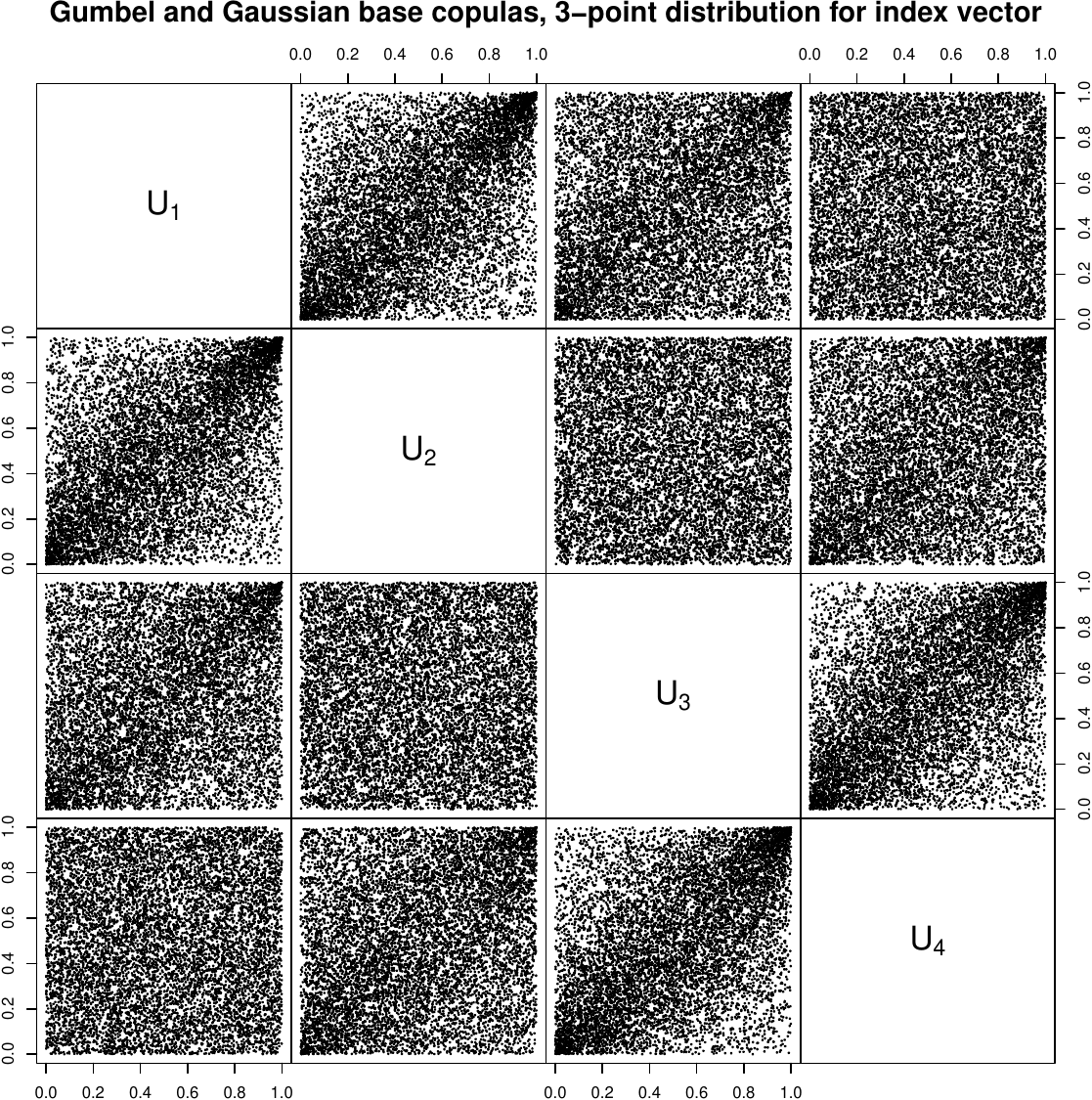}%
    \hspace*{\fill}
    \caption{Scatter-plot matrix of a sample of size 10\,000 from the
      index-mixed copula with $p_{\bm{I}}(1,1,2,2)=1/2$,
      $p_{\bm{I}}(1,2,1,2)=1/3$, $p_{\bm{I}}(2,2,1,1)=1/6$ and base copulas
      being the Gumbel copula $C^{\text{G}}$ and the Gaussian copula $C^{\text{Ga}}$.}
    \label{fig:IMC:4d}
  \end{figure}
  The bivariate margins of this index-mixed copula are given in Example~\ref{ex:biv:index:mixed}, while the pairwise upper tail dependence coefficients are computed just after Proposition~\ref{prop:tail:dependence}.
\end{example}

\section{Properties of index-mixed copulas}\label{sec:index:mixed:properties}
\subsection{Analytical form of index-mixed copulas}\label{sec:analytical:form}
  In the following Theorem~\ref{thm:index:mixed:form:C} we derive the analytical
  form of index-mixed copulas, while Corollary~\ref{cor:index:mixed:form:c}
  derives the corresponding density function.
  \begin{theorem}[Index-mixed copulas]\label{thm:index:mixed:form:C}
  Consider the setup of Definition~\ref{def:index:mixed:copulas} and let
  $I_{\cdot,k}^{\text{mat}}$ denote the $k$th column of the index matrix $I^{\text{mat}}$ associated to $\bm{I}$, as well as
  $\bm{u}^{I_{\cdot,k}^{\text{mat}}}=(u_1^{I_{1,k}^{\text{mat}}},\dots,u_d^{I_{d,k}^{\text{mat}}})$ (with the usual
  convention that $0^0=1$). Then
  $\bm{U}$ in~\eqref{U:index:mixed} follows the copula
  \begin{align}
    C(\bm{u})=\E_{\bm{I}}\biggl[\,\prod_{k=1}^KC_k(\bm{u}^{I_{\cdot, k}^{\text{mat}}})\biggr],\quad\bm{u}\in[0,1]^d.\label{eq:index:mixed:cop}
  \end{align}
\end{theorem}

\begin{remark}[Interpretation as mixture with partition of arguments]\label{rem:constr}
  In the remainder of the paper, we derive a rather remarkable number of results
  analytically. %
  It often helps to realize that each product
  $\prod_{k=1}^KC_k(\bm{u}^{I_{\cdot,k}^{\text{mat}}})$ in
  $C(\bm{u})=\E_{\bm{I}}\bigl[\,\prod_{k=1}^KC_k(\bm{u}^{I_{\cdot,k}^{\text{mat}}})\bigr]$ is
  a product of factors that depend on $u_1,\dots,u_d$ in such a way that each
  $u_j$ appears in precisely one factor (so for precisely one
  $k=1,\dots,K$). This holds because the (random) index matrix $I^{\text{mat}}$ generates a
  (random) index partition of $\{1,\ldots,d\}$; see Section~\ref{sec:index:mixing}.
 With the help of $J_1,\ldots,J_K$ we
  can more precisely say that $\prod_{k=1}^KC_k(\bm{u}^{I_{\cdot,k}^{\text{mat}}})$ in
  $C(\bm{u})=\E_{\bm{I}}\bigl[\,\prod_{k=1}^KC_k(\bm{u}^{I_{\cdot,k}^{\text{mat}}})\bigr]$ is
  a product of factors that depend on $u_1,\dots,u_d$ in such a way that $u_j$
  appears in the $k$th factor if and only if $j\in J_k$.  In the special case
  where there exists an integer $m$ that divides $d$, the index distribution is
  such that the only index partitions $\biguplus_{k=1}^{K} J_k$ occurring with
  positive probability are those with $D_k=|J_k|=m$ for all $k=1,\dots,K$, so
  the resulting index-mixed copula $C$ is a mixture where all components are of
  equal dimension $m$.
\end{remark}

To express the index distribution $F_{\bm{I}}$, we use the notation
$p_{\bm{I}}(\bm{k})=\P(\bm{I}=\bm{k})$, $\bm{k}\in\{1,\dots,K\}^d$, for its
probability mass function in what follows. As a simple example of
Theorem~\ref{thm:index:mixed:form:C}, we consider the form of bivariate index-mixed copulas,
and express them in terms of $p_{(I_1,I_2)}(k_1,k_2)=\P(I_1=k_1,I_2=k_2)$,
$k_1,k_2 \in \{1\dots,K\}$. For an illustration, see Example~\ref{ex:ill:index:mixing:2d} later.
\begin{example}[Bivariate index-mixed copulas]\label{ex:biv:index:mixed}
  If $d=2$ and $p_{(I_1,I_2)}(k_1,k_2)=\P(I_1=k_1,I_2=k_2)$, $k_1,k_2 \in \{1\dots,K\}$, then
  \begin{align}\label{eq:Copula:2D}
    C(u_1,u_2)=\sum_{k=1}^Kp_{(I_1,I_2)}(k,k)C_k(u_1,u_2)+\biggl(\,\sum_{k_1=1}^K
    \sum_{\substack{k_2=1\\ k_2\neq k_1}}^K p_{(I_1,I_2)}(k_1,k_2)\biggr)\Pi(u_1,u_2),
  \end{align}
  where $\Pi(u_1,u_2)=u_1u_2$ denotes the \emph{independence copula}, which is
  (in any dimension) simply the product of its
  components;  %
  see also (the proof of) Lemma~\ref{lem:biv:mar} Part~\ref{lem:biv:mar:1}
  below.
  In Example~\ref{ex:ill:index:mixing:4d}, the bivariate marginal copulas are given by
  \begin{align*}
    \begin{pmatrix}
      U_1 & \frac{C^{\text{G}}}{2}+\frac{\Pi}{3}+\frac{C^{\text{Ga}}}{6} & \frac{C^{\text{G}}}{3} + \frac{2\Pi}{3} & \Pi \\
      \frac{C^{\text{G}}}{2}+\frac{\Pi}{3}+\frac{C^{\text{Ga}}}{6} & U_2 & \Pi & \frac{C^{\text{Ga}}}{3} + \frac{2\Pi}{3} \\
      \frac{C^{\text{G}}}{3} + \frac{2\Pi}{3} & \Pi & U_3 & \frac{C^{\text{Ga}}}{2}+\frac{\Pi}{3}+\frac{C^{\text{G}}}{6} \\
      \Pi & \frac{C^{\text{Ga}}}{3} + \frac{2\Pi}{3} & \frac{C^{\text{Ga}}}{2}+\frac{\Pi}{3}+\frac{C^{\text{G}}}{6} & U_4
    \end{pmatrix},
  \end{align*}
  where $\frac{C^{\text{G}}}{2}+\frac{\Pi}{3}+\frac{C^{\text{Ga}}}{6}$ stands for the
  $(1/2,1/3,1/6)$-mixture of the appearing Gumbel copula, independent components and the Gaussian copula, for example.
\end{example}

The following corollary addresses the case of absolutely continuous index-mixed
copulas, that is index-mixed copulas with a probability density function.
In this context it is important to recall that not all copulas have densities, where the comonotonic copula $M(\bm{u}):=\min(u_1,\dots,u_d)$ serves as a fundamental example.
Concerning the following discussion note that
$C_k(\bm{u}^{I_{\cdot,k}^{\text{mat}}})=C_k^{I_{\cdot,k}^{\text{mat}}}(\bm{u}_{I_{\cdot,k}^{\text{mat}}})$, where
$C_k^{I_{\cdot,k}^{\text{mat}}}$ is the $D_k$-dimensional marginal copula of $C_k$
corresponding to all components of $I_{\cdot,k}^{\text{mat}}$ that are $1$ and where
$\bm{u}_{I_{\cdot,k}^{\text{mat}}}$ is the $D_k$-dimensional vector consisting of all
components of $\bm{u}$ for which $I_{\cdot,k}^{\text{mat}}$ is $1$ (in the same order); the
notation $\bm{u}_{I_{\cdot,k}^{\text{mat}}}$ will later also be used for
$\bm{x}_{I_{\cdot,k}^{\text{mat}}}$ and $\bm{t}_{I_{\cdot,k}^{\text{mat}}}$ when indexing arguments of
multivariate distribution functions or Laplace--Stieltjes transforms, for
example.  In case $D_k=0$ we set $C_k^{I_{\cdot,k}^{\text{mat}}}(\bm{u}_{I_{\cdot,k}^{\text{mat}}}) = 1$
to obtain again
$C_k(\bm{u}^{I_{\cdot,k}^{\text{mat}}})=C_k^{I_{\cdot,k}^{\text{mat}}}(\bm{u}_{I_{\cdot,k}^{\text{mat}}})$.
With the notation $C_k^{I_{\cdot,k}^{\text{mat}}}(\bm{u}_{I_{\cdot,k}^{\text{mat}}})$, we can now express
the densities of index-mixed copulas.
\begin{corollary}[Densities of index-mixed copulas]\label{cor:index:mixed:form:c}
  If the base copula $C_k$ is absolutely continuous with density $c_k$ for all
  $k=1,\dots,K$, then the index-mixed copula $C$ is absolutely continuous with
  density
  \begin{align*}
    c(\bm{u})=\E_{\bm{I}}\biggl[\,\prod_{k=1}^Kc_k^{I_{\cdot,k}^{\text{mat}}}(\bm{u}_{I_{\cdot,k}^{\text{mat}}})\biggr],\quad\bm{u}\in(0,1)^d,
  \end{align*}
  where $c_k^{I_{\cdot,k}^{\text{mat}}}$ denotes the density of $C_k^{I_{\cdot,k}^{\text{mat}}}$, where we set
$c_k^{I_{\cdot,k}^{\text{mat}}} = 1$ in case $D_k=0$.
\end{corollary}

We also immediately obtain the form of distribution functions $H$ with
index-mixed copulas, and diagonals of index-mixed copulas; the latter are
important for computing tail dependence coefficients, for example.
\begin{corollary}[Meta index-mixed models and index-mixed copula diagonals]
  \begin{enumerate}
  \item By Sklar's Theorem, Theorem~\ref{thm:index:mixed:form:C} implies that
    a joint distribution function $H$ with margins $F_1,\dots,F_d$ and index-mixed copula takes the form
    \begin{align}
      H(\bm{x})=\E_{\bm{I}}\biggl[\,\prod_{k=1}^KC_k(F_1(x_1)^{I_{1,k}^{\text{mat}}},\dots,F_d(x_d)^{I_{d,k}^{\text{mat}}})\biggr]=\E_{\bm{I}}\biggl[\,\prod_{k=1}^KH_k^{I_{\cdot,k}^{\text{mat}}}(\bm{x}_{I_{\cdot,k}^{\text{mat}}})\biggr],\quad \bm{x}\in\IR^d,\label{eq:joint:df}
    \end{align}
    where, as before, $H_k^{I_{\cdot,k}^{\text{mat}}}(\bm{x}_{I_{\cdot,k}^{\text{mat}}})$ denotes the $D_k$-dimensional marginal
    distribution of $H_k$ corresponding to those $D_k$ entries in $I_{\cdot,k}^{\text{mat}}$ which are $1$ (and $H_k^{I_{\cdot,k}^{\text{mat}}}(\bm{x}_{I_{\cdot,k}^{\text{mat}}}) = 1$ if $D_k=0$).
  \item  By Theorem~\ref{thm:index:mixed:form:C}, the diagonal
    $\delta_C$ of an index-mixed copula $C$ is given by
    \begin{align*}
      \delta_C(u) = C(u,\ldots,u) = \E_{\bm{I}}\biggl[\,\prod_{k=1}^K\delta_k^{I_{\cdot, k}^{\text{mat}}}(u)\biggr],
    \end{align*}
    where $\delta_k^{I_{\cdot, k}^{\text{mat}}}$ is the diagonal of $C_k^{I_{\cdot,k}^{\text{mat}}}$.
    Therefore, the diagonal of an index-mixed copula is the index-mixture of the
    diagonals of the underlying base copulas.
  \end{enumerate}
\end{corollary}

We finish this section with a continuation of Example~\ref{ex:realization:1},
where we exemplify the calculation of the quantities utilized in
Theorem~\ref{thm:index:mixed:form:C} and Corollary~\ref{cor:index:mixed:form:c}.
\begin{example}[Continuation of Example~\ref{ex:realization:1}]\label{ex:realization:2}
  Continuing with the data in Example~\ref{ex:realization:1}, where
  $\bm{I}=(3,1,2,1)$ and therefore $J_1 = \{2,4\}$, $J_2 = \{3\}$, $J_3 = \{1\}$
  and $J_4=J_5=\emptyset$, we now consider a given evaluation point
  $\bm{u}=(u_1,u_2,u_3,u_4)$ for $C$.  Referring to the realization of the index
  matrix $I^{\text{mat}}$ in \eqref{eq:index:matrix} we then have
  \begin{align*}
    \bm{u}^{I_{\cdot, 1}^{\text{mat}}} &= (u_1^0,u_2^1,u_3^0,u_4^1) = (1,u_2,1,u_4),\\
    \bm{u}^{I_{\cdot, 2}^{\text{mat}}} &= (u_1^0,u_2^0,u_3^1,u_4^0) = (1,1,u_3,1),\\
    \bm{u}^{I_{\cdot, 3}^{\text{mat}}} &= (u_1^1,u_2^0,u_3^0,u_4^0) = (u_1,1,1,1),\\
    \bm{u}^{I_{\cdot, 4}^{\text{mat}}} &= (u_1^0,u_2^0,u_3^0,u_4^0) = (1,1,1,1),\\
    \bm{u}^{I_{\cdot, 5}^{\text{mat}}} &= (u_1^0,u_2^0,u_3^0,u_4^0) = (1,1,1,1),
  \end{align*}
  and hence
  $C_1(\bm{u}^{I_{\cdot, 1}^{\text{mat}}}) = C_1(1,u_2,1,u_4) = C_1^{2,4}(u_2,u_4)$.  Using
  the notation introduced after Theorem~\ref{thm:index:mixed:form:C}, we equally
  have $C_1^{I_{\cdot, 1}^{\text{mat}}} = C_1^{2,4}$ and $\bm{u}_{I_{\cdot, 1}^{\text{mat}}} = (u_2,u_4)$
  and therefore, as intended,
  $C_1(\bm{u}^{I_{\cdot, 1}^{\text{mat}}}) = C_1^{I_{\cdot, 1}^{\text{mat}}}(\bm{u}_{I_{\cdot, 1}^{\text{mat}}})$.
  Along the same lines, we have
  $C_2(\bm{u}^{I_{\cdot, 2}^{\text{mat}}}) = C_2(1,1,u_3,1) = C_2^3(u_3) = u_3 =
  C_2^{I_{\cdot, 2}^{\text{mat}}}(\bm{u}_{I_{\cdot, 2}^{\text{mat}}})$ as well as
  $C_3(\bm{u}^{I_{\cdot, 3}^{\text{mat}}}) = u_1 = C_3^{I_{\cdot, 3}^{\text{mat}}}(\bm{u}_{I_{\cdot,
      3}^{\text{mat}}})$.  Lastly, noting that $D_4=D_5=0$, we have
  $C_4(\bm{u}^{I_{\cdot, 4}^{\text{mat}}}) = C_4(1,1,1,1) = 1 = C_4^{I_{\cdot,
      4}^{\text{mat}}}(\bm{u}_{I_{\cdot, 4}^{\text{mat}}})$ and equally
  $C_5(\bm{u}^{I_{\cdot, 5}^{\text{mat}}}) = 1 = C_5^{I_{\cdot, 5}^{\text{mat}}}(\bm{u}_{I_{\cdot, 5}^{\text{mat}}})$.

Concerning the product $\prod_{k=1}^KC_k(\bm{u}^{I_{\cdot, k}^{\text{mat}}})$ we can now compute
\begin{align}\label{eq:concrete:product}
\prod_{k=1}^KC_k(\bm{u}^{I_{\cdot, k}^{\text{mat}}})
&= C_1(1,u_2,1,u_4) \cdot C_2(1,1,u_3,1) \cdot C_3(u_1,1,1,1) \cdot 1 \cdot 1 = C_1^{2,4}(u_2,u_4) u_3 u_1,
\end{align}
highlighting again the fact that each $u_j$, $j\in\{1,\ldots,d\}$, is in exactly one factor of the product, or, phrased more precisely as in Remark~\ref{rem:constr}, $u_j$ appears in factor $k$ if and only if $j\in J_k$ which can easily be verified with $J_1,\ldots,J_5$ given above.

Finally, when computing the value $C(\bm{u})$ it is now necessary to evaluate the expectation in \eqref{eq:index:mixed:cop}.
For every possible realization of $\bm{I}$ we therefore have to add up terms formed similarly to \eqref{eq:concrete:product}, weighted by their respective probabilities.
From this we see that $C$ is in fact a \emph{mixture} of products of the lower dimensional margins of the base copulas $C_k$, $k=1,\dots,K$.
\end{example}

\subsection{Sampling}\label{sec_sampling}
As we started the construction of index-mixed copulas from a stochastic
representation, sampling algorithms are available. The following algorithm
outlines the procedure how to obtain $n$ (pseudo-random) realization of $\bm{U}$
from an index-mixed copula $C$ with base copulas $C_1,\ldots,C_K$ and index
vector $\bm{I}$. %
\begin{algorithm}[Sampling $\bm{U}_1,\dots,\bm{U}_n$ sequentially]\label{algo:sample:one}
  \begin{enumerate}
  \item For $i=1,\dots,n$, do:
    \begin{enumerate}
    \item Sample the copula matrix $\tilde{U}_i$ via $\tilde{U}_i=(\tilde{\bm{U}}_{i,k})_{k=1}^K$,
      where $\tilde{\bm{U}}_{i,k}\isim C_k$ for all $k=1,\dots,K$.
    \item Independently of $\tilde{U}_i$, sample $\bm{I}_i\sim F_{\bm{I}}$.
    \end{enumerate}
  \item Return
    $\bm{U}_i=\tilde{U}_{i,\bm{I}_i}=(\tilde{U}_{i,1,I_{i,1}},\dots,\tilde{U}_{i,d,I_{i,d}})$,
    $i=1,\dots,n$, where $I_{i,j}$ denotes the $j$th entry of
    $\bm{I}_i$ and $\tilde{U}_{i,j,I_{i,j}}$ the entry $(j,I_{i,j})$ in
    $\tilde{U}_i$.
  \end{enumerate}
\end{algorithm}
Algorithm~\ref{algo:sample:one} iterates over $i=1,\dots,n$. The following
algorithm provides a vectorized (parallel) version by first drawing all copula
matrices, then all index vectors, and both of these parts can be done in
parallel.
\begin{algorithm}[Sampling $\bm{U}_1,\dots,\bm{U}_n$ vectorized]
  \begin{enumerate}
  \item For $i=1,\dots,n$, sample copula matrices $\tilde{U}_i$ via $\tilde{U}_i=(\tilde{\bm{U}}_{i,k})_{k=1}^K$,
    where $\tilde{\bm{U}}_{i,k}\isim C_k$ for all $k=1,\dots,K$, $i=1,\dots,n$.
  \item Independently of $\tilde{U}_i$, $i=1,\dots,n$,
    sample $\bm{I}_i\sim F_{\bm{I}}$, $i=1,\dots,n$.
  \item Return
    $\bm{U}_i=\tilde{U}_{i,\bm{I}_i}=(\tilde{U}_{i,1,I_{i,1}},\dots,\tilde{U}_{i,d,I_{i,d}})$,
    $i=1,\dots,n$, where $I_{i,j}$ denotes the $j$th entry of
    $\bm{I}_i$ and $\tilde{U}_{i,j,I_{i,j}}$ the entry $(j,I_{i,j})$ in
    $\tilde{U}_i$.
  \end{enumerate}
\end{algorithm}

The above algorithms should be fairly fast for typical choices as base copulas.
However, for a fixed $d$, the larger the $K$ and the smaller the support of $\bm{I}$
(the set of all elements of $\{1,\dots,K\}^d$ with positive probability under
the index distribution), the more wasteful the above general-purpose algorithms
are for sampling index-mixed copulas. For example, consider $K>2$ and an index
vector with support in $\{1,2\}^d$, then one does not need to sample $C_k$,
$k=3,\dots,K$, at all. Depending on $K$ and the support of $\bm{I}$, the
following algorithm may thus be faster.
\begin{algorithm}[Efficiently sampling $\bm{U}_1,\dots,\bm{U}_n$ sequentially]%
  \begin{enumerate}
  \item For $i=1,\dots,n$, do:
    \begin{enumerate}
    \item Sample $\bm{I}_i\sim F_{\bm{I}}$.
    \item Independently of $\bm{I}_i$, for $k=1,\dots,K$,
      sample the $D_{i,k}$-dimensional
      $\tilde{U}_{i,k}\sim C_k^{I_{i,\cdot,k}^{\text{mat}}}$, where $I_{i,\cdot,k}^{\text{mat}}$ denotes
      the $k$th column of the index matrix corresponding to the index vector
      $\bm{I}_i$.
    \item Concatenate the elements of the $K$ vectors $\tilde{U}_{i,1}$,
      \ldots, $\tilde{U}_{i,K}$ (possibly of length $1$) according to the
      (on $i$ dependent) index partition $\biguplus_{k=1}^{K} J_{i,k}$
      to build $\bm{U}_i$ in the sense that the components of
      $\tilde{U}_{i,k}$ appear in the positions given by $J_{i,k}$ in the order of
      $k=1,\dots,K$.
    \end{enumerate}
  \item Return $\bm{U}_i$, $i=1,\dots,n$.
  \end{enumerate}
\end{algorithm}

The following remark provides various options for the index distribution $F_{\bm{I}}$.
\begin{remark}[Different approaches for sampling $\bm{I}$]
  \begin{enumerate}
  \item For $K=2$, one way to specify a distribution for $\bm{I}$ is via
    $\bm{I}=\bm{1}+\bm{B}$, where $\bm{B}\sim\B^{C_{\bm{B}}}_d(\bm{p})$ denotes
    a random vector from a $d$-dimensional distribution function with copula
    $C_{\bm{B}}$ and $j$th margin $\B(p_j)$, $j=1,\dots,d$.  In \R\
    realizations of $C_{\bm{B}}$ can be generated via the \texttt{copula} package, see
    \cite{hofertkojadinovicmaechleryan2018, copula}.  Applying Bernoulli
    quantile functions via \texttt{qbinom} then yields a realization of
    $\bm{I}$.  By adding up to $K-1$ $d$-dimensional Bernoulli vectors (possibly
    with different $\bm{p}$), one can obtain index vector distributions for
    $K>2$.
  \item For general $K$, another way to specify a distribution for $\bm{I}$ is via
    $\bm{I}=\bm{1}+\bm{B}$, where $\bm{B}\sim\M_d(K-1,\bm{q})$ is multinomial (a
    $d$-dimensional random vector with components in $\{0,\dots,K-1\}$ summing to
    $K-1$) for a $d$-dimensional vector of probabilities $\bm{q}$; the components
    of $\bm{I}$ sum to $K-1+d$ then. In \R, $n$ vectors $\bm{I}$ can be drawn via
    \texttt{1+t(rmultinom(n, size = K-1, prob = probs))}, where $\texttt{K}$ is $K$
    and \texttt{probs} is the $d$-dimensional vector of probabilities $\bm{q}$.
  \item Realizing that $\sum_{j=1}^dI_j\in\{d,d+1,\dots,dK\}$ one could alternatively
    first sample $\sum_{j=1}^dI_j$ from a discrete distribution with support $\{d,d+1,\dots,dK\}$
    and then draw uniformly from all $\bm{I}$ with this given sum.
  \item Yet another way to sample $\bm{I}$ is to draw from any
    $d$-dimensional copula and map the components to any marginal
    distributions with support $\{1,\dots,K\}$. This can be done with
    the \R\ package \texttt{copula}.
  \item One can also specify $p_{\bm{I}}(\bm{i})\ge 0$ for multiple
    $\bm{i}\in\{1,\dots,K\}^d$ such that $\sum_{\bm{i}} p_{\bm{I}}(\bm{i})=1$. Then draw
    $\bm{i}$'s according to their probabilities $p_{\bm{I}}(\bm{i})$. If all $\bm{i}$'s
    are stored as rows in a matrix \texttt{I}, then one can draw $n$ $\bm{I}$'s
    in \R\ via \texttt{I[sample(1:length(probs), size = n, replace = TRUE, prob
      = probs),]}, where \texttt{probs} is a vector of probabilities of length
    equal to the number of rows of \texttt{I} (so of the number of different
    $\bm{i}$ considered in this construction).
  \end{enumerate}
\end{remark}

\subsection{Comonotone index vectors}
In case of comonotone index vectors, the following result shows that
index-mixed copulas are convex combinations of their base copulas,
the construction thus generalizes traditional discrete copula mixtures.
For an illustration, see Example~\ref{ex:ill:index:mixing:2d} later.
\begin{corollary}[Comonotone index vectors]\label{cor:invar:equal:comon}
  If $p_{\bm{I}}(k,\dots,k)\ge 0$, $\sum_{k=1}^Kp_{\bm{I}}(k,\dots,k)=1$, then
  the index-mixed copula is the convex combination
  $C(\bm{u})=\sum_{k=1}^Kp_{\bm{I}}(k,\dots,k)C_k(\bm{u})$, $\bm{u}\in[0,1]^d$.
  In particular, if $C_1=\dots=C_K=:C_0$, then $C(\bm{u})=C_0(\bm{u})$,
  $\bm{u}\in[0,1]^d$, so any (single) copula $C_0$ is invariant under comonotone
  index-mixing with itself as base copula(s).
\end{corollary}
By Corollary~\ref{cor:invar:equal:comon}, index-mixed copulas can reach the full
range of concordance. For example, for $K=2$, we can choose the independence
copula and the comonotone copula as base copulas.  What allows us then to attain
the whole range of concordance is a suitable choice of $p_{\bm{I}}(1,\dots,1)$ (or
$p_{\bm{I}}(2,\dots,2)=1-p_{\bm{I}}(1,\dots,1)$). In contrast, EFGM copulas (see Section~\ref{sec:EFGM}) are
restricted to the symmetric case $p_{\bm{I}}(1,\dots,1)=p_{\bm{I}}(2,\dots,2)=1/2$ and thus, as we explain through
the lens of index-mixing in Section~\ref{sec:EFGM:limited:dep}, EFGM copulas
only attain a limited range of concordance.

As will become clear from the arguments given in the proof of
Lemma~\ref{lem:biv:mar} below, also the converse statement of
Corollary~\ref{cor:invar:equal:comon} is true. If the index vector $\bm{I}$ is
not comonotone, then the index-mixed copula $C$ is not a convex combination of
the base copulas $C_1,\dots,C_K$ alone anymore, but a convex combination that
additionally may involve (depending on the distribution of $\bm{I}$) all margins
(including one-dimensional ones) of the base copulas $C_1,\dots,C_K$.
For $d=2$, $C$ can still be expressed as a classical convex
combination, namely that of $C_1,\dots,C_K$ and
$\Pi$ as we have
seen before in \eqref{eq:Copula:2D}.  %

\subsection{Bivariate and trivariate margins}\label{sec:margins:2d:3d}
It is immediate from~\eqref{eq:index:mixed:cop} that margins of index-mixed
copulas are again index-mixed copulas, where the associated base copulas are the
respective margins of the original base copulas.  The following result provides
the explicit functional form in the case of bivariate and trivariate margins.
In principle (depending on the index distribution), higher-dimensional margins
are mixtures of the base copulas and all their lower-dimensional margins,
including those of dimension $1$. However, the combinatorial nature of the problem leads
to rather intractable formulas for dimensions higher than three. If necessary, these margins can be evaluated via
\eqref{eq:index:mixed:cop}, which is just a finite weighted sum.
\begin{lemma}[Bivariate and trivariate margins]\label{lem:biv:mar}
  Consider $\bm{U}\sim C$ for a $d$-dimensional index-mixed copula.
We then have the following results concerning the bivariate and trivariate margins of $C$:
  \begin{enumerate}
  \item\label{lem:biv:mar:1} For $1\le j_1<j_2\le d$, let $p_{(I_{j_1},I_{j_2})}(k_1,k_2)=\P(I_{j_1}=k_1, I_{j_2}=k_2)$, $k_1,k_2\in\{1,\dots,K\}$.
    Then $(U_{j_1},U_{j_2})$ follows the $(j_1,j_2)$-margin $C^{j_1,j_2}$ of $C$ given by
    \begin{align*}
      C^{j_1,j_2}(u_{j_1},u_{j_2})=\sum_{k=1}^Kp_{(I_{j_1},I_{j_2})}(k,k)C_k^{j_1,j_2}(u_{j_1},u_{j_2})+\biggl(\,\sum_{\substack{k_1,k_2=1\\ k_1\neq k_2}}^Kp_{(I_{j_1},I_{j_2})}(k_1,k_2)\biggr)\Pi(u_{j_1},u_{j_2})
    \end{align*}
    for all $u_{j_1},u_{j_2}\in[0,1]$, where $C_k^{j_1,j_2}$ denotes the $(j_1,j_2)$-margin of $C_k$, $k=1,\dots,K$.
    Note that the second sum in the above expression for $C^{j_1,j_2}$ is simply $1-\sum_{k=1}^Kp_{(I_{j_1},I_{j_2})}(k,k)$ here. %
    Furthermore, for order $K=2$, we obtain the bivariate margins
    \begin{align*}
      C^{j_1,j_2}(u_{j_1},u_{j_2})&=p_{(I_{j_1},I_{j_2})}(1,1)C_1(u_{j_1},u_{j_2})+p_{(I_{j_1},I_{j_2})}(2,2)C_2(u_{j_1},u_{j_2})\\
      &\phantom{={}}+(p_{(I_{j_1},I_{j_2})}(1,2)+p_{(I_{j_1},I_{j_2})}(2,1))\Pi(u_{j_1},u_{j_2}).
    \end{align*}
  \item\label{lem:biv:mar:2} For $1\le j_1<j_2<j_3\le d$, let $p_{(I_{j_1},I_{j_2},I_{j_3})}(k_1,k_2,k_3)=\P(I_{j_1}=k_1, I_{j_2}=k_2, I_{j_3}=k_3)$, $k_1,k_2,k_3\in\{1,\dots,K\}$.
    Then $(U_{j_1},U_{j_2},U_{j_3})$ follows the $(j_1,j_2,j_3)$-margin $C^{j_1,j_2,j_3}$ of $C$ given by
    \begin{align*}
      C^{j_1,j_2,j_3}(u_{j_1},u_{j_2},u_{j_3})
      &=\sum_{k=1}^Kp_{(I_{j_1},I_{j_2},I_{j_3})}(k,k,k)C_k^{j_1,j_2,j_3}(u_{j_1},u_{j_2},u_{j_3})\\
      &\phantom{={}}+\sum_{k=1}^K\biggl\{\biggl(\,\sum_{\substack{k_3=1\\k_3\neq k}}^Kp_{(I_{j_1},I_{j_2},I_{j_3})}(k,k,k_3)\biggr)C_k^{j_1,j_2}(u_{j_1},u_{j_2})u_{j_3}\\
      &\phantom{={}+\sum_{k=1}^K\biggl\{}+\biggl(\,\sum_{\substack{k_2=1\\k_2\neq k}}^Kp_{(I_{j_1},I_{j_2},I_{j_3})}(k,k_2,k)\biggr) C_k^{j_1,j_3}(u_{j_1},u_{j_3})u_{j_2}\\
      &\phantom{={}+\sum_{k=1}^K\biggl\{}+ \biggl(\,\sum_{\substack{k_1=1\\k_1\neq k}}^Kp_{(I_{j_1},I_{j_2},I_{j_3})}(k_1,k,k)\biggr) C_k^{j_2,j_3}(u_{j_2},u_{j_3})u_{j_1}\biggr\}\\
      &\phantom{={}}+\biggl(\,\sum_{\substack{k_1,k_2,k_3=1\\\text{all $k_l$ distinct}}}^K p_{(I_{j_1},I_{j_2},I_{j_3})}(k_1,k_2,k_3)\biggr)\Pi(u_{j_1},u_{j_2},u_{j_3})
    \end{align*}
    for all $u_{j_1},u_{j_2},u_{j_3}\in[0,1]$, where $C_k^{j_1,j_2,j_3}$ denotes the $(j_1,j_2,j_3)$-margin of $C_k$, $k=1,\dots,K$.
    Furthermore, for order $K=2$, we obtain the trivariate margins
    \begin{align*}
      C^{j_1,j_2,j_3}(u_{j_1},u_{j_2},u_{j_3})
      &=p_{(I_{j_1},I_{j_2},I_{j_3})}(1,1,1)C_1(u_{j_1},u_{j_2},u_{j_3})+p_{(I_{j_1},I_{j_2},I_{j_3})}(2,2,2)C_2(u_{j_1},u_{j_2},u_{j_3})\\
      &\phantom{{}=}+p_{(I_{j_1},I_{j_2},I_{j_3})}(1,1,2)C_1^{j_1,j_2}(u_{j_1},u_{j_2})u_{j_3}
        +p_{(I_{j_1},I_{j_2},I_{j_3})}(1,2,1) C_1^{j_1,j_3}(u_{j_1},u_{j_3})u_{j_2}\\
        &\phantom{{}=}+p_{(I_{j_1},I_{j_2},I_{j_3})}(2,1,1) C_1^{j_2,j_3}(u_{j_2},u_{j_3})u_{j_1}\\
      &\phantom{{}=}+p_{(I_{j_1},I_{j_2},I_{j_3})}(2,2,1) C_2^{j_1,j_2}(u_{j_1},u_{j_2})u_{j_3}
        +p_{(I_{j_1},I_{j_2},I_{j_3})}(2,1,2) C_2^{j_1,j_3}(u_{j_1},u_{j_3})u_{j_2}\\
        &\phantom{{}=}+p_{(I_{j_1},I_{j_2},I_{j_3})}(1,2,2) C_2^{j_2,j_3}(u_{j_2},u_{j_3})u_{j_1}.
    \end{align*}
  \end{enumerate}
\end{lemma}
Note that for Corollary~\ref{cor:index:mixed:form:c} we introduced the notation
$C_k^{I_{\cdot,k}^{\text{mat}}}$ for the marginal copula of $C_k$ corresponding to all
entries in $I_{\cdot,k}^{\text{mat}}$ that are equal to $1$, while in the current
Section~\ref{sec:margins:2d:3d} we wrote the indices ($j_1,j_2,j_3$) of these
entries as exponents.  Both notations can be convenient to work with, as much as
it is convenient to have both the notion of the index vector $\bm{I}$ and its
implied index matrix $I^{\text{mat}}$ available.

\subsection{Mixtures}
The following result considers the case where the base copulas are themselves mixtures.
\begin{corollary}[Closure with respect to mixtures of the base copulas]\label{cor:index:mixed:mixtures}
  If, for $k=1,\dots,K$, the base copula $C_k$ is a mixture $C_k(\bm{u})=\int_{\IR}C_k(\bm{u};\theta_k)\,\rd G_k(\theta_k)$,
  where $C_k(\cdot;\theta_k)$ is a copula for all $\theta_k$ in the domain of the distribution function $G_k$,
  then
  \begin{align*}
    C(\bm{u})=\int_{\IR}\dots\int_{\IR} \E_{\bm{I}}\biggl[\,\prod_{k=1}^KC_k(\bm{u}^{I_{\cdot,k}^{\text{mat}}};\theta_k)\biggr]\,\rd G_1(\theta_1)\cdots\,\rd G_K(\theta_K),\quad\bm{u}\in[0,1]^d,
  \end{align*}
  that is the index-mixed copula is a mixture of index-mixed copulas. In
  particular, if, for $k=1,\dots,K$, $G_k(\theta_k)=\sum_{l_k=1}^{m_k}q_k(l_k)\I_{[\theta_{k,l_k},\infty)}(\theta_k)$ with
  $q_k(l_k)\ge 0$, $\sum_{l_k=1}^{m_k}q_k(l_k)=1$, and $-\infty<\theta_{k,1}<\dots<\theta_{k,m_k}<\infty$
  (the case of discrete mixtures), then we have
  \begin{align*}
    C(\bm{u})=\sum_{l_K=1}^{m_K}\dots\sum_{l_1=1}^{m_1}q_1(l_1)\cdot\ldots\cdot q_K(l_K)\,\E_{\bm{I}}\biggl[\,\prod_{k=1}^KC_k(\bm{u}^{I_{\cdot,k}^{\text{mat}}};\theta_{k,l_k})\biggr].
  \end{align*}
\end{corollary}
Corollary~\ref{cor:index:mixed:mixtures} remains valid if $C_k$ is an
$r_k$-dimensional mixture for $r_k>1$, $k=1,\dots,K$, so when $\theta_k$ is
replaced by an $r_k$-dimensional vector $\bm{\theta}_k$ and $G_k$ is now an
$r_k$-dimensional distribution function.

One can also consider the case where index-mixed copulas are mixed in the
ordinary sense that $\int_{\IR}C(\bm{u};\theta)\,\rd G(\theta)$ for
$C(\bm{u};\theta)$ being an index-mixed copula with index vector $\bm{I}$ for
every $\theta$ (all sharing the same index distribution) and $\theta$ being a
parameter of every base copula. In this case
$\int_{\IR}C(\bm{u};\theta)\,\rd
G(\theta)=\int_{\IR}\E_{\bm{I}}\bigl[\,\prod_{k=1}^KC_k(\bm{u}^{I_{\cdot,k}^{\text{mat}}};\theta)\bigr]\,\rd
G(\theta)=\E_{\bm{I}}\bigl[\,\int_{\IR}\prod_{k=1}^KC_k(\bm{u}^{I_{\cdot,k}^{\text{mat}}};\theta)\,\rd
G(\theta)\bigr]$ is a mixture with respect to the index distribution
where each component is a mixture of products of copulas (with disjoint argument
indices).

\subsection{Symmetries}
The following result shows that the survival copulas of index-mixed copulas are
index-mixed copulas of the survival copulas of the underlying base copulas. As
a result, index-mixed copulas are radially symmetric if the underlying base
copulas are.
\begin{proposition}[Survival copula and radial symmetry]\label{prop:surv:cop}
  With the setup of Definition~\ref{def:index:mixed:copulas}, the survival copula of the index-mixed copula $C$ is
  \begin{align*}
    \hat{C}(\bm{u})=\E_{\bm{I}}\biggl[\,\prod_{k=1}^K\hat{C}_k(\bm{u}^{I_{\cdot,k}^{\text{mat}}})\biggr],\quad\bm{u}\in[0,1]^d,
  \end{align*}
  where $\hat{C}_k$ is the survival copula of the base copula $C_k$, $k=1,\dots,K$.
  In particular, if all base copulas are radially symmetric ($\hat{C}_k=C_k$, $k=1,\dots,K$),
  so is the corresponding index-mixed copula.
\end{proposition}
Note that proving Proposition~\ref{prop:surv:cop} analytically via the Poincar{\'e}--Sylvester sieve formula is much harder.
In contrast, our proof of Proposition~\ref{prop:surv:cop} based on the
stochastic representation of index-mixed copulas is straightforward to follow.

Another form of symmetry is exchangeability.  A copula is \emph{exchangeable} if
$C(u_{\sigma(1)},\dots,$ $u_{\sigma(d)})=C(u_1,\dots,u_d)$ for all permutations
$\sigma\in\Sigma_d := \{\sigma \colon \{1,\ldots,d\}\to\{1,\ldots,d\} : \sigma
\text{ is bijective}\}$; see \cite[Definition~1.7.12]{durantesempi2015}. If $\bm{I}$ is comonotone and thus, by
Corollary~\ref{cor:invar:equal:comon}, the index-mixed copula is a mixture of
the base copulas, it is immediate that if all base copulas are exchangeable, so
is the corresponding index-mixed copula.
However, for index-mixed copulas without comonotone index vector, lower-dimensional margins of the base copulas appear in the index-mixed copula construction.
Following permutation, these margins may now be evaluated at arguments that they previously did not depend on.
As such, their values may differ, irrespectively of whether the base copulas are exchangeable or not.
Example~\ref{ex:ill:index:mixing:4d} in Section~\ref{sec:index:mixing} provides a counterexample in
this sense: even though both base copulas are exchangeable, the resulting
index-mixed copula is not (as pairwise margins differ; see
Figure~\ref{fig:IMC:4d}). Exchangeability is lost because of the index
distribution in this case.

For deriving sufficient conditions for the index distribution and base copulas
to result in an exchangeable index-mixed copula, we need to introduce some
notation and an auxiliary result. For a vector $\bm{x} = (x_1,\ldots,x_d)$, we
define $\sigma(\bm{x}):=(x_{\sigma(1)},\ldots,x_{\sigma(d)})$ for
$\sigma\in\Sigma_d$. Building on the set of all possible values
$\mathcal{K}^d:=\{1,\ldots,K\}^d$ of index vectors, we also consider the set
$\mathcal{K}_{\text{ord}}^d:= \{\sort{\bm{i}}:\bm{i}\in\mathcal{K}^d\}$ consisting of
all $\binom{K+d-1}{d}$ possible (non-decreasingly) \emph{ordered} vectors
$\sort{\bm{i}}\in\mathcal{K}^d$. Lastly, for
$\bm{k}_{\text{ord}}\in \mathcal{K}_{\text{ord}}^d$ we denote by
$\calK_{\bm{k}_{\text{ord}}}^d := \{\bm{i}\in\mathcal{K}^d :
\sort{\bm{i}} = \bm{k}_{\text{ord}}\}$ all vectors $\bm{i}\in\mathcal{K}^d$ such that
$\sort{\bm{i}}= \bm{k}_{\text{ord}}$. The problem when determining exchangeability of
an index-mixed copula $C$ is that for a given $\sigma\in\Sigma_d$ and
$\bm{i}\in\calK_{\bm{k}_{\text{ord}}}^d$, it is not necessarily
true that $\prod_{k=1}^K C_k\bigl(\sigma(\bm{u})^{I_{\cdot, k}^{\text{mat}}(\bm{i})}\bigr) =
\prod_{k=1}^K C_k\bigl(\bm{u}^{I_{\cdot, k}^{\text{mat}}(\bm{i})}\bigr)$,
where, for $\bm{i}=(i_1,\ldots,i_d)$, $I^{\text{mat}}(\bm{i}):=(e_{i_1}^{K},\dots,e_{i_d}^{K})\T$, that is the $j$th row of
$I^{\text{mat}}(\bm{i})$ is $e^K_{i_j}$.
However, the equality holds for the sum of all terms in a given partition element
$\calK_{\bm{k}_{\text{ord}}}^d$. This leads to the following
auxiliary result.
\begin{lemma}[Exchangeability of specific sums of products of base copulas]\label{lem:exch:sums:prods:base:copulas}
Denote by $C_k$, $k=1,\dots,K$, a collection of exchangeable $d$-copulas.
  For $\sigma\in\Sigma_d$ and $\bm{k}_{\text{ord}}\in \mathcal{K}_{\text{ord}}^d$ we have
  \begin{align*}
    \sum_{\bm{i}\in \calK_{\bm{k}_{\text{ord}}}^d} \prod_{k=1}^K C_k\bigl(\sigma(\bm{u})^{I_{\cdot, k}^{\text{mat}}(\bm{i})}\bigr) &=
                                                                                                                                         \sum_{\bm{i}\in \calK_{\bm{k}_{\text{ord}}}^d} \prod_{k=1}^K C_k\bigl(\bm{u}^{I_{\cdot, k}^{\text{mat}}(\bm{i})}\bigr).
  \end{align*}
\end{lemma}
We are now ready to derive sufficient conditions for index-mixed copulas to be exchangeable.
\begin{proposition}[Sufficient conditions for exchangeability]\label{prop:exchangeable}
  Let $d\in\IN$, $d\geq 2$. Denote by $C_k$, $k=1,\dots,K$, $d$-dimensional base
  copulas and by $\bm{I}$ an index vector distributed on $\mathcal{K}^d$,
  $K\in\IN$.
  \begin{enumerate}
  \item\label{prop:exchangeable:1} If every base copula $C_k$ is exchangeable,
    then
    \begin{align*}
      C(\bm{u})&=\E_{\bm{I}}\biggl[\,\prod_{k=1}^KC_k(\bm{u}_{J_k},1,\dots,1)\biggr],
    \end{align*}
    where $J_k$ is the $k$th element of the index partition (thus depending on
    $\bm{I}$) and $\bm{u}_{J_k}$ is the $D_k$-dimensional vector consisting of
    those elements of $\bm{u}$ with indices in $J_k$ (in the same order).
  \item\label{prop:exchangeable:2} If, additionally, $\bm{I}$ is uniformly
    distributed on $\calK_{\bm{k}_{\text{ord}}}^d$ for
    every $\bm{k}_{\text{ord}}\in \mathcal{K}_{\text{ord}}^d$, that is
    $p_{\bm{I}}(\bm{i})= p_{\bm{I},\bm{k}_{\text{ord}}}$ for all
    $\bm{i}\in \calK_{\bm{k}_{\text{ord}}}^d$, then $C$ is
    exchangeable.
  \end{enumerate}
\end{proposition}
Situations in which Condition~\ref{prop:exchangeable:2} in
Proposition~\ref{prop:exchangeable} holds are readily available.  The following
Example~\ref{ex:exchangeable:uniform} provides three constructions to this
effect.
\begin{example}[Exchangeable uniform distribution for index vectors]\label{ex:exchangeable:uniform}
  The following three examples outline constructions that guarantee that
  $p_{\bm{I}}(\bm{i}_1) = p_{\bm{I}}(\bm{i}_2)$ whenever
  $ \bm{i}_1, \bm{i}_2 \in \calK_{\bm{k}_{\text{ord}}}^d$ for some
  $\bm{k}_{\text{ord}}\in \mathcal{K}_{\text{ord}}^d$.  As such, these
  constructions guarantee that Condition~\ref{prop:exchangeable:2} in
  Proposition~\ref{prop:exchangeable} holds.
  \begin{enumerate}
  \item\label{ex:exchangeable:uniform:1} A uniform distribution of $\bm{I}$ on
    $\mathcal{K}^d$ directly leads to a uniform distribution of $\bm{I}$ on
    $\calK_{\bm{k}_{\text{ord}}}^d$ for every
    $\bm{k}_{\text{ord}}\in \mathcal{K}_{\text{ord}}^d$.  Indeed, when
    $p_{\bm{I}}(\bm{i}) = \frac{1}{K^d}$, we clearly have
    $p_{\bm{I}}(\bm{i}_1) = p_{\bm{I}}(\bm{i}_2)$ for all
    $\bm{i}_1, \bm{i}_2 \in \calK_{\bm{k}_{\text{ord}}}^d$.  The total
    probability attributed to $\calK_{\bm{k}_{\text{ord}}}^d$ is then consequently
    \begin{align*}
      \P(\bm{I} \in \calK_{\bm{k}_{\text{ord}}}^d) = \sum_{\bm{i}\in \calK_{\bm{k}_{\text{ord}}}^d}\!\!\!\! p_{\bm{I}}(\bm{i}) = \frac{1}{K^d}\sum_{\bm{i}\in \calK_{\bm{k}_{\text{ord}}}^d}\!\!\!\! 1 = \frac{\left|\calK_{\bm{k}_{\text{ord}}}^d\right|}{K^d}.
    \end{align*}
    With
    $\P(\sigma(\bm{I}) = \bm{k}_{\text{ord}}) = \P(\bm{I} \in
    \calK_{\bm{k}_{\text{ord}}}^d)$ this implies that $\sigma(\bm{I})$ is not
    uniformly distributed on $\mathcal{K}_{\text{ord}}^d$, even though $\bm{I}$ is
    uniformly distributed on $\mathcal{K}^d$.
  \item\label{ex:exchangeable:uniform:2} A second possibility is to attribute the probability mass of $\bm{I}$ on $\calK_{\bm{k}_{\text{ord}}}^d$ according to the size of $\calK_{\bm{k}_{\text{ord}}}^d$.
    When defining $\bm{I}$ such that
    \begin{align*}
      p_{\bm{I}}(\bm{i}) = \frac{1}{|\calK_{\bm{k}_{\text{ord}}}^d|\binom{K+d-1}{d}},\quad \bm{i}\in \calK_{\bm{k}_{\text{ord}}}^d,\ \bm{k}_{\text{ord}}\in \mathcal{K}_{\text{ord}}^d,
    \end{align*}
    we have again trivially that $p_{\bm{I}}(\bm{i}_1)=p_{\bm{I}}(\bm{i}_2)$ for all $ \bm{i}_1, \bm{i}_2 \in \calK_{\bm{k}_{\text{ord}}}^d$.
    However, different from Part~\ref{ex:exchangeable:uniform:1}, the probability for each partition is now given as
    \begin{align*}
      \P(\bm{I} \in \calK_{\bm{k}_{\text{ord}}}^d) = \sum_{\bm{i}\in \calK_{\bm{k}_{\text{ord}}}^d}\!\!\!\!p_{\bm{I}}(\bm{i}) = \frac{1}{|\calK_{\bm{k}_{\text{ord}}}^d|\binom{K+d-1}{d}} \sum_{\bm{i}\in \calK_{\bm{k}_{\text{ord}}}^d}\!\!\!\! 1 = \frac{1}{\binom{K+d-1}{d}}.
    \end{align*}
    The distribution of $\sigma(\bm{I})$ on $\mathcal{K}_{\text{ord}}^d$ is hence
    uniform, even though the distribution of $\bm{I}$ on $\mathcal{K}^d$ is clearly not
    (except in trivial cases). Assuming that each of the base copulas
    $C_k$, $k=1,\dots,K$, is exchangeable we then have from~\eqref{eq:exchangeability:reformulate} that
    \begin{align*}
      C(\bm{u}) &= \sum_{\bm{k}_{\text{ord}}\in\mathcal{K}_{\text{ord}}^d} \frac{1}{|\calK_{\bm{k}_{\text{ord}}}^d|\binom{K+d-1}{d}} \sum_{\bm{i}\in \calK_{\bm{k}_{\text{ord}}}^d} \prod_{k=1}^KC_k(\bm{u}^{I_{\cdot, k}^{\text{mat}}(\bm{i})})\\
                &= \frac{1}{\binom{K+d-1}{d}}\sum_{\bm{k}_{\text{ord}}\in\mathcal{K}_{\text{ord}}^d} \frac{1}{|\calK_{\bm{k}_{\text{ord}}}^d|} \sum_{\bm{i}\in \calK_{\bm{k}_{\text{ord}}}^d} \prod_{k=1}^KC_k(\bm{u}_{J_k(\bm{i})},1,\ldots,1).
    \end{align*}
    Due to the uniform distribution of $\bm{I}$ on each
    $\calK_{\bm{k}_{\text{ord}}}^d$, we obtain that $C$ is
    exchangeable.

    Note that other distributions on $\mathcal{K}^d$, even those not uniform
    on the partitions $\calK_{\bm{k}_{\text{ord}}}^d$ and thus not
    covered by Proposition~\ref{prop:exchangeable}, can also induce a uniform
    distribution on $\mathcal{K}_{\text{ord}}^d$. It is therefore not sufficient
    to only consider the induced distribution on $\mathcal{K}_{\text{ord}}^d$, but
    instead the underlying distribution on $\mathcal{K}^d$ needs to be considered.
  \item Lastly, exchangeability of $\bm{I}$ is also sufficient:
    For $\bm{i}_1,\bm{i}_2 \in \calK_{\bm{k}_{\text{ord}}}^d$ there is a permutation $\sigma\in\Sigma$ such that $\bm{i}_2 = \sigma(\bm{i}_1)$.
    In case $\bm{I}$ is exchangeable this leads to
    \begin{align*}
      p_{\bm{I}}(\bm{i}_1)=\P(\bm{I}=\bm{i}_1)
      = \P(\sigma(\bm{I})=\sigma(\bm{i}_1))
        = \P(\bm{I}=\sigma(\bm{i}_1)) = \P(\bm{I}=\bm{i}_2)=p_{\bm{I}}(\bm{i}_2).
    \end{align*}
  \end{enumerate}
\end{example}

While the conditions in Proposition~\ref{prop:exchangeable} are sufficient, they
are not necessary.  As shown in the following example, even for non-exchangeable
base copulas, the resulting index-mixed copula can be exchangeable, via
a suitable index distribution.
\begin{example}[Non-exchangeable base copulas, exchangeable index-mixed copula]
  \begin{enumerate}
  \item For $k=1,2=:K$, let $C_k$ be a bivariate,
    non-exchangeable copula (for example, a Marshall--Olkin copula with unequal
    parameters) and $p_{\bm{I}}(1,2)=p_{\bm{I}}(2,1)=1/2$. Then the resulting
    index-mixed copula $C$ is the bivariate independence copula $\Pi$, which is exchangeable.
  \item For $k=1,2=:K$, set
    $C_1(u_1,u_2,u_3,u_4)=C_{0,1}(u_1,u_2)C_{0,2}(u_3,u_4)$ and
    $C_2(u_1,u_2,u_3,$ $u_4)=C_{0,2}(u_1,u_2) \cdot C_{0,1}(u_3,u_4)$ for two different
    bivariate copulas $C_{0,1},C_{0,2}$ not equal to the independence copula.
    Furthermore, let $p_{\bm{I}}(1,2,1,2)=p_{\bm{I}}(2,1,2,1)=1/2$. Then the
    resulting index-mixed copula $C$ is the four-dimensional independence copula $\Pi$, which is
    exchangeable. Similarly, one could choose $C_{0,1} = \Pi$ and
    $p_{\bm{I}}(1,1,2,2)=1$. Lastly, it is straightforward to adapt the discussed constructions for dimensions $d\geq 5$.
  \end{enumerate}
\end{example}
We finish this section with an example that highlights the calculation of the
objects referenced in Lemma~\ref{lem:exch:sums:prods:base:copulas} and
Proposition~\ref{prop:exchangeable}.

\begin{example}[Calculations for exchangeable base copulas]
  To exemplify the objects used in Lemma~\ref{lem:exch:sums:prods:base:copulas}
  and Proposition~\ref{prop:exchangeable}, consider now $d=3$ and $K=2$.  In this
  case there are $K^d = 8$ possible index vectors leading to
  $\binom{K+d-1}{d} = 4$ unique sorted configurations:
  \begin{alignat*}{5}
    \bm{i}_1 &= (1,1,1) &\ \ \tmb{}{\Longrightarrow}{partition}\ \ & J_1 = \{1,2,3\},\ J_2 = \emptyset &\ \ \tmb{}{\Longrightarrow}{\text{sort}}\ \ & (1,1,1)\\[2mm]
    \bm{i}_2 &= (1,1,2) &\ \ \tmb{}{\Longrightarrow}{partition}\ \ & J_1 = \{1,2\},\ J_2 = \{3\} &\ \ \tmb{}{\Longrightarrow}{\text{sort}}\ \ & (1,1,2)\\
    \bm{i}_3 &= (1,2,1) &\ \ \tmb{}{\Longrightarrow}{partition}\ \ & J_1 = \{1,3\},\ J_2 = \{2\} &\ \ \tmb{}{\Longrightarrow}{\text{sort}}\ \ & (1,1,2)\\
    \bm{i}_4 &= (2,1,1) &\ \ \tmb{}{\Longrightarrow}{partition}\ \ & J_1 = \{2,3\},\ J_2 = \{1\} &\ \ \tmb{}{\Longrightarrow}{\text{sort}}\ \ & (1,1,2)\\[2mm]
    \bm{i}_5 &= (1,2,2) &\ \ \tmb{}{\Longrightarrow}{partition}\ \ & J_1 = \{1\},\ J_2 = \{2,3\} &\ \ \tmb{}{\Longrightarrow}{\text{sort}}\ \ & (1,2,2)\\
    \bm{i}_6 &= (2,1,2) &\ \ \tmb{}{\Longrightarrow}{partition}\ \ & J_1 = \{2\},\ J_2 = \{1,3\} &\ \ \tmb{}{\Longrightarrow}{\text{sort}}\ \ & (1,2,2)\\
    \bm{i}_7 &= (2,2,1) &\ \ \tmb{}{\Longrightarrow}{partition}\ \ & J_1 = \{3\},\ J_2 = \{1,2\} &\ \ \tmb{}{\Longrightarrow}{\text{sort}}\ \ & (1,2,2)\\[2mm]
    \bm{i}_8 &= (2,2,2) &\ \ \tmb{}{\Longrightarrow}{partition}\ \ & J_1 = \emptyset,\ J_2 = \{1,2,3\} &\ \ \tmb{}{\Longrightarrow}{\text{sort}}\ \ & (2,2,2)
  \end{alignat*}
  so that $\mathcal{K}_{\text{ord}}^d = \{(1,1,1),(1,1,2),(1,2,2),(2,2,2)\}$.
  For the inverse images of elements in $\mathcal{K}_{\text{ord}}^d$ we
  have
\begin{align*}
\calK_{(1,1,1)_{\text{ord}}}^d &= \{(1,1,1)\},\\
\calK_{(1,1,2)_{\text{ord}}}^d &= \{(1,1,2),(1,2,1),(2,1,1)\},\\
\calK_{(1,2,2)_{\text{ord}}}^d &= \{(1,2,2),(2,1,2),(2,2,1)\},\\
\calK_{(2,2,2)_{\text{ord}}}^d &= \{(2,2,2)\}.
\end{align*}
  This specifically shows that
  $\calK_{(1,1,1)_{\text{ord}}}^d$, $\calK_{(1,1,2)_{\text{ord}}}^d$,
  $\calK_{(1,2,2)_{\text{ord}}}^d$ and $\calK_{(2,2,2)_{\text{ord}}}^d$ provide
  a partition of $\calK^d$.

  For a random vector $\bm{I}$ on $\calK^d$, the probabilities
  $p_{\bm{I}}(\bm{i}_k)=\P(\bm{I} = \bm{i}_k)$, $k=1,\dots,K^d=8$, now induce a
  probability distribution for the sorted index vector $\sort{\bm{I}}$ on
  $\mathcal{K}_{\text{ord}}^d$. We specifically have
  \begin{align*}
    &\P(\sort{\bm{I}} = (1,1,1)) = p_{\bm{I}}(\bm{i}_1),\\
    &\P(\sort{\bm{I}} = (1,1,2)) = p_{\bm{I}}(\bm{i}_2) + p_{\bm{I}}(\bm{i}_3) + p_{\bm{I}}(\bm{i}_4),\\
    &\P(\sort{\bm{I}} = (1,2,2)) = p_{\bm{I}}(\bm{i}_5) + p_{\bm{I}}(\bm{i}_6) + p_{\bm{I}}(\bm{i}_7),\\
    &\P(\sort{\bm{I}} = (2,2,2)) = p_{\bm{I}}(\bm{i}_8).
  \end{align*}
  Noting that the probabilities of $\bm{I}$ enter only as a sum, it is now easy
  to see that different distributions on $\calK^d$ can induce the same
  distribution on $\mathcal{K}_{\text{ord}}^d$ as already mentioned in
  Example~\ref{ex:exchangeable:uniform}.

  For exchangeable base copulas $C_1$ and $C_2$ we now have for $\bm{u}=(u_1,u_2,u_3)$ that
  \begin{align*}
    &\phantom{{}={}}C(\bm{u})\\
    &= \E_{\bm{I}}\biggl[\,\prod_{k=1}^KC_k(\bm{u}_{J_k},1,\dots,1)\biggr] = \sum_{j=1}^8 p_{\bm{I}}(\bm{i}_j) \prod_{k=1}^2C_k((u_1,u_2,u_3)_{J_k(\bm{i}_j)},1,\dots,1)\\
    &=\phantom{{}+\,}p_{\bm{I}}(\bm{i}_1) C_1(u_1,u_2,u_3)\\
    &\quad+ \bigl( p_{\bm{I}}(\bm{i}_2) C_1(u_1,u_2,1)C_2(u_3,1,1) + p_{\bm{I}}(\bm{i}_3) C_1(u_1,u_3,1)C_2(u_2,1,1) + p_{\bm{I}}(\bm{i}_4) C_1(u_2,u_3,1)C_2(u_1,1,1)\bigr)\\
    &\quad+ \bigl( p_{\bm{I}}(\bm{i}_5) C_1(u_1,1,1)C_2(u_2,u_3,1) + p_{\bm{I}}(\bm{i}_6) C_1(u_2,1,1)C_2(u_1,u_3,1) + p_{\bm{I}}(\bm{i}_7) C_1(u_3,1,1)C_2(u_1,u_2,1)\bigr)\\
    &\quad+ p_{\bm{I}}(\bm{i}_8) C_2(u_1,u_2,u_3).
  \end{align*}
  Without further assumptions, no additional simplifications are possible, and
  the copula $C$ is in general \emph{not} exchangeable.
  For example, in order to obtain
  $C(u_1,u_2,u_3) = C(u_2,u_1,u_3)$, one would need to impose
  $p_{\bm{I}}(\bm{i}_3) = p_{\bm{I}}(\bm{i}_4)$ and $p_{\bm{I}}(\bm{i}_5)=p_{\bm{I}}(\bm{i}_6)$ which are
  additional constraints that do not follow from assumptions about the distribution
  on $\mathcal{K}_{\text{ord}}^d$ (such as uniformity for example).

  However, if $\bm{I}$ is uniformly distributed on each
  $\calK^d_{\bm{k}_{\text{ord}}}$ for
  $\bm{k}_{\text{ord}}\in\mathcal{K}_{\text{ord}}^d$, further simplifications
  are possible.  If we have $p_1:=p_{\bm{I}}(\bm{i}_1)$,
  $p_2:=p_{\bm{I}}(\bm{i}_2) = p_{\bm{I}}(\bm{i}_3) = p_{\bm{I}}(\bm{i}_4)$,
  $p_3:=p_{\bm{I}}(\bm{i}_5) = p_{\bm{I}}(\bm{i}_6) = p_{\bm{I}}(\bm{i}_7)$ and $p_4:=p_{\bm{I}}(\bm{i}_8)$
  then
  \begin{align*}
   C(\bm{u}) &=\phantom{{}+\,} p_1 C_1(u_1,u_2,u_3)\\
    &\quad+ p_2\bigl(C_1(u_1,u_2,1)C_2(u_3,1,1) + C_1(u_1,u_3,1)C_2(u_2,1,1) + C_1(u_2,u_3,1)C_2(u_1,1,1)\bigr)\\
    &\quad+ p_3 \bigl(C_1(u_1,1,1)C_2(u_2,u_3,1) + C_1(u_2,1,1)C_2(u_1,u_3,1) + C_1(u_3,1,1)C_2(u_1,u_2,1)\bigr)\\
    &\quad+ p_4 C_2(u_1,u_2,u_3).
  \end{align*}
  The copula $C$ is now exchangeable since the probability associated to each
  group was factored out from the sum of marginal copulas.  Inside each group
  the exchangeability of the base copulas guarantees the exchangeability of $C$.
  Notably, it is not necessary to have
  $p_1=p_2=p_3=p_4$, but instead only the probability within the underlying
  groupings needs to be uniform.  There is, however, a choice of these
  probabilities such that the resulting distribution on
  $\mathcal{K}_{\text{ord}}^d$ is again uniform; see
  Example~\ref{ex:exchangeable:uniform} Part~\ref{ex:exchangeable:uniform:2}.
\end{example}

\subsection{Tail dependence}\label{sec_tail_dependence}
We now consider (matrices of pairwise) tail dependence coefficients for index-mixed copulas.
\begin{proposition}[Tail dependence]\label{prop:tail:dependence}
  Consider the setup and notation of Lemma~\ref{lem:biv:mar} Part~\ref{lem:biv:mar:1}. Suppose the
  lower and upper tail dependence coefficients of the $(j_1,j_2)$-margin
  $C^{j_1,j_2}_k$ of the base copula $C_k$ exist and are denoted by
  $\lambda_{\text{l},k}^{j_1,j_2},\lambda_{\text{u},k}^{j_1,j_2}$, $k=1,\dots,K$, respectively. Then the lower and
  upper tail dependence coefficients of the $(j_1,j_2)$-margin $C^{j_1,j_2}$ of
  the index-mixed copula $C$ are respectively given by
  \begin{align*}
    \lambda_{\text{l}}^{j_1,j_2}=\sum_{k=1}^K p_{(I_{j_1},I_{j_2})}(k,k)\lambda_{\text{l},k}^{j_1,j_2}\quad\text{and}\quad\lambda_{\text{u}}^{j_1,j_2}=\sum_{k=1}^K p_{(I_{j_1},I_{j_2})}(k,k)\lambda_{\text{u},k}^{j_1,j_2}.
  \end{align*}
\end{proposition}
In particular, index-mixed copulas can exhibit lower and/or upper tail
dependence; compare with EFGM copulas in Section~\ref{sec:EFGM:limited:dep}
which are confined to be tail independent. In
Example~\ref{ex:ill:index:mixing:4d}, the matrix of pairwise upper tail
dependence coefficients is
\begin{align*}
  \begin{pmatrix}
    1 & \frac{\lambda_{\text{u}}^{\text{G}}}{2} + \frac{0}{3} + \frac{0}{6} & \frac{\lambda_{\text{u}}^{\text{G}}}{3} + \frac{0}{3} & 0 \\
    \frac{\lambda_{\text{u}}^{\text{G}}}{2} + \frac{0}{3} + \frac{0}{6} & 1 & 0 & \frac{0}{3} + \frac{0}{3} \\
    \frac{\lambda_{\text{u}}^{\text{G}}}{3} + \frac{0}{3} & 0 & 1 & \frac{0}{2} + \frac{0}{3} + \frac{\lambda_{\text{u}}^{\text{G}}}{6}\\
    0 & \frac{0}{3} + \frac{0}{3} & \frac{0}{2} + \frac{0}{3} + \frac{\lambda_{\text{u}}^{\text{G}}}{6} & 1
  \end{pmatrix}
  =\begin{pmatrix}
    1 & \frac{2-2^{1/2}}{2} & \frac{2-2^{1/2}}{3} & 0 \\
    \frac{2-2^{1/2}}{2} & 1 & 0 & 0\\
    \frac{2-2^{1/2}}{3} & 0 & 1 & \frac{2-2^{1/2}}{6}\\
    0 & 0 & \frac{2-2^{1/2}}{6} & 1
  \end{pmatrix},
\end{align*}
where we note that the upper tail dependence coefficient for the appearing
Gaussian copula $C^{\text{Ga}}$ is always zero.  For the appearing Gumbel copula $C^{\text{G}}$, Kendall's
tau is given by $\tau=\frac{\theta-1}{\theta}=1/2$, so $\theta=2$ here, and the
upper tail dependence coefficient $\lambda_{\text{u}}^{\text{G}}$ of
$C^{\text{G}}$ is $\lambda_{\text{u}}^{\text{G}}=2-2^{1/\theta}$.  Matrices of
pairwise Spearman's rho and Kendall's tau can be obtained numerically; see
Section~\ref{section_concordance_measures}.

\begin{remark}[On bounds of pairwise tail dependence]\label{remark_bounds_tail_dependence}
  With $p_{(I_{j_1},I_{j_2})}^K:=\sum_{k=1}^K p_{(I_{j_1},I_{j_2})}(k,k)$, Proposition~\ref{prop:tail:dependence} implies the
  lower and upper bounds
  \begin{alignat*}{5}
    0&\leq p_{(I_{j_1},I_{j_2})}^K \min_{1\leq k\leq K}\lambda_{\text{l},k}^{j_1,j_2}&\leq \lambda_{\text{l}}^{j_1,j_2}&\leq p_{(I_{j_1},I_{j_2})}^K \max_{1\leq k\leq K}\lambda_{\text{l},k}^{j_1,j_2}&\leq p_{(I_{j_1},I_{j_2})}^K &\leq 1,\\
    0&\leq p_{(I_{j_1},I_{j_2})}^K \min_{1\leq k\leq K}\lambda_{\text{u},k}^{j_1,j_2}&\leq \lambda_{\text{u}}^{j_1,j_2}&\leq p_{(I_{j_1},I_{j_2})}^K \max_{1\leq k\leq K}\lambda_{\text{u},k}^{j_1,j_2}&\leq p_{(I_{j_1},I_{j_2})}^K &\leq 1
  \end{alignat*}
  for $\lambda_{\text{l}}^{j_1,j_2}$ and $\lambda_{\text{u}}^{j_1,j_2}$.
  Specifically, $\lambda_{\text{l}}^{j_1,j_2} = 1$ or
  $\lambda_{\text{u}}^{j_1,j_2} = 1$ imply that $p_{(I_{j_1},I_{j_2})}^K=1$, so only the
  probabilities $p_{(I_{j_1},I_{j_2})}(k,k)=\P(I_{j_1}=k,I_{j_2}=k)$, $k=1,\dots,K$, can
  be non-zero, which means that $(I_{j_1},I_{j_2})$ must be comonotone.  More
  generally, if a pair $(U_{j_1},U_{j_2})$ satisfies
  $\lambda_{\text{l}}^{j_1,j_2}\approx 1$ or
  $\lambda_{\text{u}}^{j_1,j_2}\approx 1$, then $p_{(I_{j_1},I_{j_2})}^K\approx 1$, which
  limits the mass on the remaining pairs $(U_{j_1'},U_{j_2'})$ for
  $j_1',j_2'\in\{1,\dots,d\}:(j_1',j_2')\neq(j_1,j_2)$, and thus limits their ability
  to contribute to the construction (also in terms of tail dependence).
\end{remark}

\subsection{Measures of concordance}\label{section_concordance_measures}
\subsubsection{Matrices of pairwise measures of concordance}
By \cite[p.~182]{nelsen2006}, Blomqvist's beta for a bivariate copula $C$ is
given by $\beta = 4C(1/2,1/2)-1$. By Lemma~\ref{lem:biv:mar}, it is
straightforward to compute Blomqvist's beta for the margin $C^{j_1,j_2}$,
$1 \leq j_1 < j_2 \leq d$, of an index-mixed copula as
\begin{align*}
  \beta =\sum_{k=1}^Kp_{(I_{j_1},I_{j_2})}(k,k)\beta_k^{j_1,j_2},
\end{align*}
where $\beta_k^{j_1,j_2} = 4C^{j_1,j_2}_k(1/2,1/2)-1$ denotes Blomqvist's beta of
the $(j_1,j_2)$-margin of the base copula $C_k$, $1\leq k\leq K$.

Next, we consider Spearman's rho and Kendall's tau for bivariate margins of
index-mixed copulas. To this end, we also consider the $(j_1,j_2)$-margin
$C^{j_1,j_2}$ of an index-mixed copula $C$, given (according to
Lemma~\ref{lem:biv:mar} Part~\ref{lem:biv:mar:1}) by
\begin{align}
  C^{j_1,j_2}(u_{j_1},u_{j_2})&=\sum_{k=1}^Kp_{(I_{j_1},I_{j_2})}(k,k)C_k^{j_1,j_2}(u_{j_1},u_{j_2})+\biggl(1-\sum_{k=1}^Kp_{(I_{j_1},I_{j_2})}(k,k)\biggr)\Pi(u_{j_1},u_{j_2})\\
                              &=\sum_{k=1}^{K+1}p_{(I_{j_1},I_{j_2})}(k,k)C_k^{j_1,j_2}(u_{j_1},u_{j_2}),\label{eq:biv:mar:compact}
\end{align}
where $p_{(I_{j_1},I_{j_2})}(K+1,K+1):=1-\sum_{k=1}^Kp_{(I_{j_1},I_{j_2})}(k,k)$ and $C_{K+1}^{j_1,j_2}:=\Pi$.
By \cite[Chapter~5]{nelsen2006}, Spearman's rho is
\begin{align*}
  \rho_{\text{S}}^{j_1,j_2}=12\int_{[0,1]^2} C^{j_1,j_2}(u_{j_1},u_{j_2})\,\rd u_{j_1}\,\rd u_{j_2}-3=12\E[C^{j_1,j_2}(V_1,V_2)]-3
\end{align*}
for $V_1,V_2\isim\U(0,1)$ and Kendall's tau is
\begin{align*}
  \tau^{j_1,j_2}=4\int_{[0,1]^2} C^{j_1,j_2}(u_{j_1},u_{j_2})\,\rd C^{j_1,j_2}(u_{j_1},u_{j_2})-1=4\E[C^{j_1,j_2}(V_1,V_2)]-1
\end{align*}
for $(V_1,V_2)\sim C^{j_1,j_2}$.
To simplify the appearing expressions, let
\begin{align*}
  \mu_{C',C}=\int_{[0,1]^d} C'(\bm{u})\,\rd C(\bm{u})
\end{align*}
for two copulas $C,C'$ of equal dimension $d$. A notable feature of $\mu_{C',C}$
concerning mixtures is that if $C' = \sum^{K}_{k=1}\alpha_k C_k'$ for some
copulas $C_k'$ and weights $\alpha_k\geq 0$ with $\sum^{K}_{k=1} \alpha_k=1$, and $C = \sum^{M}_{l=1}\beta_l C_l$ for some
copulas $C_l$ and weights $\beta_l\geq 0$ with $\sum^{M}_{l=1} \beta_l=1$, then
\begin{align}\label{eq:mu:mixture}
  \mu_{C',C} = \int_{[0,1]^d} \sum^{K}_{k=1}\alpha_k C_k'(\bm{u})\,\rd \biggl(\,\sum^{M}_{j=1}\beta_j C_j(\bm{u})\biggr) = \sum^{K}_{k=1} \sum^{M}_{j=1} \alpha_k \beta_j \mu_{C_k',C_j}.
\end{align}
In the bivariate case we furthermore have
\begin{align}
  \mu_{C',C}=\mu_{C,C'},\label{eq:switch:bivariate:integrand:integrator}
\end{align}
which is implied by the analogous result for $Q(C,C')=4\mu_{C',C}-1$; see
\cite[Corollary~5.1.2]{nelsen2006}.  By \cite[Theorem~5.1.3,
Theorem~5.1.6]{nelsen2006}, $\tau^{j_1,j_2}=4\mu_{C^{j_1,j_2},C^{j_1,j_2}}-1$
and $\rho_{\text{S}}^{j_1,j_2}=12\mu_{C^{j_1,j_2},\Pi}-3$, which we will
use in the proof of the following result.
\begin{proposition}[Bivariate Spearman's rho, Kendall's tau]\label{prop_bivariate_Spearman_Kendall}
  Consider the setup and notation of Lemma~\ref{lem:biv:mar} Part~\ref{lem:biv:mar:1}
  and denote by $\rho_{\text{S},k}^{j_1,j_2}$ and $\tau_{k}^{j_1,j_2}$
  Spearman's rho and Kendall's tau of the base copula $C_k^{j_1,j_2}$,
  $k=1,\dots,K$, respectively. Furthermore, let
  $\tau_{k_1,k_2}^{j_1,j_2}=4\mu_{C_{k_1}^{j_1,j_2},C_{k_2}^{j_1,j_2}}-1$,
  $k_1,k_2\in\{1,\dots,K\}$.
  Then Spearman's rho and Kendall's tau of the $(j_1,j_2)$-margin $C^{j_1,j_2}$
  of the index-mixed copula $C$ are given by
  \begin{align*}
    \rho_{\text{S}}^{j_1,j_2}&=\sum_{k=1}^K p_{(I_{j_1},I_{j_2})}(k,k)\rho_{\text{S},k}^{j_1,j_2} \quad \text{and}\\
    \tau^{j_1,j_2}&=\sum_{k=1}^K(p_{(I_{j_1},I_{j_2})}(k,k))^2\tau_k^{j_1,j_2}+2\!\!\!\!\sum_{1\le k_1<k_2\le K}\!\!\!\! p_{(I_{j_1},I_{j_2})}(k_1,k_1)p_{(I_{j_1},I_{j_2})}(k_2,k_2)\tau_{k_1,k_2}^{j_1,j_2} \\
    &\phantom{={}}+\frac{2}{3}\biggl(1-\sum_{k=1}^Kp_{(I_{j_1},I_{j_2})}(k,k)\biggr)\sum_{k=1}^Kp_{(I_{j_1},I_{j_2})}(k,k)\rho_{\text{S},k}^{j_1,j_2}.
  \end{align*}
\end{proposition}
Similar to Remark~\ref{remark_bounds_tail_dependence} in case of tail
dependence, we can derive bounds on the possible values of Spearman's rho as
discussed next.  A similar effect can be expected for Kendall's tau, even though
the exact interplay between the components is harder to quantify.

\begin{remark}[On bounds of pairwise Spearman's rho]
  Index-mixing impacts $\rho_{S}^{j_1,j_2}$ relative to the Spearman's rho of
  the base copulas in two different ways.  At first glance, $\rho_{S}^{j_1,j_2}$
  is a weighted sum of the individual $\rho_{S,k}^{j_1,j_2}$,
  $k\in\{1,\dots,K\}$.  However, setting
  $q_{(I_{j_1},I_{j_2})} = 1 - \sum_{k=1}^{K}p_{(I_{j_1},I_{j_2})}(k,k) \geq 0$, it
  becomes apparent that additionally a shrinkage towards zero takes place if
  $q_{(I_{j_1},I_{j_2})} > 0$.  For example, in the special case when all base copulas
  have the same Spearman's rho $\rho_{S,k}^{j_1,j_2}=\rho_{S,\star}^{j_1,j_2}$
  for all $k$, we observe the shrinkage effect via
  $\rho_{S}^{j_1,j_2} = (1-q_{(I_{j_1},I_{j_2})} )\rho_{S,\star}^{j_1,j_2}$.  In the
  general case we set $\rho_{S,-}^{j_1,j_2} := \min_{k=1,\dots,K}\rho_{S,k}$ and
  $\rho_{S,+}^{j_1,j_2} := \max_{k=1,\dots,K}\rho_{S,k}$ to obtain the bounds
  $(1-q_{(I_{j_1},I_{j_2})})\rho_{S,-}^{j_1,j_2} \leq \rho_{S}^{j_1,j_2} \leq
  (1-q_{(I_{j_1},I_{j_2})})\rho_{S,+}^{j_1,j_2}$.  Note that these bounds do not violate
  the forthcoming Proposition~\ref{prop:K:ge:d}, in which case the upper and
  lower bound collapses to zero.

  In order to avoid shrinkage of the Spearman's rho bounds for all bivariate
  margins \emph{at the same time}, positive probability can only be attributed
  to index vectors of the form $\bm{I}=(k,k,\dots,k)\in\IR^d$,
  $k\in\{1,\dots,K\}$.  In this case we clearly have
  $\sum_{k=1}^{K}p_{(I_{j_1},I_{j_2})}(k,k) = \sum_{k=1}^{K}\P(I_{j_1}=k,
  I_{j_2}=k) = 1$ for all $1 \leq j_1 < j_2 \leq d$.  However, if we assign a
  positive probability to an index vector $\bm{I} \neq (k,k,\dots,k)$,
  $k\in\{1,\dots,K\}$, then there are indices $1\leq j_1 < j_2 \leq d$ such that
  $q_{(I_{j_1},I_{j_2})}>0$.  Consider without loss of generality that
  $\P(I_1=k,I_2=\ell,I_3=k,\dots,I_d=k)>0$, $k\neq \ell$.  In this case
  $\sum_{k=1}^{K} p_{(I_{j_1},I_{j_2})}(k,k) < 1 \Rightarrow q_{(I_1,I_2)} > 0$
  and hence we have the discussed bounds shrinkage for the margin $C^{1,2}$.

  If we are only interested in preventing bounds shrinkage for a \emph{specific
    pair} $1 \leq j_1 < j_2 \leq d$, choices other than
  $\bm{I}=(k,k,\dots,k)\in\IR^d$ are possible.  Taking without loss of
  generality $j_1=1$ and $j_2=2$, the only requirement for
  $\sum_{k=1}^{K} p_{(I_{j_1},I_{j_2})}(k,k) = 1$ is that all considered index
  vectors are of the form $\bm{I}=(k,k,\ell_3,\dots,\ell_d)$, but there is no
  requirement that $\ell_j=k$, $j=3,\dots,d$.
\end{remark}

\subsubsection{Multivariate measures of concordance}
We now consider multivariate extensions of Blomqvist's beta and Spearman's rho.

By \cite[Section~10.3.3]{jaworskidurantehaerdlerychlik2010}, multivariate
Blomqvist's beta is given by
\begin{align*}
  \beta^d=\frac{2^{d-1}(C(\bm{1/2})+\hat{C}(\bm{1/2}))-1}{2^{d-1}-1},
\end{align*}
where $\bm{1/2}=(1/2,\dots,1/2)$; note that
\cite[Section~10.3.3]{jaworskidurantehaerdlerychlik2010} use
$\bar{C}(\bm{u})=\P(\bm{U}>\bm{u})=\hat{C}(\bm{1}-\bm{u})$ instead of
$\hat{C}(\bm{u})$ but the two coincide at $\bm{1/2}$.
If $C$ is radially symmetric, then
\begin{align*}
  \beta^d=\frac{2^dC(\bm{1/2})-1}{2^{d-1}-1}.
\end{align*}
For a radially symmetric index-mixed copula $C$ (see
Proposition~\ref{prop:surv:cop}), utilizing $d=\sum_{k=1}^K D_k$ implies that
\begin{align*}
  \beta^d&=\frac{2^dC(\bm{1/2})-1}{2^{d-1}-1}=\frac{2^d\E_{\bm{I}}\bigl[\,\prod_{k=1}^K\delta_k^{I_{\cdot, k}^{\text{mat}}}(1/2)\bigr]-1}{2^{d-1}-1}=\frac{\E_{\bm{I}}\bigl[\,\prod_{k=1}^K2^{D_k}\delta_k^{I_{\cdot, k}^{\text{mat}}}(1/2)\bigr]-1}{2^{d-1}-1}\\
         &=\frac{\E_{\bm{I}}\bigl[\,\prod_{k=1}^K\bigl((2^{D_k-1}-1)\beta^{D_k}(C_k^{I_{\cdot,k}^{\text{mat}}})+1\bigr)\bigr]-1}{2^{d-1}-1},
\end{align*}
where $\beta^{D_k}(C_k^{I_{\cdot,k}^{\text{mat}}})$ denotes the $D_k$-dimensional Blomqvist's beta of $C_k^{I_{\cdot,k}^{\text{mat}}}$.

Also multivariate Spearman's rho can be computed for index-mixed copulas. For
example, \cite{wolff1980}, \cite{joe1990} or
\cite[Sections~10.3.1]{jaworskidurantehaerdlerychlik2010} %
consider
\begin{align*}
  \rho_{\text{S}}^{\text{l},d}=\frac{d+1}{2^d-(d+1)}\biggl(2^d\mu_{C,\Pi}-1\biggr),\quad\text{and}\quad
  \rho_{\text{S}}^{\text{u},d}=\frac{d+1}{2^d-(d+1)}\biggl(2^d\mu_{\Pi,C}-1\biggr).
\end{align*}
\cite{nelsen1996} introduced
$\rho_{\text{S}}^{\text{c},d}=(\rho_{\text{S}}^{\text{l},d}+\rho_{\text{S}}^{\text{u},d})/2$, where all three measures coincide with $\rho_{\text{S}}$ for
$d=2$. %
\begin{proposition}[Multivariate extensions of Spearman's rho]\label{prop:mult:spearman}
  Consider the setup of Definition~\ref{def:index:mixed:copulas}.
  Then
  \begin{align*}
    \rho_{\text{S}}^{\text{l},d}&=\frac{d+1}{2^d-(d+1)}\biggl(\E_{\bm{I}}\biggl[\,\prod_{k=1}^K\biggl(\frac{2^{D_k}-(D_k+1)}{D_k+1}\rho_{\text{S}}^{\text{l},D_k}(C_k^{I_{\cdot,k}^{\text{mat}}})+1\biggr)\biggr]-1\biggr),\\
    \rho_{\text{S}}^{\text{u},d}&=\frac{d+1}{2^d-(d+1)}\biggl(\E_{\bm{I}}\biggl[\,\prod_{k=1}^K\biggl(\frac{2^{D_k}-(D_k+1)}{D_k+1}\rho_{\text{S}}^{\text{u},D_k}(C_k^{I_{\cdot,k}^{\text{mat}}})+1\biggr)\biggr]-1\biggr),
  \end{align*}
  where $\rho_{\text{S}}^{\text{l},D_k}(C_k^{I_{\cdot,k}^{\text{mat}}})$ and
  $\rho_{\text{S}}^{\text{u},D_k}(C_k^{I_{\cdot,k}^{\text{mat}}})$ denote the
  $D_k$-dimensional Spearman's rhos of $C_k^{I_{\cdot,k}^{\text{mat}}}$.
\end{proposition}
Note that these measures can also be made dependent on covariates. For example,
\cite[Section~2.3]{gijbelsmatterne2021} consider a setting where a mixture of
two bivariate copulas depends on covariates through the mixture weight. In our
case, one could consider the index distribution to depend on covariates.

\subsection{Concordance ordering}
We discuss probabilistic orderings for index-mixed copulas in the following
section, where the following definitions are essentially taken from
\cite[Chapter~5.7]{nelsen2006}.  Two $d$-dimensional distribution functions $F$
and $G$ are ordered according to the \emph{lower concordance
  order} %
denoted by $F\preceq_l G$, if, pointwise, $F\le G$. They are ordered according
to the \emph{upper concordance order}, denoted by $F\preceq_u G$, if, pointwise,
$\bar{F}\le\bar{G}$.  And they are ordered according to the \emph{concordance
  order}, denoted by $F\preceq_c G$, if, $F\preceq_l G$ and $F\preceq_u
G$. %
In general, note that upper and lower concordance order (and hence also
concordance order) are equivalent for $d=2$ if both distributions have the same margins.
The same is not true in higher dimensions, see \cite[Example 5.26, p. 222]{nelsen2006}.
Furthermore, concordance ordering as defined above is only sensible for joint distributions with identical margins, see \cite[Definition~3.8.5, p. 111]{mullerstoyan2002} and the discussion therein.
\begin{remark}[On different definitions of certain orderings]
   Ordering concepts for random vectors are widely studied in the literature
   and, unfortunately, the utilized definitions are not always perfectly
   aligned.  For example, while we follow \cite[Definition~5.7.2]{nelsen2006}
   for the definitions of $\preceq_l$ and $\preceq_u$, comprehensive discussions
   of stochastic ordering concepts can also be found in \cite{mullerstoyan2002}
   and \cite{shakedshanthikumar2007}. While the definitions of lower
   $\preceq_{\text{lo}}$ and upper $\preceq_{\text{uo}}$ orthant order coincide
   in the two references \cite[p.~90]{mullerstoyan2002} and
   \cite[p.~308]{shakedshanthikumar2007}, they differ from $\preceq_l$ and
   $\preceq_u$ in \cite[Definition~5.7.2]{nelsen2006} in the sense that
   $F \preceq_{\text{uo}} G \Leftrightarrow F \preceq_{u} G$, but
   $F \preceq_{\text{lo}} G \Leftrightarrow G \preceq_{l} F$.
 \end{remark}
If $\prod_{j=1}^d G_j \preceq_l G$, where $G_1,\dots,G_d$ are the margins of $G$,
then $G$ is called \emph{positively lower orthant dependent (PLOD)} and if $\prod_{j=1}^d G_j \preceq_u G$, $G$ is called \emph{positively upper orthant dependent
  (PUOD)}.
If both inequalities are satisfied, $G$ is called \emph{positively orthant dependent (POD)}.
This leads us to the following results on orthant dependence for index-mixed copulas.
\begin{proposition}[Orthant dependence]\label{prop:orthant:dep}
  Consider the setup of Definition~\ref{def:index:mixed:copulas}.
  \begin{enumerate}
  \item\label{prop:orthant:dep:1} If $C_k\succeq_l \Pi$, $k=1,\dots,K$, then $C$ is PLOD.
  \item\label{prop:orthant:dep:2} If $\hat{C}_k\succeq_l \Pi$, $k=1,\dots,K$, then $C$ is PUOD.
  \item If $C_k\succeq_l \Pi$ and $\hat{C}_k\succeq_l \Pi$, $k=1,\dots,K$, then $C$ is POD.
  \end{enumerate}
\end{proposition}
Now let us consider the case where we compare two index-mixed copulas.
\begin{proposition}[Concordance orderings]\label{prop_concordance_ordering}
  Let $\bm{I},\bm{I}'$ be index vectors, $\tilde{\bm{U}}_k\sim C_k$,
  $k=1,\dots,K$, be independent and $\tilde{\bm{U}}_k'\sim C_k'$, $k=1,\dots,K$, be independent.  Furthermore,
  let $\tilde{U},\tilde{U}'$ be the corresponding copula matrices and let
  $\bm{U}=\tilde{U}_{\bm{I}}\sim C$ and $\bm{U}'=\tilde{U}_{\bm{I}'}'\sim C'$ be as
  in~\eqref{U:index:mixed}.
Lastly, appearing index vectors are assumed to be independent
  of appearing copula matrices.
  \begin{enumerate}
  \item\label{prop:concordance:same:I:1} If $\bm{I}=\bm{I}'$ almost surely and $C_k\preceq_l C_k'$, $k=1,\dots,K$, then $C\preceq_l C'$.
  \item\label{prop:concordance:same:I:2} If $\bm{I}=\bm{I}'$ almost surely and $C_k\preceq_u C_k'$, $k=1,\dots,K$, then $C\preceq_u C'$.
  \item If $\bm{I}=\bm{I}'$ almost surely and $C_k\preceq_c C_k'$, $k=1,\dots,K$, then $C\preceq_c C'$.
  \end{enumerate}
\end{proposition}
The following example considers the case of (unequal) ordered index vectors.  It
shows that applying ordered index vectors, ordered either in the $\preceq_l$ or
the almost sure partial order, to the same copula matrix does not lead to $\preceq_l$
ordered index-mixed copulas even if the base copulas are ordered according to
$\preceq_l$.
The question if additional (stronger) assumptions, one could for example simultaneously consider $\bm{I}_{\Pi} \preceq_l \bm{I} \preceq_l \bm{I}'$ and $\Pi \preceq_l C_k \preceq_l C_k'$, where $\bm{I}_{\Pi}$ is an index vector with independent components, leads to an order for the resulting index-mixed copulas is an open question for further research.
\begin{example}[Ordering of base copulas and index vectors does not imply ordering of index-mixed copulas]
  For both examples we set $K=2$ and $d=2$.
  \begin{enumerate}
  \item\label{ex:ordering:1} Consider the copula matrix $\tilde{U}$ with base copulas $C_1=W$ and
    $C_2=M$. Additionally, let $\bm{I}$ be such that
    $\P(I_1=1,I_2=2)=\P(I_1=2,I_1=1)=1/2$ and let $\bm{I}'$ be such that
    $\P(I_1'=k_1,I_2'=k_2)=1/4$ for all $k_1,k_2\in\{1,2\}$. Then the index-mixed
    copula $C$ of $\tilde{U}_{\bm{I}}$ is $C=\Pi$ and the index-mixed copula $C'$
    of $\tilde{U}_{\bm{I}'}$ is $C'=\frac{1}{4}M +\frac{1}{4}W +\frac{1}{2}\Pi$.
    Furthermore, we have $C(0.75,0.75)=0.5625\le 0.59375=C'(0.75,0.75)$, but
    $C(0.75,0.25)=0.1875\geq 0.15625=C'(0.75,0.25)$ implying there is no order
    between $C$ and $C'$ even though the index vectors and base copulas satisfy
    $\bm{I}\preceq_l\bm{I}'$ and $C_1\preceq_l C_2$.
  \item Consider $C_1 = \frac{1}{4}M +\frac{1}{4}W +\frac{1}{2}\Pi$ and $C_2=M$.
    Furthermore, define $\bm{I}$ such that $p_{\bm{I}}(1,1)=1$, and $\bm{I}'$
    such that $p_{\bm{I}'}(1,2)=p_{\bm{I}'}(2,1)=1/2$.  This configuration
    now leads to the copula $C$ of $\tilde{U}_{\bm{I}}$ to be $C_1$ and the
    copula $C'$ of $\tilde{U}_{\bm{I}'}$ to be $\Pi$. As in Part~\ref{ex:ordering:1},
    there is no order between $C$ and $C'$ even though $C_1 \preceq_l C_2$ and
    $\bm{I} \leq \bm{I}'$ almost surely in each component.
  \end{enumerate}
\end{example}

\subsection{Distribution of the sum}
In equation~\eqref{eq:joint:df} we presented the form of the distribution
function $H$ with margins $F_1,\dots,F_d$ and index-mixed copula $C$. The
following result is useful for risk management purposes, where the aggregate
loss $S_n=X_1+\dots+X_n\sim F_{S_n}$ of dependent, nonnegative risks
$X_1\sim F_1,\dots,X_d\sim F_d$ is of interest; see also
\cite{blierwongcossettemarceau2022c}. To this end, let
$\LS[H](\bm{t})=\E[e^{-\bm{t}\T\bm{X}}]$, $\bm{t}\in[0,\infty)^d$, denote the
Laplace--Stieltjes transform of the random vector $\bm{X}\sim H$.
\begin{theorem}[Laplace--Stieltjes transform of $\bm{X}$ and its sum]\label{thm:LS:sum}
  Consider the setup of Definition~\ref{def:index:mixed:copulas} and
  marginal distributions $F_1,\dots,F_d$ supported on $[0,\infty)$.
  Then the Laplace--Stieltjes transform of $\bm{X}=(X_1,\dots,X_d)$
  with margins $F_1,\dots,F_d$ and index-mixed copula $C$ is
  \begin{align*}
    \LS[H](\bm{t})=\E_{\bm{I}}\biggl[\,\prod_{k=1}^K\LS[H_k](\bm{t}I_{\cdot,k}^{\text{mat}})\biggr],\quad\bm{t}\in[0,\infty)^d,
  \end{align*}
  where $H_k$ denotes the distribution function with copula $C_k$ and margins
  $F_1$, $\dots$, $F_d$, and $\bm{t}I_{\cdot,k}^{\text{mat}}=(t_1I_{1,k}^{\text{mat}},\dots,t_dI_{d,k}^{\text{mat}})$ denotes
  the (elementwise) Hadamard product between $\bm{t}$ and the $k$th column of $I^{\text{mat}}$.
  In particular,
  $\LS[F_{S_n}](t)=\LS[H](t,\dots,t)$, $t\ge 0$.
\end{theorem}

\begin{example}[Distribution of index-mixed-dependent sum under exponential margins]
  Let $\bm{X}\sim H$ with index-mixed copula $C$ where $K=2$, $C_1=M$, $C_2=\Pi$, and let $X_j\sim\Exp(\lambda)$,
  $j=1,\dots,d$. With $\bm{X}_1:=(-\frac{1}{\lambda}\log(U)$, $\dots,-\frac{1}{\lambda}\log(U))\sim H_1$ for $U\sim\U(0,1)$ we thus have
  \begin{align*}
    \LS[H_1](\bm{t})&=\E\bigl[e^{-\bm{t}\T\bm{X}_1}\bigr]=\E\Bigl[e^{-\sum_{j=1}^dt_jX_{1,j}}\Bigr]=\E\Bigl[e^{\log(U)\sum_{j=1}^dt_j/\lambda}\Bigr]=\E\Bigl[U^{\sum_{j=1}^dt_j/\lambda}\Bigr]
                    =\frac{1}{1+\frac{1}{\lambda}\sum_{j=1}^d t_j}.
  \end{align*}
  Furthermore, setting
  $\bm{X}_2:=(-\frac{1}{\lambda}\log(U_1),\dots,-\frac{1}{\lambda}\log(U_d))\sim
  H_2$ for $U_1$, $\dots$, $U_d\isim\U(0,1)$ yields
  \begin{align*}
    \LS[H_2](\bm{t})&=\E\bigl[e^{-\bm{t}\T\bm{X}_2}\bigr]=\E\Bigl[e^{-\sum_{j=1}^dt_j(-\log(U_j))/\lambda}\Bigr]=\E\biggl[\,\prod_{j=1}^d U_j^{t_j/\lambda}\biggr]=\prod_{j=1}^d \E\bigl[U_j^{t_j/\lambda}\bigr]=\prod_{j=1}^d \frac{1}{1+\frac{t_j}{\lambda}}.
  \end{align*}
  By Theorem~\ref{thm:LS:sum},
  \begin{align*}
    \LS[H](\bm{t})&=\E_{\bm{I}}\bigl[\LS[H_1](\bm{t}I_{\cdot,1}^{\text{mat}})\LS[H_2](\bm{t}I_{\cdot,2}^{\text{mat}})\bigr]=\E_{\bm{I}}\biggl[\frac{\prod_{j=1}^d \frac{1}{1+t_jI_{j,1}^{\text{mat}}/\lambda}}{1+\frac{1}{\lambda}\sum_{j=1}^d t_jI_{j,2}^{\text{mat}}}\biggr]=\E_{\bm{I}}\biggl[\frac{\prod_{j:I_{j,1}^{\text{mat}}=1} \frac{1}{1+t_j/\lambda}}{1+\frac{1}{\lambda}\sum_{j:I_{j,2}^{\text{mat}}=1} t_j}\biggr].
  \end{align*}
Noting that the expectation with respect to $\bm{I}$ is simply a finite sum we have for $t\ge 0$ that
  \begin{align*}
    \LS[F_{S_n}](t)&=\LS[H](t,\dots,t)=\E_{\bm{I}}\biggl[\frac{\prod_{j:I_{j,1}^{\text{mat}}=1} \frac{1}{1+t/\lambda}}{1+\frac{1}{\lambda}\sum_{j:I_{j,2}^{\text{mat}}=1} t}\biggr]=\E_{\bm{I}}\biggl[\frac{(1+t/\lambda)^{-D_1}}{1+tD_2/\lambda}\biggr]\\
               &=\E_{\bm{I}}\biggl[\Bigl(1+\frac{t}{\lambda}\Bigr)^{-D_1}\Bigl(1+\frac{t}{\lambda/D_2}\Bigr)^{-1}\biggr]=\sum_{\bm{i}\in\{1,2\}^d}p_{\bm{I}}(\bm{i})\Bigl(1+\frac{t}{\lambda}\Bigr)^{-D_1(\bm{i})}\Bigl(1+\frac{t}{\lambda/D_2(\bm{i})}\Bigr)^{-1},
  \end{align*}
  where, similarly as in Section~\ref{sec:index:mixing}, $D_k(\bm{i})$
  denotes the number of $1$s in column $k$ of the index matrix corresponding to
  the index vector $\bm{i}$, $k=1,2$. Note that $(1+t/\beta)^{-\alpha}$ is the
  Laplace--Stieltjes transform of the $\Gamma(\alpha,\beta)$ distribution, where
  $\beta$ is the rate parameter. For $G_1^{\bm{i}}$ and $G_2^{\bm{i}}$ being the distribution
  functions of $\Gamma(D_1(\bm{i}),\lambda)$ and
  $\Gamma(1,\lambda/D_2(\bm{i}))$, respectively (interpreting
  $\Gamma(0,\lambda)$ and $\Gamma(1,\infty)$ as the point mass at $0$), we
  obtain for $t\ge 0$ that
  \begin{align*}
    \LS[F_{S_n}](t)&=\sum_{\bm{i}\in\{1,2\}^d}p_{\bm{I}}(\bm{i})\LS[G_1^{\bm{i}}](t)\LS[G_2^{\bm{i}}](t)=\E_{\bm{I}}\bigl[\LS[G_1^{\bm{I}}](t)\LS[G_2^{\bm{I}}](t)\bigr].
  \end{align*}
  As the inverse Laplace--Stieltjes transform (Mellin's inverse
  formula) %
  is linear, we obtain that $F_{S_n}=\E_{\bm{I}}\bigl[\LSi[\LS[G_1^{\bm{I}}]\LS[G_2^{\bm{I}}]]\bigr]$, which
  is the $\bm{I}$-mixture of the independent sum of a
  $\Gamma(D_1(\bm{I}),\lambda)$ and a $\Gamma(1,\lambda/D_2(\bm{I}))$ random
  variable. In (other) words, the distribution of the sum of $\Exp(\lambda)$
  random variables under the index-mixed copula with base copulas $M$ and $\Pi$
  is the expectation with respect to $\bm{I}$ of the independent sum of a
  $\Gamma(D_1(\bm{I}),\lambda)$ and a $\Gamma(1,\lambda/D_2(\bm{I}))$ random
  variable.
\end{example}

\subsection{Additional properties}

\subsubsection{Base copulas and independence}
If $K<d$, the pigeonhole principle implies that there are at least two
components of the index vector $\bm{I}$ which must be equal, so that at least one marginal copula of the corresponding base copulas appears in the index-mixed
copula. However, if $K\ge d$ (and especially $K\gg d$) it can happen that
the index-mixed copula is $\Pi$ even though no base copula is.
\begin{proposition}[$K\ge d$ and independence]\label{prop:K:ge:d}
  Suppose $K\ge d$. If the index distribution is such that $\P(I_{j_1}=I_{j_2})=0$
  for all $1\le j_1<j_2\le d$, then the index-mixed copula $C$ is $\Pi$ regardless of the choice of base copulas $C_k$, $1\leq k \leq K$.
\end{proposition}
Proposition~\ref{prop:K:ge:d} helps to understand how to avoid that $C$ equals
$\Pi$ if $K\ge d$. First, there should be at least one base copula that is not
$\Pi$; note that this can still mean that such a base copula is the product of
a lower-dimensional $\Pi$ and a non-trivial copula. Second, the index
distribution should be chosen such that positive probability mass is put on
selecting at least two dependent components of that base copula.

In the same vein, one could be interested in constructing ``continuously
index-mixed copulas'' based on an infinite (countable or uncountable) set of
base copulas. However, as noted before, one would need to guarantee that the
probability that at least two components of the index vector are equal is
non-zero for obtaining more than just the independence copula.

Index-mixed copulas also can be constructed in such a way that a subset of margins is independent of the remaining margins.
The following Proposition~\ref{prop_block_wise_independent} provides a sufficient condition to this end.
For the proof recall the definition of $J_{k}$ as a function of the index matrix as provided in Section~\ref{sec:index:mixing}.
\begin{proposition}[Block-wise independence]\label{prop_block_wise_independent}
  Denote by $J_{k}(\bm{I})$ the $k$th index partition associated with the index
  matrix linked to the index vector $\bm{I}$.  If, for a fixed index
  $1\leq \ell \leq K$, the index partition $J_{\ell}(\bm{I})$ is independent of
  $\bm{I}$, for all $\bm{I}$ with positive probability, that is
  $J_{\ell}(\bm{I}) = J_{\ell} = (\ell_1,\dots,\ell_m)$, then the index-mixed
  copula factorizes into two independent blocks as
  \begin{align}
    C(u_1,\dots,u_d) = C_{\ell}(u_{\ell_1},\dots,u_{\ell_m})\overline{C}_{\ell}(u_{s_1},\dots,u_{s_{m^*}}),
  \end{align}
  where $(s_1,\dots,s_{m^*}) := \{1,\dots,d\} \setminus J_{\ell} $ and
  $\overline{C}_{\ell}$ is a copula of the remaining variables.  Throughout we
  assume $\emptyset \neq J_{\ell} \neq \{1,\dots,d\}$ to avoid having only a
  single block.
\end{proposition}
The result in Proposition~\ref{prop_block_wise_independent} can be generalized to the case of multiple independent blocks by keeping multiple index partitions fixed across all index vectors.

\subsubsection{Pairwise conditional copulas}
Conditional copulas $C^{j|1,\dots,j-1}$, $j=2,\dots,d$, and thus the transform
of \cite{rosenblatt1952} are often complicated to compute;
see~\cite[Section~2.7]{hofertkojadinovicmaechleryan2018}. %
\cite{hofertmaechler2014a} considered pairwise Rosenblatt transforms, which, for
a bivariate margin with indices $j_1,j_2$ are given by
Lemma~\ref{lem:biv:mar} Part~\ref{lem:biv:mar:1} as
\begin{align}
  C^{j_2|j_1}(u_{j_2}\,|\,u_{j_1})=\sum_{k=1}^Kp_{(I_{j_1},I_{j_2})}(k,k)C_k^{j_2|j_1}(u_{j_2}\,|\,u_{j_1})+\biggl(\,\sum_{\substack{k_1,k_2=1\\ k_1\neq k_2}}^Kp_{(I_{j_1},I_{j_2})}(k_1,k_2)\biggr) u_{j_2}\label{eq:rosen:biv}
\end{align}
for all $u_{j_1},u_{j_2}\in[0,1]$, where
$C_k^{j_2|j_1}(u_{j_2}\,|\,u_{j_1})=\frac{\partial}{\partial
  u_{j_1}}C_k^{j_1,j_2}(u_{j_1},u_{j_2})$ denotes the $k$th conditional base
copula at $u_{j_2}$ given $u_{j_1}$. Equation~\eqref{eq:rosen:biv} can be used for
graphical model assessment in the sense of \cite{hofertmaechler2014a}.

\subsubsection{Extendibility}
It is often of interest to known whether a given copula appears as the copula of
finite-dimensional distributions of a stochastic process.
Even in the class of exchangeable copulas, this is not always possible,
leading to the following definition; see \cite{maischerer2012}. %
\begin{definition}[Extendible copula]
  A $d$-dimensional copula $C$ is \emph{extendible}, if there exists an infinite
  exchangeable sequence $(U_n)_{n=1}^{\infty}$ such that $C$ is the distribution
  function of every $d$-dimensional finite-dimensional distribution, that is
  $(U_{j_1},\ldots,U_{j_d})\sim C$ for all integer indices
  $1 \leq j_1 < j_2 < \cdots < j_d$.
\end{definition}
We can now establish extendibility of index-mixed copulas.
\begin{proposition}[Extendibility of index-mixed copulas]\label{prop_extendibility}
  Let $C$ denote the $d$-dimensional index-mixed copula defined in
  Definition~\ref{def:index:mixed:copulas} based on an index vector $\bm{I}$ and
  $d$-dimensional base copulas $C_k$, $1 \leq k \leq K$.  If all base copulas
  $C_k$ are extendible, $\bm{I}$ has identical univariate marginals $G$ and
  there exists an extendible copula $C_0$ such that
  $\P(I_1 \leq i_1, \ldots, I_d \leq i_d) = C_0(G(i_1),\ldots,G(i_d))$, then
  the index-mixed copula $C$ is extendible.
\end{proposition}

\section{Examples, illustration, application}\label{sec:examples}
In this section, we consider examples, illustrations and an outlined application.
As trivial examples, one can directly see that an index-mixed copula is the independence copula if
(i) either all base copulas are the independence copula, or (ii) if,
irrespective of the base copulas, all components of $\bm{I}$ differ almost
surely.  Along the same lines, if the index vector $\bm{I}$ is comonotone and
$C_k=C_0$, $k=1,\dots,K$, for some copula $C_0$, then we have by
Corollary~\ref{cor:invar:equal:comon} that the index-mixed copula is (also)
$C=C_0$.

\begin{example}[Illustrating the effect of index-mixing for $d=2$]\label{ex:ill:index:mixing:2d}
  We now consider scatter plots of samples of size 10\,000 from the index-mixed
  copulas $C$ of order $K=2$ with base copulas $\Pi$ (the independence copula), $M$ (the comonotone copula),
  $C^{\text{C}}$ (a Clayton copula with parameter such that Kendall's tau equals 0.5) and $C^{\text{G}}$ (a Gumbel copula with parameter such that Kendall's tau equals 0.5);
  see Figure~\ref{fig:base:cop:samples} for scatter plots of these base copulas.
  \begin{figure}[htbp]
    \includegraphics[width=0.48\textwidth]{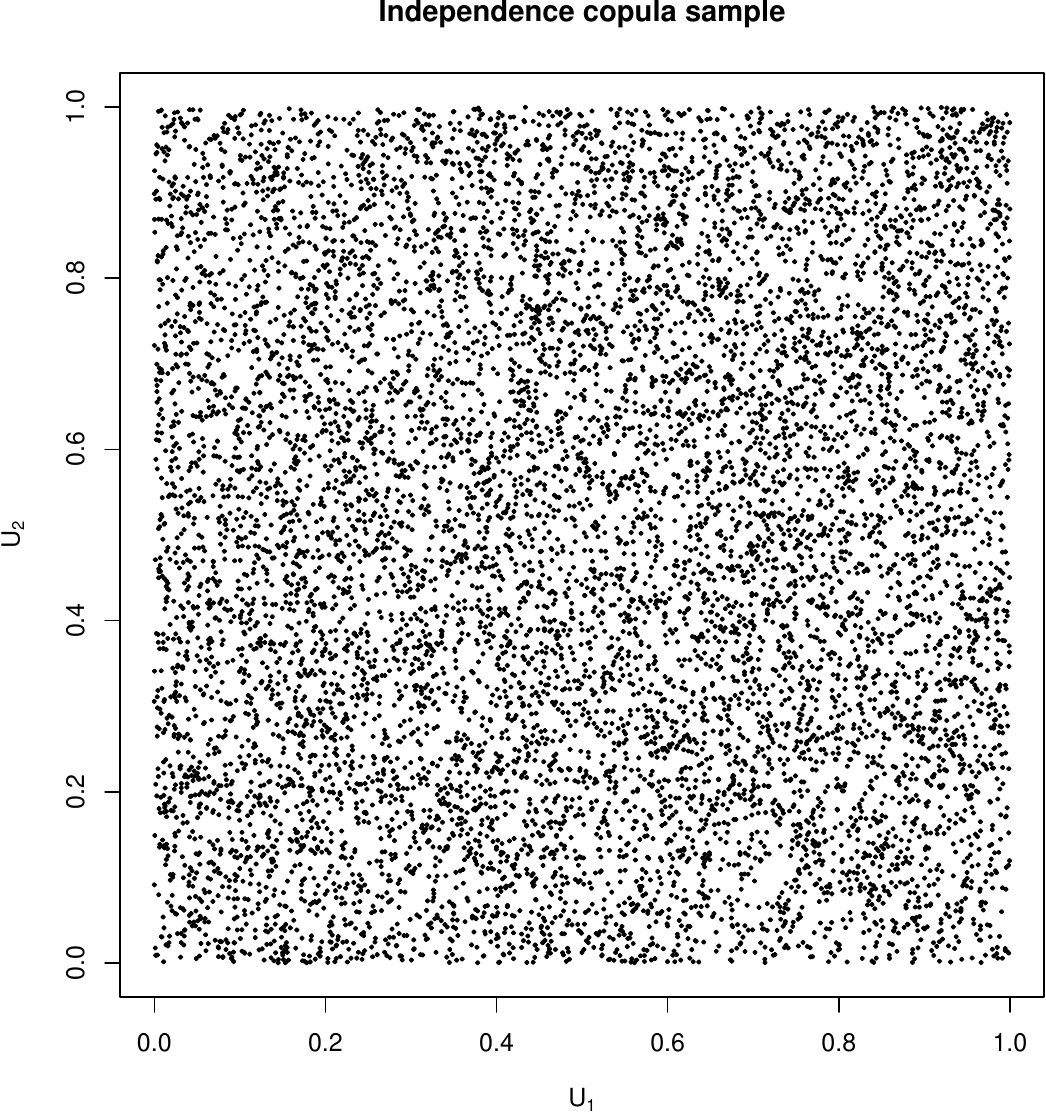}%
    \hfill
    \includegraphics[width=0.48\textwidth]{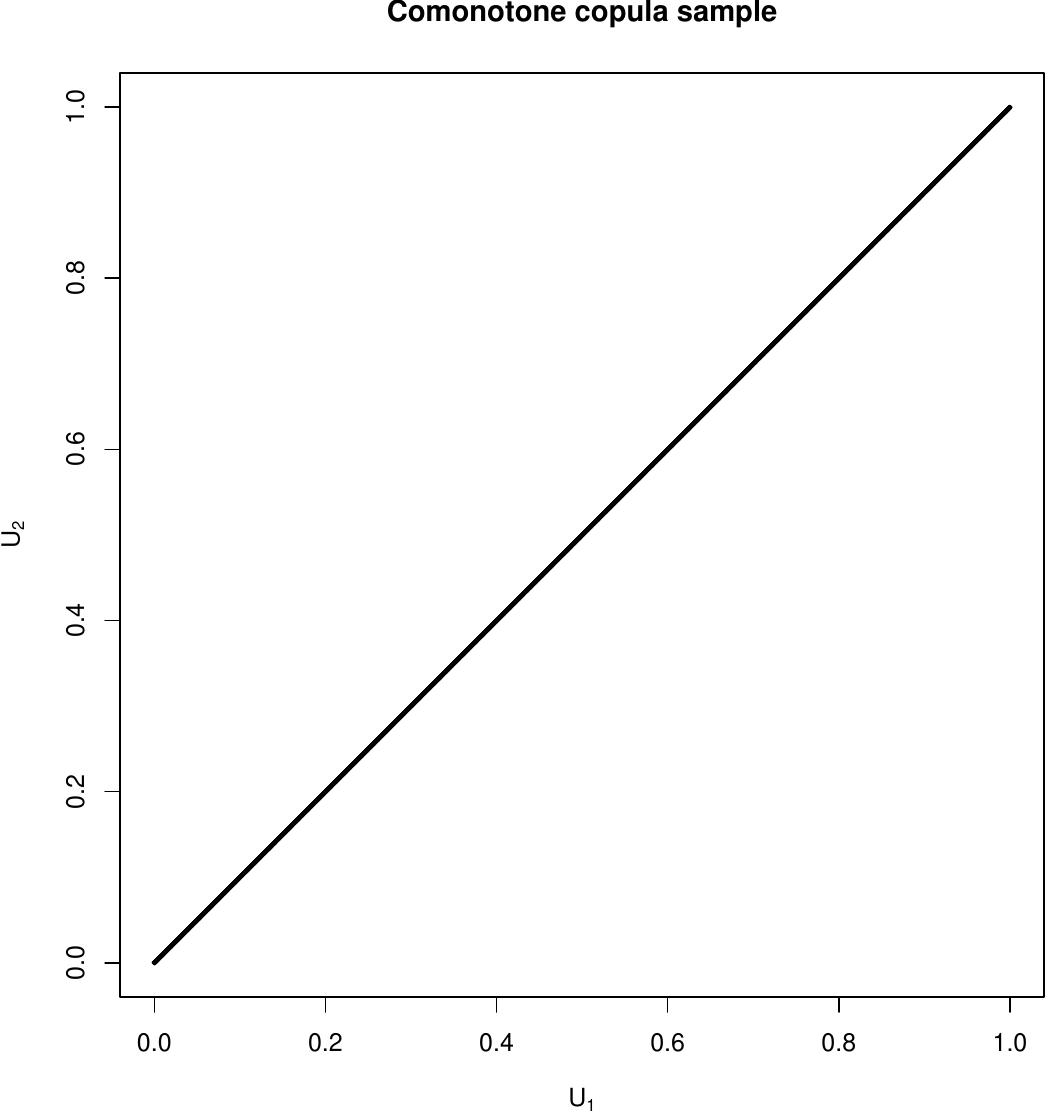}\\[4mm]%
    \includegraphics[width=0.48\textwidth]{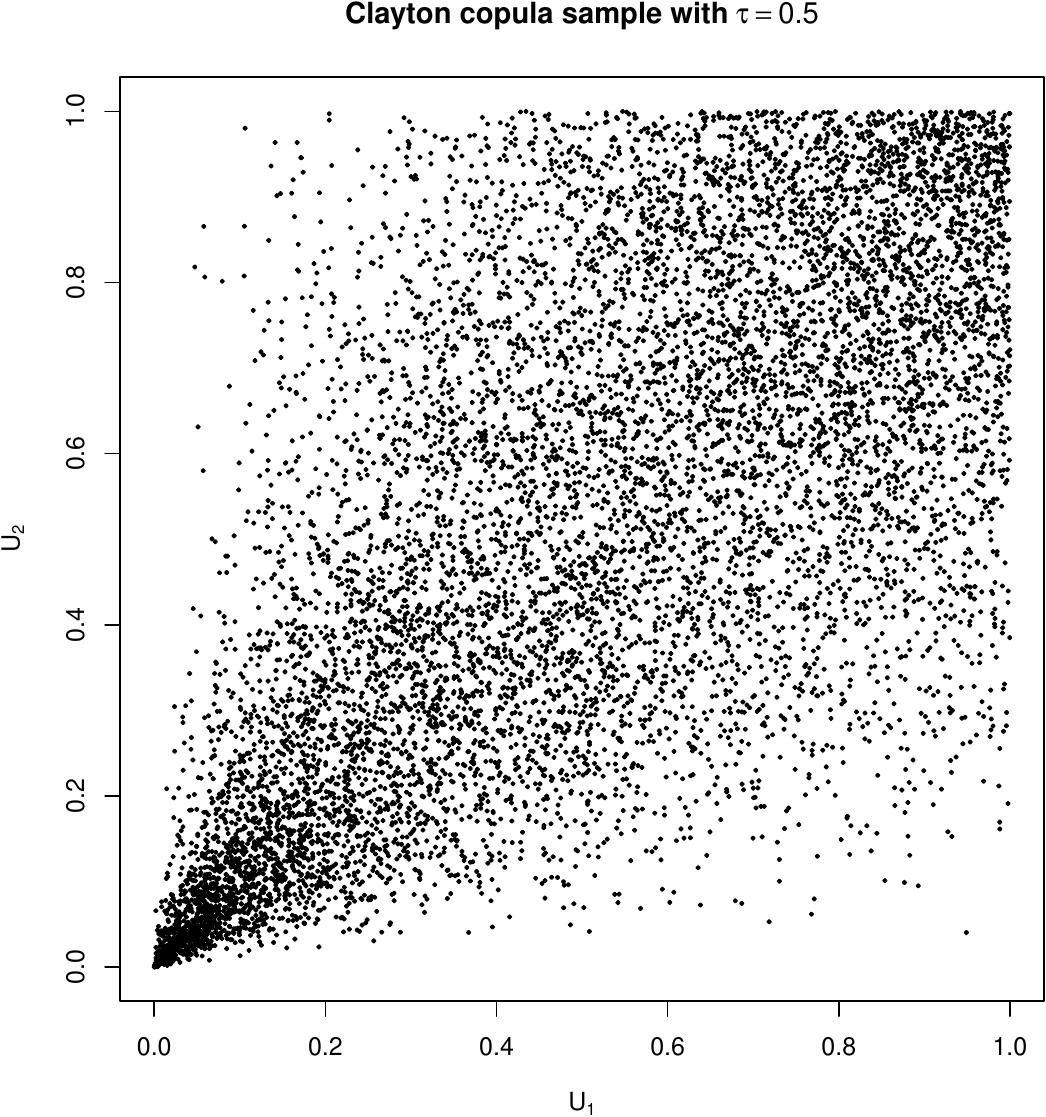}%
    \hfill
    \includegraphics[width=0.48\textwidth]{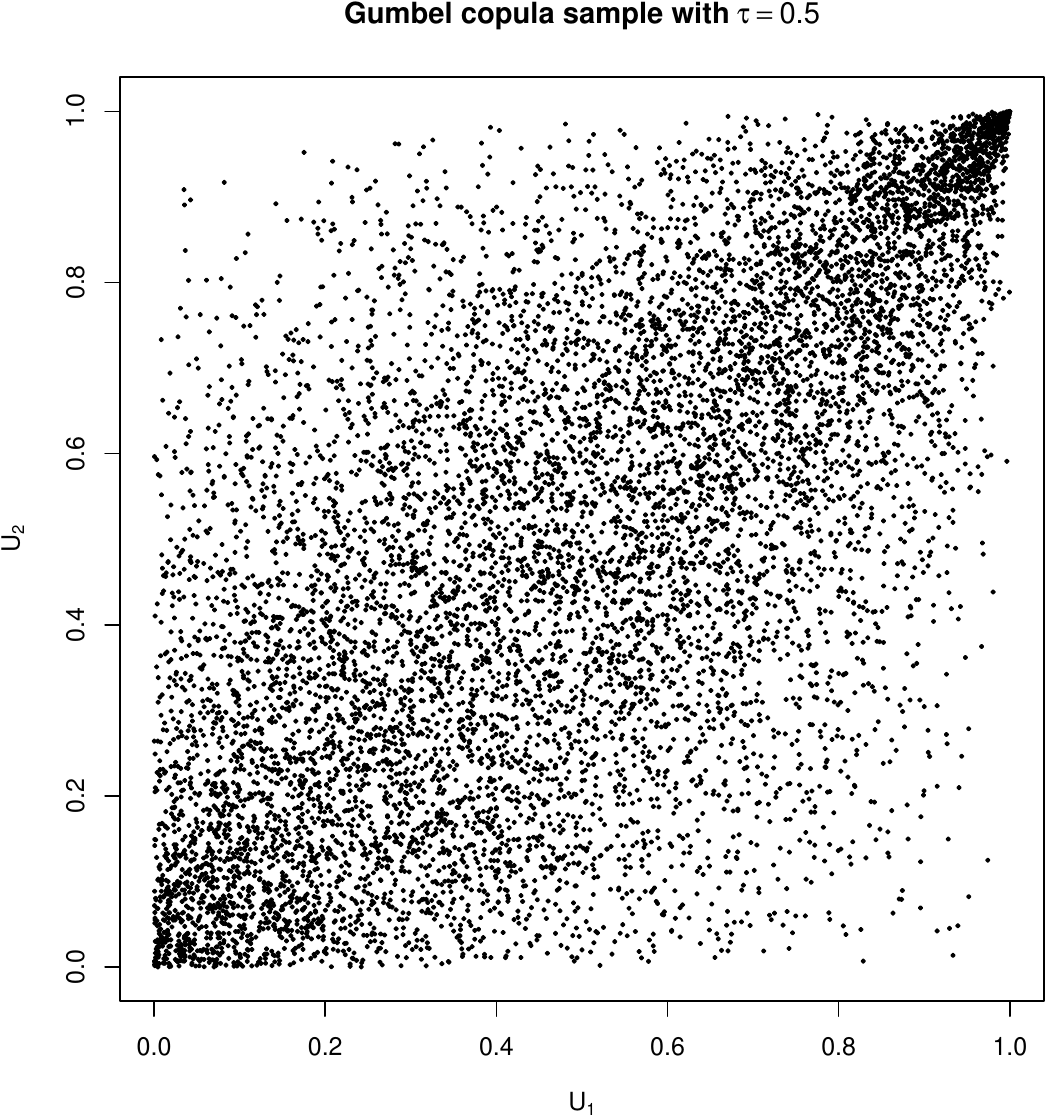}%
    \caption{Sample of size 10\,000 from the independence copula $\Pi$ (top left),
      comonotone copula $M$ (top right), Clayton copula $C^{\text{C}}$ (bottom left; with Kendall's
      tau $0.5$) and Gumbel copula $C^{\text{G}}$ (bottom right; with Kendall's tau $0.5$).}
    \label{fig:base:cop:samples}
  \end{figure}

  Figure~\ref{fig:IMC:2d:I:indep} shows samples from $C$ with index vector of
  the form $\bm{I}=\bm{1}+\bm{B}$ for $\bm{B}\sim\B^\Pi_2(1/2,\dots,1/2)$. On
  the left-hand side, the base copulas are $C^{\text{C}}$ and $C^{\text{G}}$. On
  the right-hand side, both base copulas are the comonotone copula
  $M$. %
  \begin{figure}[htbp]
    \includegraphics[width=0.48\textwidth]{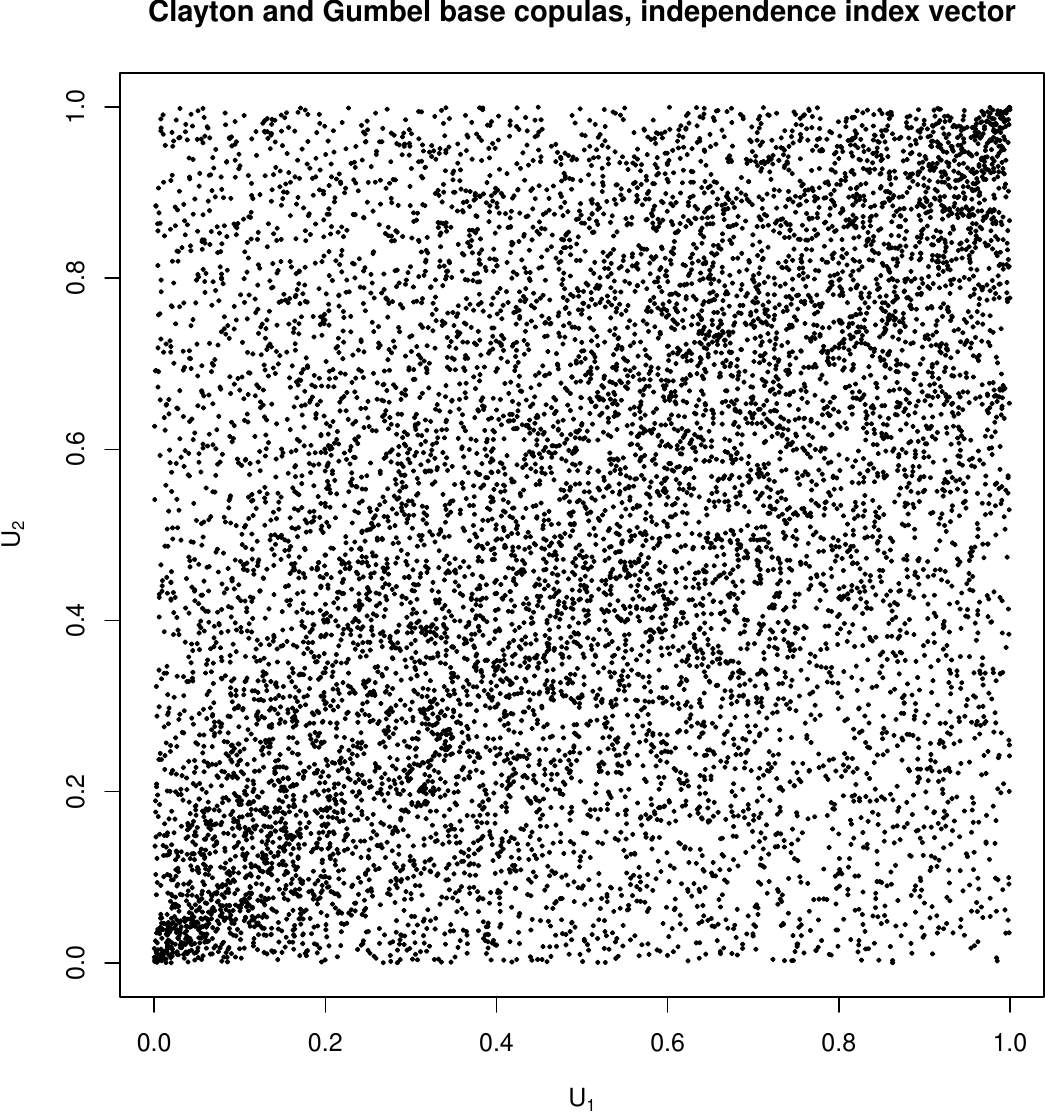}%
    \hfill
    \includegraphics[width=0.48\textwidth]{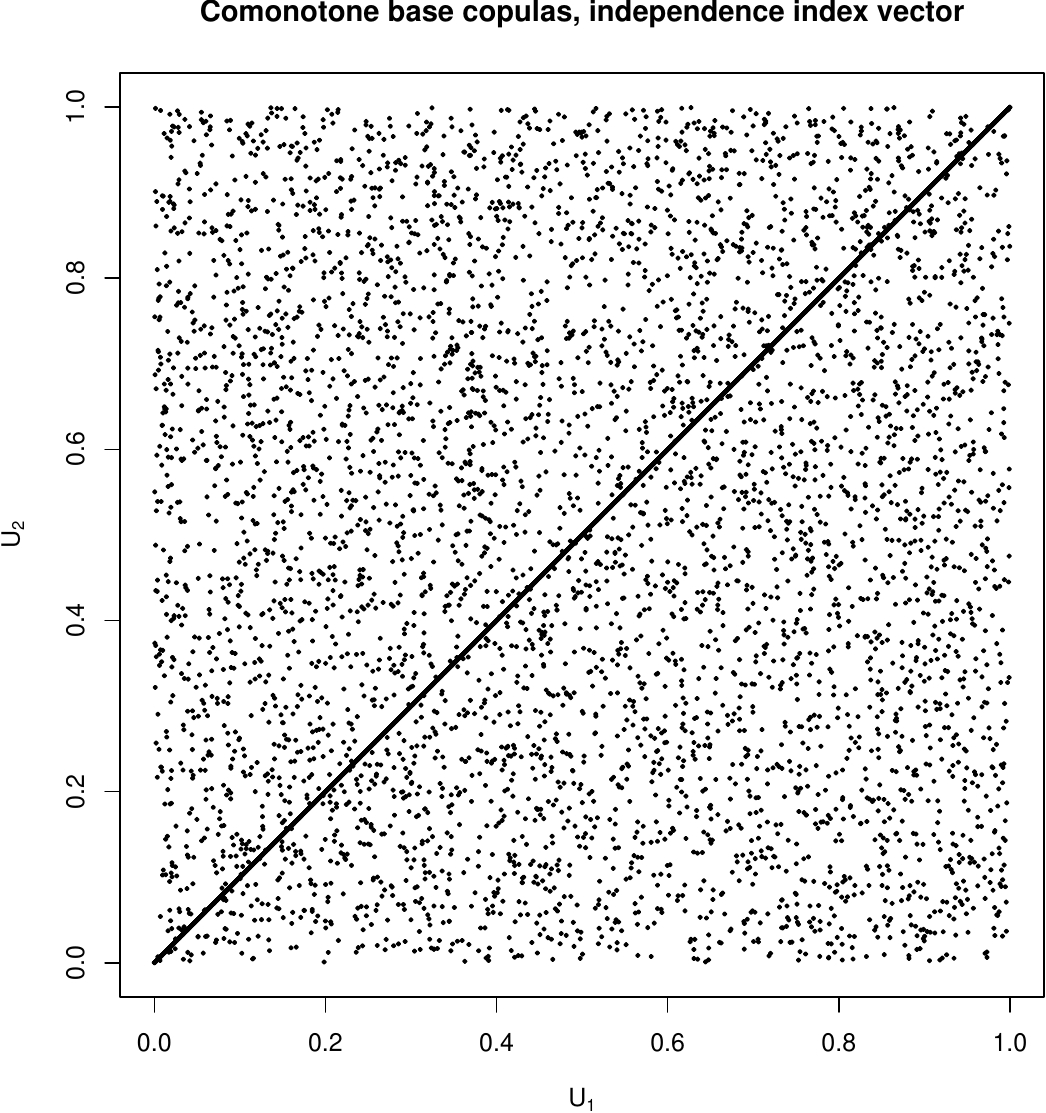}%
    \caption{Sample of size 10\,000 from the index-mixed copula $C$ with
      $\bm{I}=\bm{1}+\bm{B}$ for $\bm{B}\sim\B^{\Pi}_2(1/2)$ and with base
      copulas given by the Clayton copula $C^{\text{C}}$ and the Gumbel copula
      $C^{\text{G}}$ (left), and twice the comonotone
      copula $M$ (right).}
    \label{fig:IMC:2d:I:indep}
  \end{figure}
  With $\bm{I}=\bm{1}+\bm{B}$ for $\bm{B}\sim\B^\Pi_2(1/2,\dots,1/2)$, we have
  $p_{(I_1,I_2)}(k_1,k_2)=\P(I_1=k_1,I_2=k_2)=1/4$, $k_1,k_2\in\{1,2\}$, and thus, by
  Example~\ref{ex:biv:index:mixed} (see also Lemma~\ref{lem:biv:mar} Part~\ref{lem:biv:mar:1}),
  \begin{align*}
    C(u_1,u_2)=\frac{C^{\text{C}}}{4}+\frac{C^{\text{G}}}{4}+\frac{\Pi}{2},\quad u_1,u_2\in[0,1],
  \end{align*}
  that is the index-mixed copula is a mixture with probability $1/4$ on the
  Clayton copula, $1/4$ on the Gumbel copula and $1/2$ on the independence
  copula $\Pi$. The tail dependence coefficients of $C$ are $\lambda_{\text{l}}=\frac{2^{-1/2}}{4}$
  and $\lambda_{\text{u}}=\frac{2-2^{1/2}}{4}$.
  The left-hand side of Figure~\ref{fig:IMC:2d:I:indep} shows a sample of
  $C$. We clearly see features of all three copulas.

  Similarly, on the right-hand side of Figure~\ref{fig:IMC:2d:I:indep}, we see a
  mixture that puts probability $1/4+1/4=1/2$ on the comonotone copula $M$ and
  probability $1/2$ on $\Pi$, so
  \begin{align*}
    C(u_1,u_2)=\frac{M}{2}+\frac{\Pi}{2},\quad u_1,u_2\in[0,1].
  \end{align*}
  Note that even though the components of $\bm{I}$ are independent, we see that
  a larger range of concordance is attainable than for EFGM copulas (see
  Section~\ref{sec:EFGM:limited:dep}). In particular, by
  Proposition~\ref{prop_bivariate_Spearman_Kendall}, the equi-probable mixture
  of $M$ and $\Pi$ has Spearman's rho and Kendall's tau
  \begin{align*}
    \rho_{\text{S}}&=p_{(I_1,I_2)}(1,1)\rho_{\text{S},1}+p_{(I_1,I_2)}(2,2)\rho_{\text{S},2}=\frac{1}{4}\cdot 1+\frac{1}{4}\cdot 1=\frac{1}{2},\\
    \tau&=p_{(I_1,I_2)}(1,1)^2\tau_1+p_{(I_1,I_2)}(2,2)^2\tau_2+2p_{(I_1,I_2)}(1,1)p_{(I_1,I_2)}(2,2)(4\mu_{M,M}-1)\\
                   &\phantom{={}}+\frac{2}{3}(1-(p_{(I_1,I_2)}(1,1)+p_{(I_1,I_2)}(2,2)))(p_{(I_1,I_2)}(1,1)(12\mu_{M,\Pi}-3)+p_{(I_1,I_2)}(2,2)(12\mu_{M,\Pi}-3))\\
                   &=\biggl(\frac{1}{4}\biggr)^2\cdot 1+\biggl(\frac{1}{4}\biggr)^2\cdot 1+2\cdot\frac{1}{4}\cdot\frac{1}{4}\biggl(4\frac{1}{2}-1\biggr)+\frac{2}{3}\biggl(1-\biggl(\frac{1}{4}+\frac{1}{4}\biggr)\biggr)\biggl(\frac{1}{4}\biggl(12\frac{1}{3}-3\biggr)+\frac{1}{4}\biggl(12\frac{1}{3}-3\biggr)\biggr)\\
                   &=\frac{1}{8}+\frac{1}{8}+\frac{1}{6}=\frac{5}{12},
  \end{align*}%
  where we used that $\mu_{M,M}=\E(M(U,U))=\E(U)=1/2$ for $U\sim\U(0,1)$ and
  that $\mu_{M,\Pi}=\E(M(U_1,U_2))=\E(X)=1/3$ for $U_1,U_2\isim\U(0,1)$ and
  $X=\min\{U_1,U_2\}\sim F_X$ with $F_X(x)=1-(1-x)^2$, $x\in[0,1]$.
  By Proposition~\ref{prop:tail:dependence}, the tail dependence coefficients
  are
  \begin{align*}
    \lambda_{\text{l}}=\lambda_{\text{u}}=p_{(I_1,I_2)}(1,1)\lambda_{\text{u},1}+p_{(I_1,I_2)}(2,2)\lambda_{\text{u},2}=\frac{1}{4}\cdot 1+\frac{1}{4}\cdot 1=\frac{1}{2}.
  \end{align*}
  Also, asymmetric marginal distributions of $\bm{I}$ are possible, of
  course.

  Figure~\ref{fig:IMC:2d:I:comon} shows samples from the index-mixed copulas with
  $\bm{I}=\bm{1}+\bm{B}$ for $\bm{B}\sim\B^{M}_2(1/2,\dots,1/2)$;
  see Corollary~\ref{cor:invar:equal:comon} for the analytical form.
  The base copulas are twice $\Pi$ (top left), $C^{\text{C}}$ and $C^{\text{G}}$ (top
  right), $C^{\text{C}}$ and $M$ (bottom left) and twice $M$ (bottom right).
  \begin{figure}[htbp]
    \includegraphics[width=0.48\textwidth]{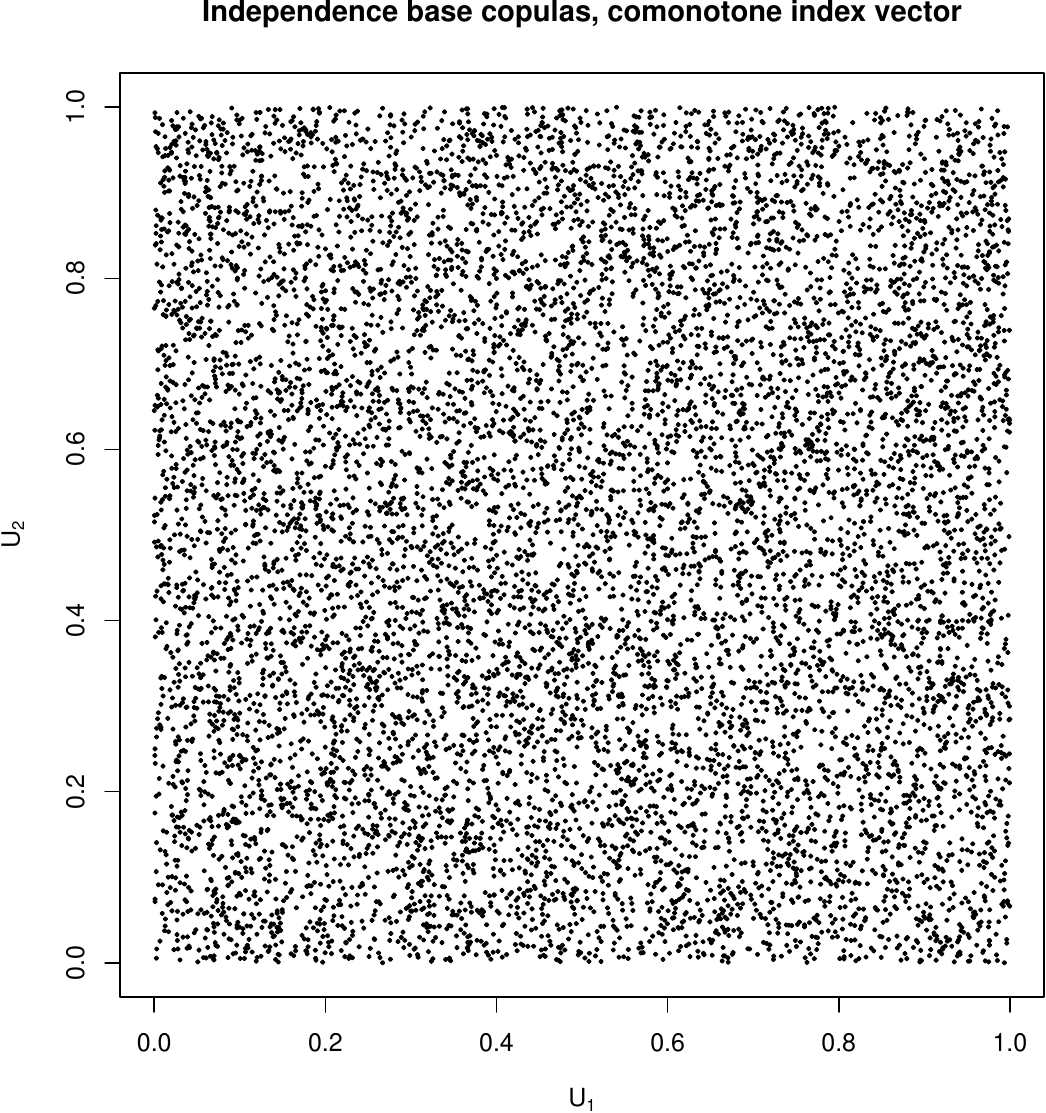}%
    \hfill
    \includegraphics[width=0.48\textwidth]{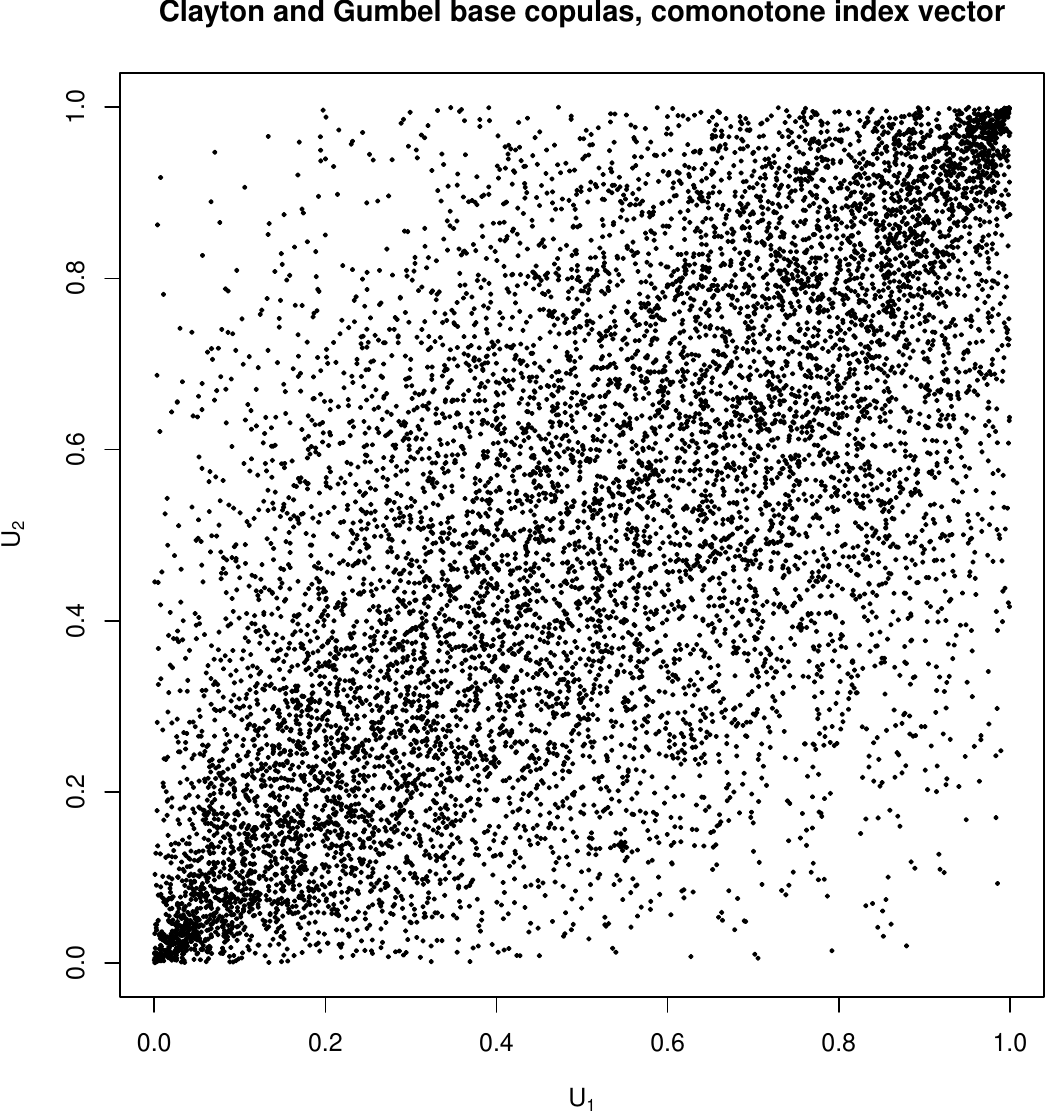}\\[4mm]%
    \includegraphics[width=0.48\textwidth]{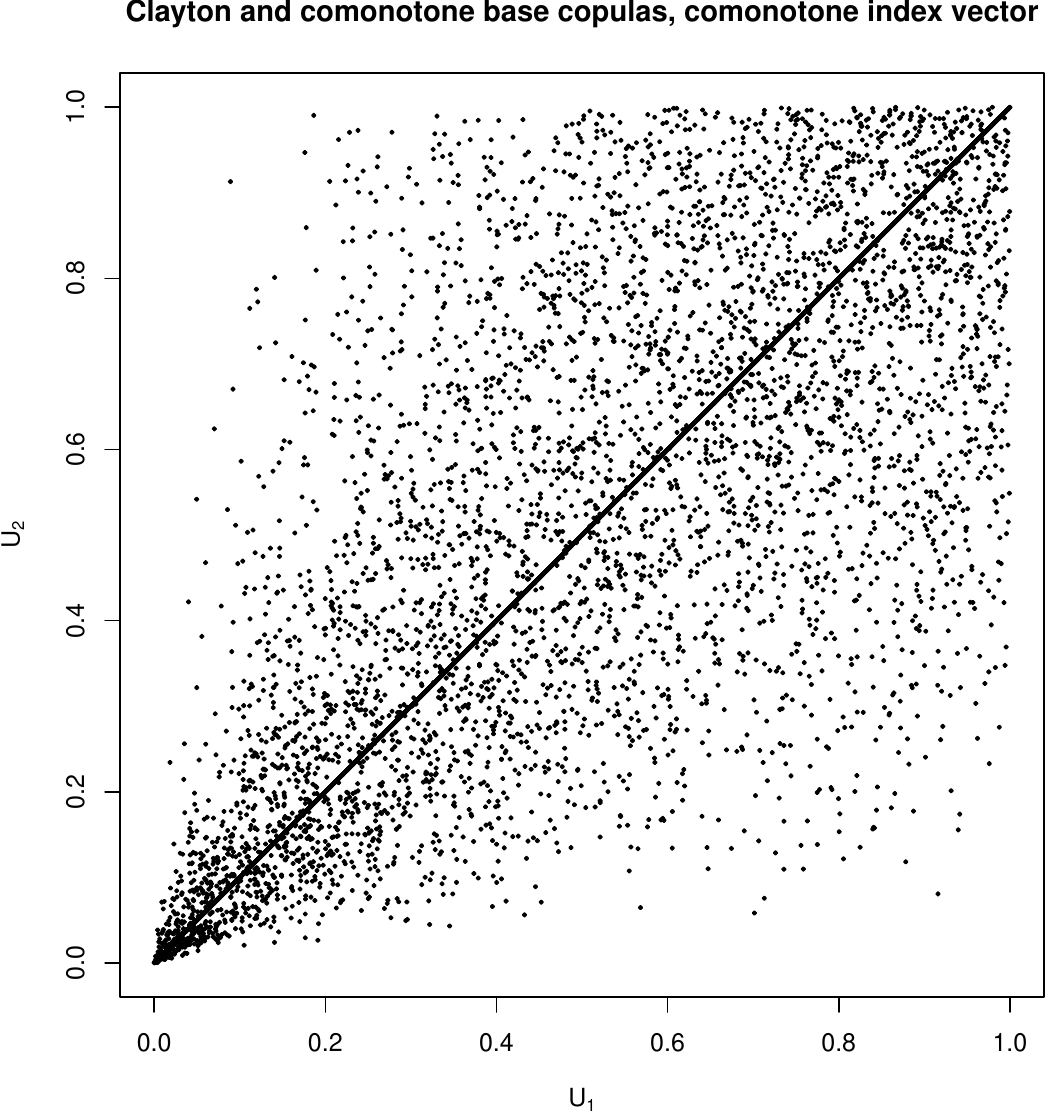}%
    \hfill
    \includegraphics[width=0.48\textwidth]{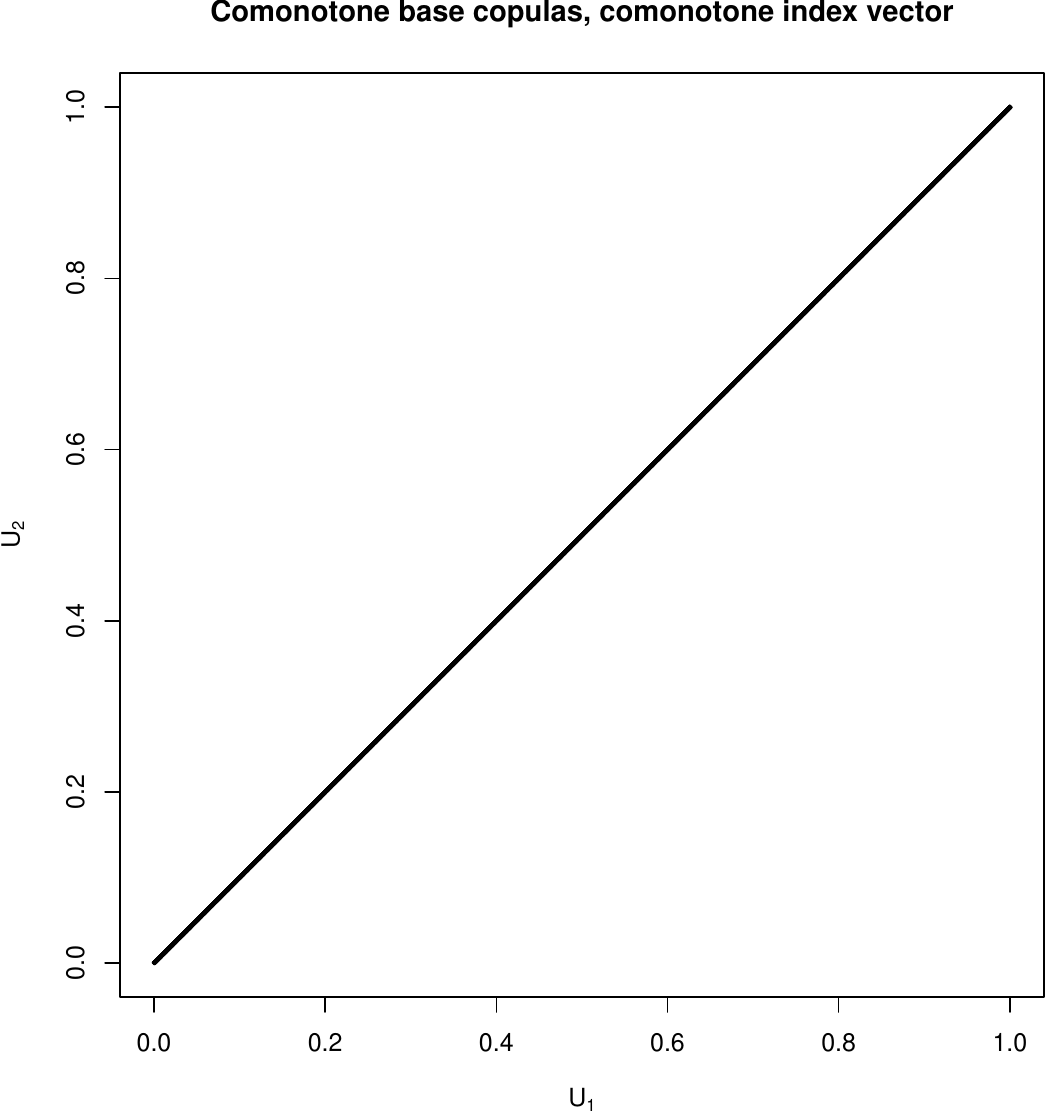}%
    \caption{Sample of size 10\,000 from the index-mixed copula $C$ with
      $\bm{I}=\bm{1}+\bm{B}$ for $\bm{B}\sim\B^{M}_2(1/2,\dots,1/2)$ and with base copulas
      given by the independence copula $\Pi$ (top left), the Clayton copula $C^{\text{C}}$ and Gumbel copula $C^{\text{G}}$
      (top right), the Clayton copula $C^{\text{C}}$
      and comonotone copula $M$ (bottom left) and twice the comonotone copula $M$ (bottom right).}
    \label{fig:IMC:2d:I:comon}
  \end{figure}

  Comonotone $\bm{I}$ select the two columns of the copula matrix with
  probabilities $1/2$ in our considered cases here. By
  Corollary~\ref{cor:invar:equal:comon}, we thus observe samples from (classical) mixtures
  of the base copulas. If $\bm{I}$ has comonotone components and if both base
  copulas are $M$, then we have a mixture between $M$ and $M$, which results in
  $M$; see the bottom right plot of Figure~\ref{fig:IMC:2d:I:comon}.
\end{example}

\begin{example}[Stress testing application]
  A direct application of index-mixed copulas is stress testing. A given
  $d$-dimensional benchmark copula model currently in use in a financial firm
  can serve as base copula $C_1$ with $p_{\bm{I}}(\bm{1})>0$, and the base
  copulas $C_2,\dots,C_K$ can be regarded as $K-1$ stress(-testing) copulas. Through suitably
  chosen index distributions, the resulting index-mixed copula mixes (possibly
  $d$-dimensional) margins of $C_1$ with margins of $C_2,\dots,C_K$, replacing
  the former by the latter; see Remark~\ref{rem:constr}. Due to the independence
  assumption of different columns in the copula matrix, each mixture
  contribution is (again by Remark~\ref{rem:constr}) the product of copulas with
  non-overlapping arguments, so two components with one belonging to $C_1$ and
  one belonging to any of $C_2,\dots,C_K$ are independent.

  Let us now consider a simple example in this context. Let $C^{t_\nu,\tau}$
  denote a homogeneous $t$ copula with $\nu$ degrees of freedom and homogeneous
  correlation parameter such that Kendall's tau of any pairwise marginal copula
  is $\tau$. Suppose the financial firm currently models losses in its five
  business lines ($d=5$) by $C_1(u_1,\dots,u_5)=C^{t_4,0.5}(u_1,\dots,u_5)$ with
  standard normal margins.  As one of its stress models, the financial firm
  considers the copula
  $C_2(u_1,\dots,u_5)=C^{t_4,0.4}(u_1,u_2,u_3)C^{t_4,0.6}(u_4,u_5)$, while the
  margins remain as before. To see the effect of a potential restructuring of
  the five business lines into a group with three and one with two business
  lines, the firm considers various scenarios, specified via a sequence of
  probability weights between $0$ and $1$ that determine how much probability
  mass is put on the stress model. The remaining probability mass is then
  automatically put on the initial benchmark model. Hence $\bm{I}=(1+\B(p))\bm{1}$,
  where $p$ runs on a grid between $0$ and $1$.

  The left-hand side of Figure~\ref{fig:app:total:loss} shows boxplots of $10^6$
  simulated total losses across all five business lines as functions of the probability mass
  put on the stress model.
  \begin{figure}[htbp]
    \includegraphics[height=0.25\textheight]{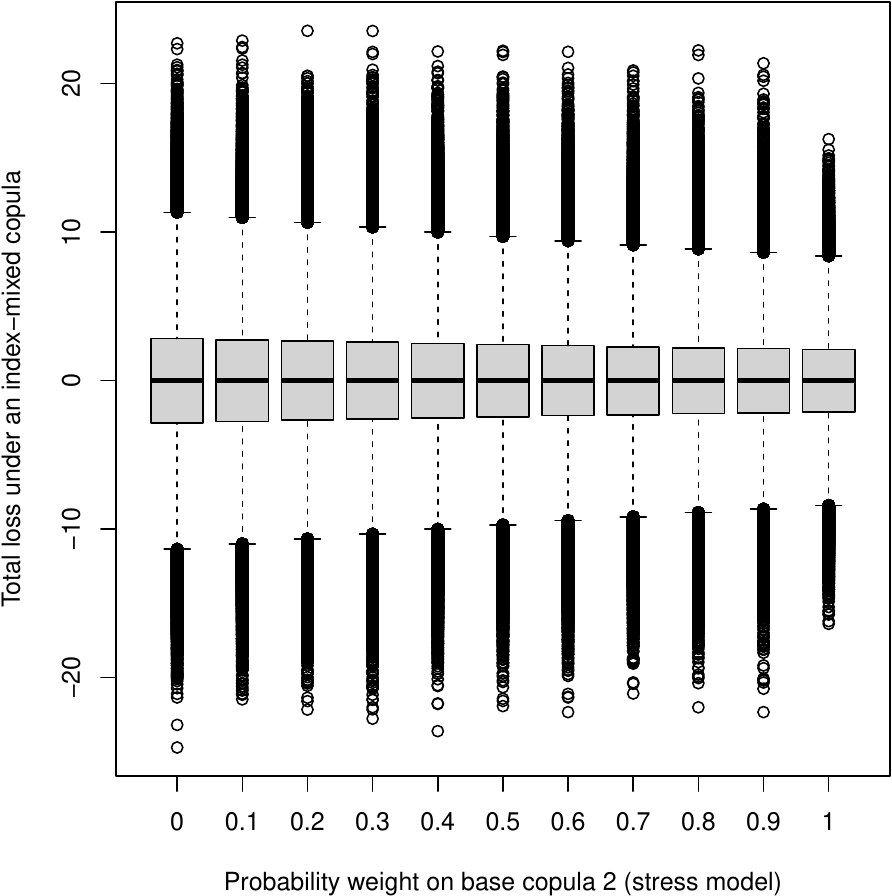}
    \hfill
    \includegraphics[height=0.25\textheight]{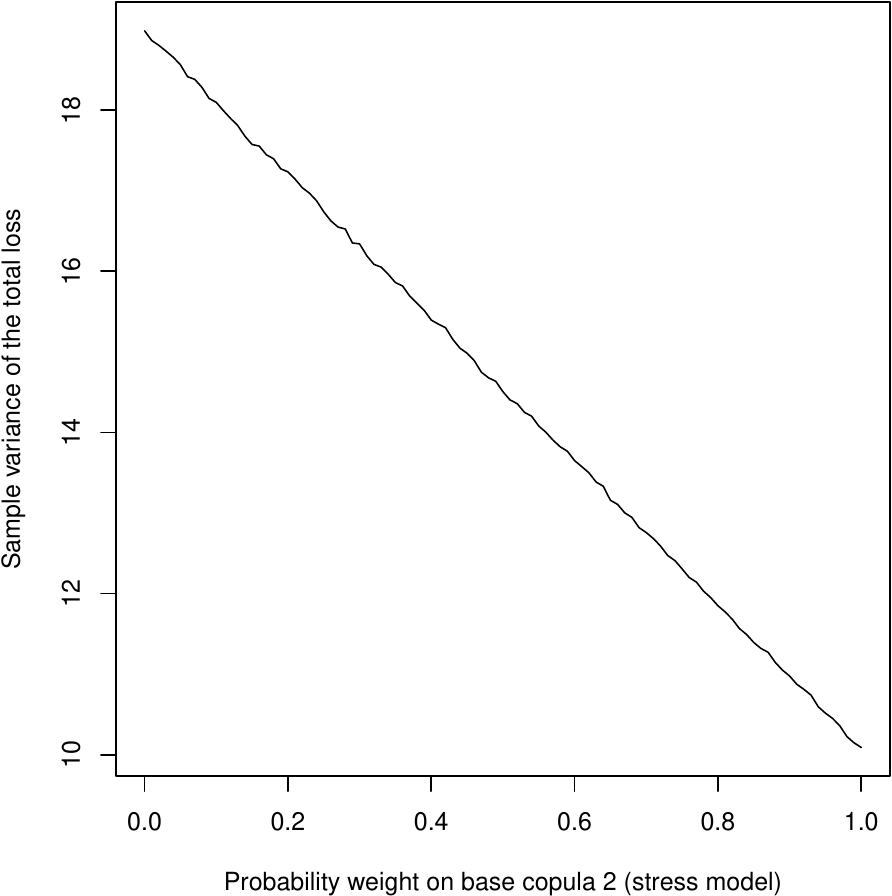}%
    \hfill
    \includegraphics[height=0.25\textheight]{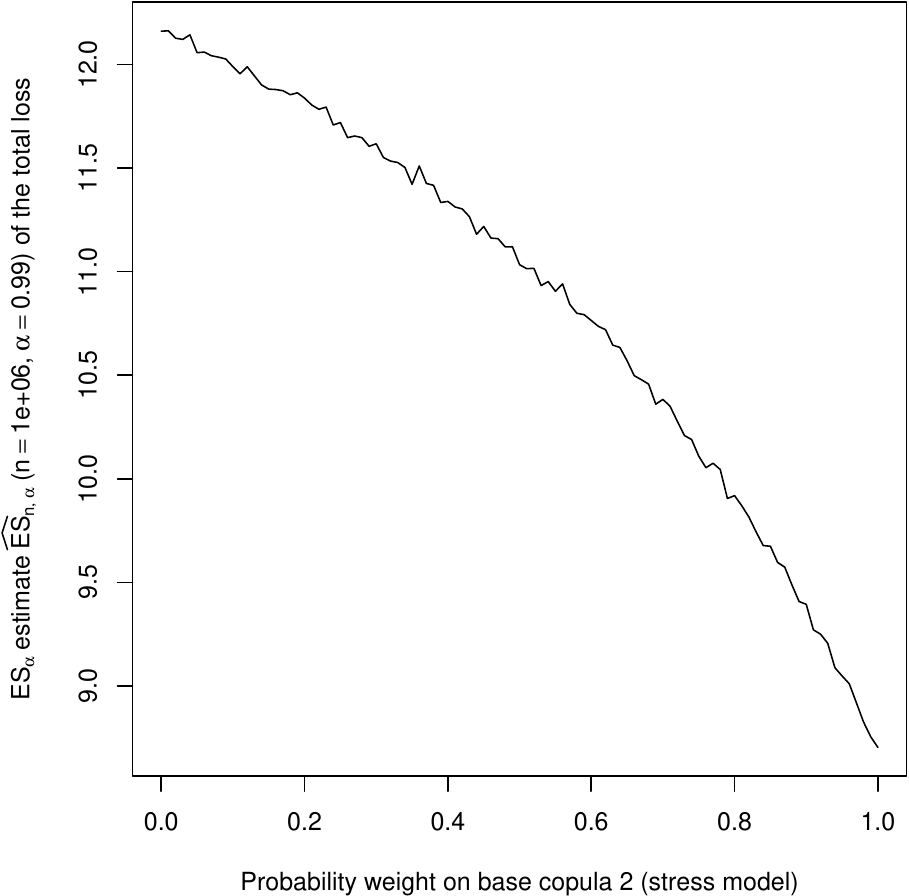}%
    \caption{Boxplots (left), sample variance (center) and nonparametric
      expected shortfall estimate at confidence level 99\% (right)
      of $10^6$ simulated total losses in five business lines of a financial firm under standard normal %
      margins and an index-mixed copula with probability weight $p$ (displayed on the x-axis)
      on the stress model $C_2(u_1,\dots,u_5)=C^{t_4,0.4}(u_1,u_2,u_3)C^{t_4,0.6}(u_4,u_5)$
      and remaining weight on the benchmark model $C_1(u_1,\dots,u_5)=C^{t_4,0.5}(u_1,\dots,u_5)$.}
    \label{fig:app:total:loss}
  \end{figure}
  Note that the mean is not affected by a change of the underlying dependence
  structure.  Looking closely, we may see a decrease in variance if more mass is
  put on the stress model. This is confirmed by a plot of the sample variance in
  the center of Figure~\ref{fig:app:total:loss}. The right-hand side of
  Figure~\ref{fig:app:total:loss} displays the nonparametric estimator of the
  expected shortfall of the total loss at confidence level 99\%; see
  \cite[Section~9.1]{mcneilfreyembrechts2015} for details on the estimation
  procedure.
We see that the more mass is put on the stress model,
the smaller the risk for the firm in terms of variance and expected shortfall.
\end{example}

\section{Conclusion}\label{sec:concl}
We introduced the class of index-mixed copulas and investigated its properties.
By construction, index-mixed copulas are a convex combination of products of
margins of the utilized base copulas and hence generalize traditional
multivariate mixture distributions.
Even in this more general and flexible setup, one can compute many important
copula properties analytically. Among others, we derived the distribution
function and density, sampling algorithms, bivariate and trivariate margins,
mixtures of index-mixed copulas, addressed symmetries such as radial symmetry
and exchangeability, measures of association such as tail dependence
coefficients, Blomqvist's beta, Spearman's rho or Kendall's tau, and orthant
dependence and concordance orderings.  We also provided examples with
illustrations and outlined applications in actuarial science to the analysis of sums of dependent random variables or stress testing of
dependence structures.
Similar to the discussed stress testing, index-mixed copulas allow for inclusion of expert opinions in the modeling process through the specification of base copulas and mixing weights in the construction.
Collectively, this shows the rather remarkable
tractability of index-mixed copulas. Through the lens of index-mixing we also
revealed the limited range of dependence captured by EFGM copulas and why they
are tail independent.

We note the following possible avenues for future research.

First, one could generalize index-mixed copulas by allowing the index matrix to
have more than a single $1$ in each row. For row $j$ of the copula matrix, the
resulting two or more standard uniforms selected by the index matrix would then
need to be transformed to the single standard uniform $U_j$. As transformations,
one could consider the minimum, maximum, or sum, for example.  In each case, one
would then need to map with the respective marginal distribution to guarantee
that $U_j\sim\U(0,1)$. For the minimum, maximum and sum, the distribution
functions are $1-(1-u)^{S_j}$, $u^{S_j}$ and the Irwin-Hall distribution
function
$\frac{1}{S_j!}\sum_{l=0}^{\lfloor
  u\rfloor}(-1)^l\binom{S_j}{l}(u-l)^{S_j}$, %
respectively, where $S_j=\sum_{k=1}^K I_{j, k}^{\text{mat}}$ is the number of $1$s in the
$j$th row of $I^{\text{mat}}$.

Second, while the article has focused specifically on copulas, the representation
in Theorem~\ref{thm:index:mixed:form:C} similarly applies to quasi-copulas,
which suggests the study of index-mixed quasi-copulas.

Third, there are a number of possible estimation procedures to be explored.
While a direct maximum likelihood estimation is at odds with the mixture
structure of index-mixed copulas, the EM-algorithm, see
\cite{McLachlanKrishnan2008}, composite likelihood methods, see
\cite{GoreckiHofert2023}, and maximum mean discrepancy based parameter
estimation, see \cite{AlquierCheriefAbdellatifDerumignyFermanian2023}, for
copulas provide a suitable framework in which estimation can be performed.

Fourth, and also in the realm of statistical inference, a visual assessment or
formal test of a copula being index-mixed would be of interest.

\subsection*{Acknowledgments}
This work was supported by NSERC grant RGPIN-2020-05784.

\subsection*{Conflict of interest statement}
The authors report there are no competing interests to declare.

\printbibliography[heading=bibintoc]

\appendix

\section{Proofs}\label{sec_proofs}

\begin{proof}[Proof of Theorem~\ref{thm:index:mixed:form:C}]
  With independent $\tilde{\bm{U}}_k=(\tilde{U}_{1,k},\dots,\tilde{U}_{d,k})\sim C_k$, $k=1,\dots,K$,
  we have
  \begin{align*}
    \P(\bm{U}\le\bm{u})&=\E_{\bm{I}}[\P(\tilde{U}_{\bm{I}}\le\bm{u}\,|\,\bm{I})]=\E_{\bm{I}}[\P(\tilde{U}_{j,I_j}\le u_j\ \forall\,j\,|\,\bm{I})]=\E_{\bm{I}}\biggl[\,\prod_{k=1}^KC_k(\bm{u}^{I_{\cdot,k}^{\text{mat}}})\biggr],\quad\bm{u}\in[0,1]^d,
  \end{align*}
  where, for $k=1,\dots,K$, the last equality follows from the independence of the base copulas and the fact that all
  entries $\tilde{U}_{j,I_j}$ with the same value $I_j=k$ follow the base copula
  $C_k$.
  With $u_l=1$ for all $l\neq j$ we have for $u_j\in[0,1]$ that
  \begin{align*}
    \P(U_j\le u_j)&=\E_{\bm{I}}\biggl[\,\prod_{k=1}^KC_k(1,\dots,1,u_j^{I_{j,k}^{\text{mat}}},1,\dots,1)\biggr]=\E_{\bm{I}}\biggl[\,\prod_{k=1}^K u_j^{I_{j,k}^{\text{mat}}}\biggr]=\E_{\bm{I}}\biggl[u_j^{\sum_{k=1}^K I_{j,k}^{\text{mat}}}\biggr]
    =\E_{\bm{I}}[u_j]=u_j,
  \end{align*}
  so $C$ is indeed a copula.
\end{proof}

\begin{proof}[Proof of Corollary~\ref{cor:index:mixed:form:c}]
  The statement follows immediately by noting that $\E_{\bm{I}}$ is a weighted
  sum with a finite number of summands, by linearity of differentiation, and by
  noting that each product $\prod_{k=1}^KC_k(\bm{u}^{I_{\cdot,k}^{\text{mat}}})$ contains
  each of $u_1,\dots,u_d$ precisely once.
\end{proof}

\begin{proof}[Proof of Corollary~\ref{cor:invar:equal:comon}]
  If $p_{\bm{I}}(k,\dots,k)\ge 0$, $\sum_{k=1}^Kp_{\bm{I}}(k,\dots,k)=1$, then,
  with probability $p_{\bm{I}}(k,\dots,k)$, the $k$th column $I_{\cdot,k}^{\text{mat}}$
  of the index matrix is $\bm{1}$ and $I_{\cdot,l}^{\text{mat}}=\bm{0}$ for all $l\neq k$. We thus have
  \begin{align*}
    C(\bm{u})=\E_{\bm{I}}\biggl[\,\prod_{k=1}^KC_k(\bm{u}^{I_{\cdot,k}^{\text{mat}}})\biggr]=\sum_{k=1}^Kp_{\bm{I}}(k,\dots,k)\biggl(\,C_k(\bm{u}^{\bm{1}})\prod_{\substack{l=1\\l\neq k}}^KC_l(\bm{u}^{\bm{0}})\biggr)=\sum_{k=1}^Kp_{\bm{I}}(k,\dots,k)C_k(\bm{u}),\quad \bm{u}\in[0,1]^d.
  \end{align*}
  The last statement follows as a special case.
\end{proof}

\begin{proof}[Proof of Lemma~\ref{lem:biv:mar}]
  To allow for a more compact notation we make use of the Kronecker delta
  $\delta_{ij}=\I_{\{i\}}(j)$ for two indices $i,j\in\IN$.
  \begin{enumerate}
  \item We have for all $u_{j_1},u_{j_2}\in[0,1]$ that
    \begin{align*}
      C^{j_1,j_2}(u_{j_1},u_{j_2})&=C(1,\dots,1,u_{j_1},1,\dots,1,u_{j_2},1,\dots,1)\\
                                  &=\E_{(I_{j_1},I_{j_2})}\biggl[\,\prod_{k=1}^KC_k(1,\dots,1,u_{j_1}^{I_{j_1,k}^{\text{mat}}},1,\dots,1,u_{j_2}^{I_{j_2,k}^{\text{mat}}},1,\dots,1)\biggr]\\
                                  &=\sum_{k_1,k_2=1}^Kp_{(I_{j_1},I_{j_2})}(k_1,k_2)\biggl(\,\prod_{k=1}^KC_k(1,\dots,1,u_{j_1}^{\delta_{k_1k}},1,\dots,1,u_{j_2}^{\delta_{k_2k}},1,\dots,1)\biggr)\\
                                  &=\sum_{k=1}^Kp_{(I_{j_1},I_{j_2})}(k,k)C_k^{j_1,j_2}(u_{j_1},u_{j_2})+\biggl(\,\sum_{\substack{k_1,k_2=1\\ k_1\neq k_2}}^Kp_{(I_{j_1},I_{j_2})}(k_1,k_2)\biggr)\Pi(u_{j_1},u_{j_2}).
    \end{align*}
    The remaining part of the statement ($K=2$) is immediate.
  \item We have for all $u_{j_1},u_{j_2},u_{j_3}\in[0,1]$ that
    \begin{align*}
      &\phantom{{}={}}C^{j_1,j_2,j_3}(u_{j_1},u_{j_2},u_{j_3})\\
      &=\sum_{k_1,k_2,k_3=1}^Kp_{(I_{j_1},I_{j_2},I_{j_3})}(k_1,k_2,k_3)\biggl(\,\prod_{k=1}^KC_k(1,\dots,1,u_{j_1}^{\delta_{k_1k}},1,\dots,1,u_{j_2}^{\delta_{k_2k}},1,\dots,1,u_{j_3}^{\delta_{k_3k}},1\dots,1)\biggr)\\ %
      &=\sum_{k=1}^K p_{(I_{j_1},I_{j_2},I_{j_3})}(k,k,k)C_k^{j_1,j_2,j_3}(u_{j_1},u_{j_2},u_{j_3})\\ %
      &\phantom{{}={}}+\sum_{\substack{k_1,k_2,k_3=1\\\text{two $k_l$ equal}}}^K p_{(I_{j_1},I_{j_2},I_{j_3})}(k_1,k_2,k_3)\biggl(\,\prod_{k=1}^KC_k(1,\dots,1,u_{j_1}^{\delta_{k_1k}},1,\dots,1,u_{j_2}^{\delta_{k_2k}},\\
      & \quad \qquad \qquad \qquad \qquad \qquad \qquad \qquad \qquad \qquad \qquad 1,\dots,1,u_{j_3}^{\delta_{k_3 k}},1\dots,1)\biggr)\\ %
      &\phantom{{}={}}+\biggl(\,\sum_{\substack{k_1,k_2,k_3=1\\\text{all $k_l$ distinct}}}^K p_{(I_{j_1},I_{j_2},I_{j_3})}(k_1,k_2,k_3)\biggr)\Pi(u_{j_1},u_{j_2},u_{j_3})\\ %
      &=\sum_{k=1}^K p_{(I_{j_1},I_{j_2},I_{j_3})}(k,k,k) C_k^{j_1,j_2,j_3}(u_{j_1},u_{j_2},u_{j_3})\\
      &\phantom{{}={}}+\sum_{k=1}^K\biggl(\biggl(\,\sum_{\substack{k_3=1\\k_3\neq k}}^K p_{(I_{j_1},I_{j_2},I_{j_3})}(k,k,k_3) \biggr)C_k^{j_1,j_2}(u_{j_1},u_{j_2})u_{j_3}\\
      &\phantom{{}={}+\sum_{k=1}^K\biggl(} + \biggl(\,\sum_{\substack{k_2=1\\k_2\neq k}}^Kp_{(I_{j_1},I_{j_2},I_{j_3})}(k,k_2,k) \biggr) C_k^{j_1,j_3}(u_{j_1},u_{j_3})u_{j_2}\\
      &\phantom{{}={}+\sum_{k=1}^K\biggl(}+ \biggl(\,\sum_{\substack{k_1=1\\k_1\neq k}}^Kp_{(I_{j_1},I_{j_2},I_{j_3})}(k_1,k,k) \biggr) C_k^{j_2,j_3}(u_{j_2},u_{j_3})u_{j_1}\biggr)\\
      &\phantom{{}={}}+\biggl(\,\sum_{\substack{k_1,k_2,k_3=1\\\text{all $k_l$ distinct}}}^K p_{(I_{j_1},I_{j_2},I_{j_3})}(k_1,k_2,k_3)\biggr)\Pi(u_{j_1},u_{j_2},u_{j_3}).
    \end{align*}
    The remaining part of the statement ($K=2$) is immediate. \qedhere
  \end{enumerate}
\end{proof}

\begin{proof}[Proof of Corollary~\ref{cor:index:mixed:mixtures}]
  By Tonelli's Theorem, we have
  \begin{align*}
    C(\bm{u})&=\E_{\bm{I}}\biggl[\,\prod_{k=1}^KC_k(\bm{u}^{I_{\cdot,k}^{\text{mat}}})\biggr]=\E_{\bm{I}}\biggl[\,\prod_{k=1}^K\int_{\IR}C_k(\bm{u}^{I_{\cdot,k}^{\text{mat}}};\theta_k)\,\rd G_k(\theta_k)\biggr]\\
             &=\E_{\bm{I}}\biggl[\,\int_{\IR}\dots\int_{\IR} \prod_{k=1}^KC_k(\bm{u}^{I_{\cdot,k}^{\text{mat}}};\theta_k)\,\rd G_1(\theta_1)\cdots\,\rd G_K(\theta_K)\biggr]\\
             &=\int_{\IR}\dots\int_{\IR} \E_{\bm{I}}\biggl[\,\prod_{k=1}^KC_k(\bm{u}^{I_{\cdot,k}^{\text{mat}}};\theta_k)\biggr]\,\rd G_1(\theta_1)\cdots\,\rd G_K(\theta_K),\quad\bm{u}\in[0,1]^d.
  \end{align*}
  The last statement follows as a special case.
\end{proof}

\begin{proof}[Proof of Proposition~\ref{prop:surv:cop}]
  Denote by $E\in\IR^{d\times K}$ a matrix where each entry is equal to $1$.  If
  $\bm{U}\sim C$, then $\bm{1}-\bm{U}\sim\hat{C}$.  By~\eqref{U:index:mixed},
  \begin{align*}
    \bm{1}-\bm{U}=\bm{1}-\tilde{U}_{\bm{I}}=\begin{pmatrix}1-\tilde{U}_{1,I_1}\\ \vdots \\ 1-\tilde{U}_{d,I_d}\end{pmatrix}=\begin{pmatrix}(E-\tilde{U})_{1,I_1}\\ \vdots \\ (E-\tilde{U})_{d,I_d}\end{pmatrix}=(E-\tilde{U})_{\bm{I}}.
  \end{align*}
  Since the columns of $E-\tilde{U}$ contain random vectors from the survival
  copulas of the base copulas, we obtain that the survival copula of the
  index-mixed copula $C$ is the index-mixed copula of the survival copulas of
  the base copulas. The statement about radial symmetry is rather immediate.
\end{proof}

\begin{proof}[Proof of Lemma~\ref{lem:exch:sums:prods:base:copulas}]
  Let $e_j^d=(0,\dots,0,1,0,\dots,0)\T$, $j=1,\dots,d$, denote the $j$th standard
  basis column vector of $\IR^d$ and let
  $P_{\sigma}:=(e_{\sigma(1)}^{d},\dots,e_{\sigma(d)}^{d})\T\in\{0,1\}^{d\times d}$ be the permutation matrix associated with
  $\sigma$ (stacking $e_{\sigma(1)}^d,\dots,e_{\sigma(d)}^d$ as rows of
  $P_{\sigma}$). Then multiplying $P_{\sigma}$ with a matrix
  $M\in\IR^{d\times m}$ from the right permutes the rows of $M$ in such a way
  that the $j$th row of $P_{\sigma}M$ is the $\sigma(j)$th row of $M$. With
  $\sigma(\bm{i}) = (i_{\sigma(1)},\ldots,i_{\sigma(d)})$ and using the
  permutation property of $P_{\sigma}$, we thus have
  \begin{align*}
    P_{\sigma}I^{\text{mat}}(\bm{i}) = \begin{pmatrix}{e_{i_{\sigma(1)}}^K}\!\!\T \\\vdots\\ {e_{i_{\sigma(d)}}^K}\!\!\T \end{pmatrix} = I^{\text{mat}}(\sigma(\bm{i})).
  \end{align*}
  This further implies that the $k$th column $I_{\cdot,k}^{\text{mat}}(\sigma(\bm{i}))$ of
  $I^{\text{mat}}(\sigma(\bm{i}))$ is a permutation of $I_{\cdot,k}^{\text{mat}}(\bm{i})$ since each
  column $k$ of $I^{\text{mat}}(\bm{i})$ is permuted according to $\sigma$.  Therefore,
  \begin{align*}
    I_{\cdot,k}^{\text{mat}}(\sigma(\bm{i})) = \sigma(I_{\cdot,k}^{\text{mat}}(\bm{i})).
  \end{align*}
  Noting that, for any permutation $\sigma\in\Sigma_d$, we have
  $\sort{\sigma(\bm{i})} = \sort{\bm{i}}$, that is ordering annihilates any previous
  permutation, we clearly have
  $\bm{i}\in \calK_{\bm{k}_{\text{ord}}}^d$ if and only if
  $\sigma(\bm{i}) \in \calK_{\bm{k}_{\text{ord}}}^d$.  This allows us to
  rewrite $\sum_{\bm{i}\in \calK_{\bm{k}_{\text{ord}}}^d}
    \prod_{k=1}^K C_k\bigl(\sigma(\bm{u})^{I_{\cdot, k}^{\text{mat}}(\bm{i})}\bigr)$ as
  \begin{align*}
    \sum_{\bm{i}\in \calK_{\bm{k}_{\text{ord}}}^d} \prod_{k=1}^K C_k\bigl(\sigma(\bm{u})^{I_{\cdot, k}^{\text{mat}}(\bm{i})}\bigr)
    &=
    \sum_{\bm{i}\in \calK_{\bm{k}_{\text{ord}}}^d} \prod_{k=1}^K C_k\bigl(\sigma(\bm{u})^{I_{\cdot, k}^{\text{mat}}(\sigma(\bm{i}))}\bigr)
    =
    \sum_{\bm{i}\in \calK_{\bm{k}_{\text{ord}}}^d} \prod_{k=1}^K C_k\bigl(\sigma(\bm{u})^{\sigma(I_{\cdot, k}^{\text{mat}}(\bm{i}))}\bigr).
  \end{align*}
  Writing out the argument
  $\sigma(\bm{u})^{\sigma(I_{\cdot, k}^{\text{mat}}(\bm{i}))} =
  \bigl(u_{\sigma(1)}^{I_{\sigma(1),k}^{\text{mat}}(\bm{i})},\ldots,u_{\sigma(d)}^{I_{\sigma(d),k}^{\text{mat}}(\bm{i})}\bigr)$ in the last term,
  we see that exchangeability of $C_k$ can be used to reorder the arguments in
  such a way that $C_k\bigl(\sigma(\bm{u})^{\sigma(I_{\cdot, k}^{\text{mat}}(\bm{i}))}\bigr) = C_k\bigl(\bm{u}^{I_{\cdot, k}^{\text{mat}}(\bm{i})}\bigr)$,
  which then implies that
  \begin{align*}
    \sum_{\bm{i}\in \calK_{\bm{k}_{\text{ord}}}^d} \prod_{k=1}^K C_k\bigl(\sigma(\bm{u})^{I_{\cdot, k}^{\text{mat}}(\bm{i})}\bigr)
    &=
    \sum_{\bm{i}\in \calK_{\bm{k}_{\text{ord}}}^d} \prod_{k=1}^K C_k\bigl(\bm{u}^{I_{\cdot, k}^{\text{mat}}(\bm{i})}\bigr).\qedhere
  \end{align*}
\end{proof}

\begin{proof}[Proof of Proposition~\ref{prop:exchangeable}]\mbox{}
  \begin{enumerate}
  \item Recall that
    $C(\bm{u})=\E_{\bm{I}}\bigl[\,\prod_{k=1}^KC_k(\bm{u}^{I_{\cdot, k}^{\text{mat}}})\bigr]$.
    Fix a specific index vector $\bm{I}$ with corresponding index matrix $I^{\text{mat}}$ and
    index partition $\{J_1,\dots,J_K\}$; see
    Section~\ref{sec:index:mixing}. For each $k=1,\dots,K$,
    exchangeability of $C_k$ allows us to reorder the arguments of
    $C_k(\bm{u}^{I_{\cdot, k}^{\text{mat}}})$ and write this term as
    $C_k(\bm{u}_{J_k},1,\dots,1)$ leading to the first
    expression. %
  \item We first note that the sets
    $\calK_{\bm{k}_{\text{ord}}}^d$,
    $\bm{k}_{\text{ord}}\in\mathcal{K}_{\text{ord}}^d$, by definition provide a
    partition of $\mathcal{K}^d$. As such we can group index vectors and rewrite
    the expectation with respect to the index distribution in the definition of
    index-mixed copulas. Denoting by $I_{\cdot,k}^{\text{mat}}(\bm{i})$ the $k$th column of the
    realization of the random index matrix $I^{\text{mat}}$ for $\bm{I}=\bm{i}$,
    we obtain
    \begin{align}\label{eq:exchangeability:reformulate}
      C(\bm{u})
      &= \E_{\bm{I}}\biggl[\,\prod_{k=1}^KC_k(\bm{u}^{I_{\cdot, k}^{\text{mat}}})\biggr]= \sum_{\bm{i}\in\mathcal{K}^d} p_{\bm{I}}(\bm{i}) \prod_{k=1}^KC_k(\bm{u}^{I_{\cdot, k}^{\text{mat}}(\bm{i})})
        = \sum_{\bm{k}_{\text{ord}}\in\mathcal{K}_{\text{ord}}^d}\ \sum_{\bm{i}\in \calK_{\bm{k}_{\text{ord}}}^d}\!\!\!\!\!p_{\bm{I}}(\bm{i}) \prod_{k=1}^KC_k(\bm{u}^{I_{\cdot, k}^{\text{mat}}(\bm{i})})\notag\\&= \sum_{\bm{k}_{\text{ord}}\in\mathcal{K}_{\text{ord}}^d}\!\!\! p_{\bm{I},\bm{k}_{\text{ord}}}\!\!\!\sum_{\bm{i}\in \calK_{\bm{k}_{\text{ord}}}^d} \prod_{k=1}^KC_k(\bm{u}^{I_{\cdot, k}^{\text{mat}}(\bm{i})}),
    \end{align}
    where
    we made use of the fact that
    $\bm{I}$ is uniformly distributed over
    $\calK_{\bm{k}_{\text{ord}}}^d$ with probability
    $p_{\bm{I}}(\bm{i})=p_{\bm{I},\bm{k}_{\text{ord}}}$ for all
    $\bm{i}\in\calK_{\bm{k}_{\text{ord}}}^d$. To show that $C$ is
    exchangeable, consider now a permutation $\sigma\in\Sigma_d$.
    From \eqref{eq:exchangeability:reformulate} we then have that
    \begin{align*}
      C(\sigma(\bm{u}))
      &= \sum_{\bm{k}_{\text{ord}}\in\mathcal{K}_{\text{ord}}^d} p_{\bm{I},\bm{k}_{\text{ord}}} \sum_{\bm{i}\in \calK_{\bm{k}_{\text{ord}}}^d} \prod_{k=1}^K C_k\bigl(\sigma(\bm{u})^{I_{\cdot, k}^{\text{mat}}(\bm{i})}\bigr).
    \end{align*}
    By Lemma~\ref{lem:exch:sums:prods:base:copulas},
    \begin{align*}
      \sum_{\bm{i}\in \calK_{\bm{k}_{\text{ord}}}^d}
      \prod_{k=1}^K C_k\bigl(\sigma(\bm{u})^{I_{\cdot, k}^{\text{mat}}(\bm{i})}\bigr)
      =
      \sum_{\bm{i}\in \calK_{\bm{k}_{\text{ord}}}^d}
      \prod_{k=1}^K C_k\bigl(\bm{u}^{I_{\cdot, k}^{\text{mat}}(\bm{i})}\bigr)
    \end{align*}
    and thus
    \begin{align*}
      C(\sigma(\bm{u}))=\sum_{\bm{k}_{\text{ord}}\in\mathcal{K}_{\text{ord}}^d} p_{\bm{I},\bm{k}_{\text{ord}}}\sum_{\bm{i}\in \calK_{\bm{k}_{\text{ord}}}^d}\prod_{k=1}^K C_k\bigl(\bm{u}^{I_{\cdot, k}^{\text{mat}}(\bm{i})}\bigr),
    \end{align*}
    which, by \eqref{eq:exchangeability:reformulate}, is $C(\bm{u})$,
    so $C$ is exchangeable.
  \end{enumerate}
\end{proof}

\begin{proof}[Proof of Proposition~\ref{prop:tail:dependence}]
  By Lemma~\ref{lem:biv:mar} Part~\ref{lem:biv:mar:1},
  \begin{align*}
    \lambda_{\text{l}}^{j_1,j_2}&=\lim_{u\downarrow 0}\frac{C^{j_1,j_2}(u,u)}{u}
    =\lim_{u\downarrow 0}\frac{\sum_{k=1}^Kp_{(I_{j_1},I_{j_2})}(k,k)C_k^{j_1,j_2}(u,u)+\bigl(1-\sum_{k=1}^Kp_{(I_{j_1},I_{j_2})}(k,k)\bigr)\Pi(u,u)}{u}\\
                       &=\sum_{k=1}^Kp_{(I_{j_1},I_{j_2})}(k,k)\lim_{u\downarrow 0}\frac{C_k^{j_1,j_2}(u,u)}{u}=\sum_{k=1}^K p_{(I_{j_1},I_{j_2})}(k,k)\lambda_{\text{l},k}^{j_1,j_2}.
  \end{align*}
  Combining Lemma~\ref{lem:biv:mar} Part~\ref{lem:biv:mar:1} and Proposition~\ref{prop:surv:cop}, %
  we obtain
  \begin{align*}
    \lambda_{\text{u}}^{j_1,j_2}&=\lim_{u\uparrow 1}\frac{1-2u+C^{j_1,j_2}(u,u)}{1-u}=\lim_{u\downarrow 0}\frac{-1+2u+C^{j_1,j_2}(1-u,1-u)}{u}
    =\lim_{u\downarrow 0}\frac{\hat{C}^{j_1,j_2}(u,u)}{u}\\
    &=\lim_{u\downarrow 0}\frac{\sum_{k=1}^Kp_{(I_{j_1},I_{j_2})}(k,k)\hat{C}_k^{j_1,j_2}(u,u)+\bigl(1-\sum_{k=1}^Kp_{(I_{j_1},I_{j_2})}(k,k)\bigr)\Pi(u,u)}{u}\\
    &=\sum_{k=1}^Kp_{(I_{j_1},I_{j_2})}(k,k)\lim_{u\downarrow 0}\frac{\hat{C}_k^{j_1,j_2}(u,u)}{u}=\sum_{k=1}^K p_{(I_{j_1},I_{j_2})}(k,k)\lambda_{\text{u},k}^{j_1,j_2}.\qedhere
  \end{align*}
\end{proof}

\begin{proof}[Proof of Proposition~\ref{prop_bivariate_Spearman_Kendall}]
  By~\eqref{eq:biv:mar:compact},
  \begin{align*}
    C^{j_1,j_2}(u_{j_1},u_{j_2})=\sum_{k=1}^{K+1}p_{(I_{j_1},I_{j_2})}(k,k)C_k^{j_1,j_2}(u_{j_1},u_{j_2})
  \end{align*}
with $\sum_{k=1}^{K+1}p_{(I_{j_1},I_{j_2})}(k,k) = 1$.
  For simplicity of the arguments, we extend all notions to $K+1$, so $\rho_{\text{S},K+1}^{j_1,j_2}$
  is Spearman's rho for $C_{K+1}^{j_1,j_2}$ and simmilarly for Kendall's tau later on.
  With $V_1,V_2\isim\U(0,1)$, Spearman's rho of $C^{j_1,j_2}$ is
  \begin{align*}
    \rho_{\text{S}}^{j_1,j_2}&=12\E[C^{j_1,j_2}(V_1,V_2)]-3=12\E\biggl[\,\sum_{k=1}^{K+1}p_{(I_{j_1},I_{j_2})}(k,k)C_k^{j_1,j_2}(V_1,V_2)\biggr]-3\\
    &=\sum_{k=1}^{K+1}p_{(I_{j_1},I_{j_2})}(k,k)\bigl(12\E[C_k^{j_1,j_2}(V_1,V_2)]-3\bigr)=\sum_{k=1}^{K+1} p_{(I_{j_1},I_{j_2})}(k,k)\rho_{\text{S},k}^{j_1,j_2}=\sum_{k=1}^K p_{(I_{j_1},I_{j_2})}(k,k)\rho_{\text{S},k}^{j_1,j_2},
  \end{align*}
  where we used that Spearman's rho of $C_{K+1}^{j_1,j_2}=\Pi$ is $0$.

Now consider Kendall's tau. With $(V_1,V_2)\sim C^{j_1,j_2}$ we have with \eqref{eq:mu:mixture} that
\begin{align*}
\E[C^{j_1,j_2}(V_1,V_2)] = \mu_{C^{j_1,j_2},C^{j_1,j_2}} = \sum^{K+1}_{k=1} \sum^{K+1}_{l=1} p_{(I_{j_1},I_{j_2})}(k,k) p_{(I_{j_1},I_{j_2})}(l,l) \mu_{C_k^{j_1,j_2},C_l^{j_1,j_2}},
\end{align*}
and thus Kendall's tau of $C^{j_1,j_2}$ is
\begin{align*}
\tau^{j_1,j_2}
  &=4\E[C^{j_1,j_2}(V_1,V_2)]-1\\
  &=4\sum^{K+1}_{k=1} \sum^{K+1}_{l=1} p_{(I_{j_1},I_{j_2})}(k,k) p_{(I_{j_1},I_{j_2})}(l,l) \mu_{C_k^{j_1,j_2},C_l^{j_1,j_2}} - \sum^{K+1}_{k=1} \sum^{K+1}_{l=1} p_{(I_{j_1},I_{j_2})}(k,k) p_{(I_{j_1},I_{j_2})}(l,l)\\
  &=\sum^{K+1}_{k=1} \sum^{K+1}_{l=1} p_{(I_{j_1},I_{j_2})}(k,k) p_{(I_{j_1},I_{j_2})}(l,l) \left(4 \mu_{C_k^{j_1,j_2},C_l^{j_1,j_2}} - 1\right)\\
  &=\sum_{k=1}^{K+1}\sum_{l=1}^{K+1}p_{(I_{j_1},I_{j_2})}(k,k)p_{(I_{j_1},I_{j_2})}(l,l)\tau_{k,l}^{j_1,j_2}\\
  &=\sum_{k=1}^{K+1}(p_{(I_{j_1},I_{j_2})}(k,k))^2\tau_k^{j_1,j_2}+\sum_{\substack{k_1,k_2=1\\ k_1\neq k_2}}^{K+1}p_{(I_{j_1},I_{j_2})}(k_1,k_1) p_{(I_{j_1},I_{j_2})}(k_2,k_2) \tau_{k_1,k_2}^{j_1,j_2}.
\end{align*}
  As Kendall's tau of $C_{K+1}^{j_1,j_2}=\Pi$ is $0$, the first sum is $\sum_{k=1}^K(p_{(I_{j_1},I_{j_2})}(k,k))^2\tau_k^{j_1,j_2}$.
  By~\eqref{eq:switch:bivariate:integrand:integrator}, $\tau_{k,l}^{j_1,j_2}=\tau_{l,k}^{j_1,j_2}$ and thus
  \begin{align*}
    \tau^{j_1,j_2}&=\sum_{k=1}^K(p_{(I_{j_1},I_{j_2})}(k,k))^2\tau_k^{j_1,j_2}+2\!\!\!\!\sum_{1\le k_1<k_2\le K+1}\!\!\!\! p_{(I_{j_1},I_{j_2})}(k_1,k_1) p_{(I_{j_1},I_{j_2})}(k_2,k_2) \tau_{k_1,k_2}^{j_1,j_2}\\
                  &=\sum_{k=1}^K(p_{(I_{j_1},I_{j_2})}(k,k))^2\tau_k^{j_1,j_2}+2\!\!\!\!\sum_{1\le k_1<k_2\le K}\!\!\!\! p_{(I_{j_1},I_{j_2})}(k_1,k_1)p_{(I_{j_1},I_{j_2})}(k_2,k_2)\tau_{k_1,k_2}^{j_1,j_2}\\
                  &\phantom{={}}+2\sum_{k=1}^Kp_{(I_{j_1},I_{j_2})}(k,k) p_{(I_{j_1},I_{j_2})}(K+1,K+1)\tau_{k,K+1}^{j_1,j_2}.
  \end{align*}
  Since $p_{(I_{j_1},I_{j_2})}(K+1,K+1)=1-\sum_{k=1}^Kp_{(I_{j_1},I_{j_2})}(k,k)$ and $\tau_{k,K+1}^{j_1,j_2}=4\mu_{C_k^{j_1,j_2},\Pi}-1=4((\rho_{\text{S},k}^{j_1,j_2}+3)/12)-1=\rho_{\text{S},k}^{j_1,j_2}/3$, $k=1,\dots,K$, we obtain
  \begin{align*}
    \tau^{j_1,j_2}&=\sum_{k=1}^K(p_{(I_{j_1},I_{j_2})}(k,k))^2\tau_k^{j_1,j_2}+2\!\!\!\!\sum_{1\le k_1<k_2\le K}\!\!\!\! p_{(I_{j_1},I_{j_2})}(k_1,k_1) p_{(I_{j_1},I_{j_2})}(k_2,k_2)\tau_{k_1,k_2}^{j_1,j_2}\\
                  &\phantom{={}}+\frac{2}{3}\biggl(1-\sum_{k=1}^Kp_{(I_{j_1},I_{j_2})}(k,k)\biggr)\sum_{k=1}^Kp_{(I_{j_1},I_{j_2})}(k,k)\rho_{\text{S},k}^{j_1,j_2}.\qedhere
  \end{align*}
\end{proof}

\begin{proof}[Proof of Proposition~\ref{prop:mult:spearman}]
  Let $C(\bm{u})=\E_{\bm{I}}\bigl[\prod_{k=1}^KC_k(\bm{u}^{I_{\cdot,k}^{\text{mat}}})\bigr]$. For $\rho_{\text{S}}^{\text{l},d}=\frac{d+1}{2^d-(d+1)}\bigl(2^d\mu_{C,\Pi}-1\bigr)$, we have
  \begin{align*}
    \mu_{C,\Pi}=\E_{\bm{I}}\biggl[\,\int_{[0,1]^d}\prod_{k=1}^KC_k(\bm{u}^{I_{\cdot,k}^{\text{mat}}})\,\rd\Pi(\bm{u})\biggr]=\E_{\bm{I}}\biggl[\,\prod_{k=1}^K\int_{[0,1]^{D_k}}\!\!\! C_k(\bm{u}^{I_{\cdot,k}^{\text{mat}}})\,\rd\Pi(\bm{u}^{I_{\cdot,k}^{\text{mat}}})\biggr], %
  \end{align*}
  where the last equality holds since $I^{\text{mat}}$ contains in each row precisely one entry $1$,
  so each $u_j$ appears in only one factor of the product $\prod_{k=1}^K$ (so for only one index $k$);
  see also Remark~\ref{rem:constr} in this regard.
  Since $d=\sum_{k=1}^KD_k$, we obtain
  \begin{align*}
    2^d\mu_{C,\Pi}&=\E_{\bm{I}}\biggl[\,\prod_{k=1}^K\biggl(2^{D_k}\int_{[0,1]^{D_k}}\!\!\! C_k(\bm{u}^{I_{\cdot,k}^{\text{mat}}})\,\rd\Pi(\bm{u}^{I_{\cdot,k}^{\text{mat}}})\biggr)\biggr]\\
    &=\E_{\bm{I}}\biggl[\,\prod_{k=1}^K\biggl(\frac{2^{D_k}-(D_k+1)}{D_k+1}\rho_{\text{S}}^{\text{l},D_k}(C_k^{I_{\cdot,k}^{\text{mat}}})+1\biggr)\biggr]
  \end{align*}
  and thus
  \begin{align*}
    \rho_{\text{S}}^{\text{l},d}=\frac{d+1}{2^d-(d+1)}\biggl(\E_{\bm{I}}\biggl[\,\prod_{k=1}^K\biggl(\frac{2^{D_k}-(D_k+1)}{D_k+1}\rho_{\text{S}}^{\text{l},D_k}(C_k^{I_{\cdot,k}^{\text{mat}}})+1\biggr)\biggr]-1\biggr).
  \end{align*}
  For $\rho_{\text{S}}^{\text{u},d}=\frac{d+1}{2^d-(d+1)}\bigl(2^d\int_{[0,1]^d}\Pi(\bm{u})\,\rd
  C(\bm{u})-1\bigr)$, noting that the expectation in the definition of $C$ is a finite sum combined with linearity with respect to the integrator %
  implies that
  \begin{align*}
    \mu_{\Pi,C}&= \int_{[0,1]^d}\Pi(\bm{u})\,\rd \E_{\bm{I}}\biggl[\biggl(\,\prod_{k=1}^KC_k(\bm{u}^{I_{\cdot,k}^{\text{mat}}})\biggr)\biggr]
= \E_{\bm{I}}\biggl[\,\int_{[0,1]^d}\Pi(\bm{u})\,\rd \biggl(\,\prod_{k=1}^KC_k(\bm{u}^{I_{\cdot,k}^{\text{mat}}})\biggr)\biggr].
  \end{align*}
  For fixed $\bm{I}$, the integrator is a product of (marginal) copulas with
  non-overlapping components where each $u_j$, $j=1,\dots,d$, appears in precisely
  one factor; once more, see Remark~\ref{rem:constr}. Therefore,
  \begin{align*}
    \mu_{\Pi,C}=\E_{\bm{I}}\biggl[\,\int_{[0,1]^{D_K}}\dots\int_{[0,1]^{D_1}}\Pi(\bm{u})\,\biggl(\,\prod_{k=1}^K\rd C_k(\bm{u}^{I_{\cdot,k}^{\text{mat}}})\biggr)\biggr].
  \end{align*}
  With the same idea (based on Remark~\ref{rem:constr}), we have $\Pi(\bm{u})=\prod_{k=1}^K\Pi(\bm{u}^{I_{\cdot,k}^{\text{mat}}})$ and thus
  \begin{align*}
    \mu_{\Pi,C}&=\E_{\bm{I}}\biggl[\,\int_{[0,1]^{D_K}}\dots\int_{[0,1]^{D_1}}\prod_{k=1}^K\Pi(\bm{u}^{I_{\cdot,k}^{\text{mat}}})\,\biggl(\,\prod_{k=1}^K\rd C_k(\bm{u}^{I_{\cdot,k}^{\text{mat}}})\biggr)\biggr]
=\E_{\bm{I}}\biggl[\,\prod_{k=1}^K\int_{[0,1]^{D_k}}\Pi(\bm{u}^{I_{\cdot,k}^{\text{mat}}})\,\rd C_k(\bm{u}^{I_{\cdot,k}^{\text{mat}}})\biggr].
  \end{align*}
  Since $d=\sum_{k=1}^KD_k$, we then obtain
  \begin{align*}
    2^d\mu_{\Pi,C}&=\E_{\bm{I}}\biggl[\,\prod_{k=1}^K\biggl(2^{D_k}\int_{[0,1]^{D_k}}\Pi(\bm{u}^{I_{\cdot,k}^{\text{mat}}})\,\rd C_k(\bm{u}^{I_{\cdot,k}^{\text{mat}}})\biggr)\biggr]
=\E_{\bm{I}}\biggl[\,\prod_{k=1}^K\biggl(\frac{2^{D_k}-(D_k+1)}{D_k+1}\rho_{\text{S}}^{\text{u},D_k}(C_k^{I_{\cdot,k}^{\text{mat}}})+1\biggr)\biggr]
  \end{align*}
  and therefore
  \begin{align*}
    \rho_{\text{S}}^{\text{u},d}&=\frac{d+1}{2^d-(d+1)}\biggl(\E_{\bm{I}}\biggl[\,\prod_{k=1}^K\biggl(\frac{2^{D_k}-(D_k+1)}{D_k+1}\rho_{\text{S}}^{\text{u},D_k}(C_k^{I_{\cdot,k}^{\text{mat}}})+1\biggr)\biggr]-1\biggr).\qedhere
  \end{align*}
\end{proof}

\begin{proof}[Proof of Proposition~\ref{prop:orthant:dep}]
  Recall $C(\bm{u})=\E_{\bm{I}}\bigl[\,\prod_{k=1}^KC_k(\bm{u}^{I_{\cdot, k}^{\text{mat}}})\bigr]$, and
  $\hat{C}(\bm{u})=\E_{\bm{I}}\bigl[\,\prod_{k=1}^K\hat{C}_k(\bm{u}^{I_{\cdot,k}^{\text{mat}}})\bigr]$.
  \begin{enumerate}
  \item By Remark~\ref{rem:constr} we have for $\bm{u}\in[0,1]^d$ that
    \begin{align*}
      C(\bm{u})&=\E_{\bm{I}}\biggl[\,\prod_{k=1}^KC_k(\bm{u}^{I_{\cdot, k}^{\text{mat}}})\biggr]\ge \E_{\bm{I}}\biggl[\,\prod_{k=1}^K\Pi(\bm{u}^{I_{\cdot, k}^{\text{mat}}})\biggr]=\E_{\bm{I}}\biggl[\,\prod_{k=1}^K\prod_{j=1}^du_j^{I_{j, k}^{\text{mat}}}\biggr]\\
      &=\E_{\bm{I}}\biggl[\,\prod_{j=1}^d\prod_{k=1}^Ku_j^{I_{j, k}^{\text{mat}}}\biggr]=\E_{\bm{I}}\biggl[\,\prod_{j=1}^du_j^{\sum_{k=1}^K I_{j, k}^{\text{mat}}}\biggr]=\E_{\bm{I}}\biggl[\,\prod_{j=1}^du_j\biggr]=\Pi(\bm{u}).
    \end{align*}
  \item One has similarly to Part~\ref{prop:orthant:dep:1} that $\hat{C}(\bm{u}) \geq \Pi(\bm{u})$ for all $\bm{u}\in[0,1]^d$. Therefore, $\bar{C}(\bm{u}) = \hat{C}(\bm{1}-\bm{u}) \geq \Pi(\bm{1}-\bm{u}) = \bar{\Pi}(\bm{u})$, i..e, $C$ is PUOD.
  \item Follows from Parts~\ref{prop:orthant:dep:1} and \ref{prop:orthant:dep:2}. \qedhere
  \end{enumerate}
\end{proof}

\begin{proof}[Proof of Proposition~\ref{prop_concordance_ordering}]
  Recall $C(\bm{u})=\E_{\bm{I}}\bigl[\,\prod_{k=1}^KC_k(\bm{u}^{I_{\cdot, k}^{\text{mat}}})\bigr]$, and
  $\hat{C}(\bm{u})=\E_{\bm{I}}\bigl[\,\prod_{k=1}^K\hat{C}_k(\bm{u}^{I_{\cdot,k}^{\text{mat}}})\bigr]$.
  \begin{enumerate}
  \item This statement holds since
    \begin{align*}
      C(\bm{u})=\E_{\bm{I}}\biggl[\,\prod_{k=1}^KC_k(\bm{u}^{I_{\cdot, k}^{\text{mat}}})\biggr]&\le \E_{\bm{I}}\biggl[\,\prod_{k=1}^KC_k'(\bm{u}^{I_{\cdot, k}^{\text{mat}}})\biggr]
=\E_{\bm{I}'}\biggl[\,\prod_{k=1}^KC_k'(\bm{u}^{I_{\cdot, k}^{\text{mat}}})\biggr]=C'(\bm{u})
    \end{align*}
    for all $\bm{u}\in[0,1]^d$.
  \item Note first that $\forall \bm{u}\in[0,1]^d\colon\bar{C}(\bm{u}) \leq \bar{C}'(\bm{u}) \Leftrightarrow \forall \bm{u}\in[0,1]^d\colon\hat{C}(\bm{u}) \leq \hat{C}'(\bm{u})$, where the second condition follows similarly to Part~\ref{prop:concordance:same:I:1}.
  \item Follows from Parts~\ref{prop:concordance:same:I:1} and \ref{prop:concordance:same:I:2}. \qedhere
  \end{enumerate}
\end{proof}

\begin{proof}[Proof of Theorem~\ref{thm:LS:sum}]
  By equation~\eqref{eq:joint:df},
  $H(\bm{x})=\E_{\bm{I}}\bigl[\,\prod_{k=1}^KH_k^{I_{\cdot,k}^{\text{mat}}}(\bm{x}_{I_{\cdot,k}^{\text{mat}}})\bigr]$,
  $\bm{x}\in\IR^d$. Proceeding similarly as in the proof of Proposition~\ref{prop:mult:spearman}, we obtain
  \begin{align*}
    \LS[H](\bm{t})&=\int_{[0,\infty)^d}\!\!\!\!e^{-\bm{t}\T\bm{x}}\,\rd H(\bm{x})=\E_{\bm{I}}\biggl[\,\int_{[0,\infty)^d} e^{-\bm{t}\T\bm{x}}\,\rd\biggl(\,\prod_{k=1}^KH_k^{I_{\cdot,k}^{\text{mat}}}(\bm{x}_{I_{\cdot,k}^{\text{mat}}})\biggr)\biggr]\\
                  &=\E_{\bm{I}}\biggl[\,\int_{[0,\infty)^{D_K}}\dots\int_{[0,\infty)^{D_1}} e^{-\bm{t}\T\bm{x}}\biggl(\,\prod_{k=1}^K\rd H_k^{I_{\cdot,k}^{\text{mat}}}(\bm{x}_{I_{\cdot,k}^{\text{mat}}})\biggr)\biggr]\\
                  &=\E_{\bm{I}}\biggl[\,\int_{[0,\infty)^{D_K}}\dots\int_{[0,\infty)^{D_1}} \prod_{k=1}^K e^{-\bm{t}_{I_{\cdot,k}^{\text{mat}}}\T\bm{x}_{I_{\cdot,k}^{\text{mat}}}}\biggl(\,\prod_{k=1}^K\rd H_k^{I_{\cdot,k}^{\text{mat}}}(\bm{x}_{I_{\cdot,k}^{\text{mat}}})\biggr)\biggr]\\
    &=\E_{\bm{I}}\biggl[\,\prod_{k=1}^K\int_{[0,\infty)^{D_k}} e^{-\bm{t}_{I_{\cdot,k}^{\text{mat}}}\T\bm{x}_{I_{\cdot,k}^{\text{mat}}}}\,\rd H_k^{I_{\cdot,k}^{\text{mat}}}(\bm{x}_{I_{\cdot,k}^{\text{mat}}})\biggr]\\
    &=\E_{\bm{I}}\biggl[\,\prod_{k=1}^K\LS[H_k^{I_{\cdot,k}^{\text{mat}}}](\bm{t}_{I_{\cdot,k}^{\text{mat}}})\biggr]=\E_{\bm{I}}\biggl[\,\prod_{k=1}^K\LS[H_k](\bm{t}I_{\cdot,k}^{\text{mat}})\biggr],\quad\bm{t}\in[0,\infty)^d.
  \end{align*}
  The statement about $\LS[F_{S_n}](t)$ is immediate.
\end{proof}

\begin{proof}[Proof of Proposition~\ref{prop:K:ge:d}]
  Consider the first statement.
  Since almost surely the index vector $\bm{I}$ satisfies $I_{j_1}\neq I_{j_2}$
  for all $1\le j_1<j_2\le d$, the index matrix $I^{\text{mat}}$ almost surely consists of
  precisely $d$ columns $k_1,\dots,k_d$ with the $l$th of these columns
  containing a unique row $j_l$ with an entry $1$, while all other entries in this
  column $k_l$ are $0$. Therefore
  \begin{align*}
    C(\bm{u})&=\E_{\bm{I}}\biggl[\,\prod_{k=1}^KC_k(\bm{u}^{I_{\cdot, k}^{\text{mat}}})\biggr]=\E_{\bm{I}}\biggl[\,\prod_{l=1}^dC_{k_l}(1,\dots,1,u_{j_l},1,\dots,1)\biggr]=\E_{\bm{I}}\biggl[\,\prod_{l=1}^du_{j_l}\biggr]
=\E_{\bm{I}}\bigl[\,\Pi(\bm{u})\bigr]=\Pi(\bm{u}).\qedhere
  \end{align*}
\end{proof}

\begin{proof}[Proof of Proposition~\ref{prop_block_wise_independent}]
  Consider the representation of $C$ in \eqref{eq:index:mixed:cop} of
  Theorem~\ref{thm:index:mixed:form:C}.  For a given $\bm{k}\in\{1,\dots,K\}^d$
  we have for each summand with positive probability in
  \eqref{eq:index:mixed:cop} the form
  \begin{align*}
    p_{\bm{I}}(\bm{k})\prod_{k=1}^KC_k(\bm{u}^{I_{\cdot, k}^{\text{mat}}}) = C_{\ell}(\bm{u}^{I_{\cdot, \ell}^{\text{mat}}}) p_{\bm{I}}(\bm{k})\prod_{\stackrel{k=1}{k\neq\ell}}^K C_k(\bm{u}^{I_{\cdot, k}^{\text{mat}}}),
  \end{align*}
  where $p_{\bm{I}}(\bm{k})=\P(I_1=k_1,\dots,I_d=k_d)>0$ for
  $\bm{k}=(k_1,\dots,k_d)$.  Due to the assumption
  $J_{\ell}(\bm{I}) = (\ell_1,\dots,\ell_m)$ independent of $\bm{I}$ we now have
  $\bm{u}^{I_{\cdot, \ell}^{\text{mat}}} = (u_{\ell_1},\dots,u_{\ell_m})$ independent of $I^{\text{mat}}$
  and the associated $\bm{I}$, and also of $\bm{k}$.  For $\bm{u}\in[0,1]^d$ we
  therefore have
  \begin{align*}
    C(\bm{u})
    &=\E_{\bm{I}}\biggl[\,\prod_{k=1}^KC_k(\bm{u}^{I_{\cdot, k}^{\text{mat}}})\biggr]
      = \sum_{\bm{k}\in\{1,\dots,K\}^d}p_{\bm{I}}(\bm{k})\prod_{k=1}^KC_k(\bm{u}^{I_{\cdot, k}^{\text{mat}}})
      = \sum_{\bm{k}\in\{1,\dots,K\}^d} C_{\ell}(u_{\ell_1,\dots,u_{\ell_m}}) p_{\bm{I}}(\bm{k})\prod_{\stackrel{k=1}{k\neq\ell}}^K C_k(\bm{u}^{I_{\cdot, k}^{\text{mat}}})\\
    &= C_{\ell}(u_{\ell_1,\dots,u_{\ell_m}}) \sum_{\bm{k}\in\{1,\dots,K\}^d} p_{\bm{I}}(\bm{k})\prod_{\stackrel{k=1}{k\neq\ell}}^K C_k(\bm{u}^{I_{\cdot, k}^{\text{mat}}}).\qedhere
  \end{align*}
\end{proof}

\begin{proof}[Proof of Proposition~\ref{prop_extendibility}]
  Fix $k\in\{0,\ldots,K\}$. By extendibility of $C_k$, we can associate an
  exchangeable stochastic process $(U_{k,n})_{n=1}^{\infty}$ with $C_k$ such
  that $(U_{k,j_1},\ldots,U_{k,j_d})\sim C_k$ for all
  $1 \leq j_1 < j_2 < \cdots < j_d$. Since $(U_{k,n})_{n=1}^{\infty}$ is
  exchangeable, de Finetti's theorem implies that, conditional on the
  exchangeable $\sigma$-algebra $\mathcal{E}_k$,
  $(U_{k,n}\,|\,\mathcal{E}_k)_{n=1}^{\infty}$ is a sequence of iid random
  variables; see \cite[Theorem~4.7.9]{durrett2019}. Now define a new stochastic
  process $(U_n)_{n=1}^{\infty}$ via
  \begin{align*}
    U_n:=U_{I_n,n},
  \end{align*}
  where $I_n := G^{-1}(U_{0,n})\in\{1,\ldots,K\}$ is the $n$th element of the exchangeable process $(G^{-1}(U_{0,n}))_{n=1}^{\infty}$ with $(I_{j_1},\ldots,I_{j_d})\sim\bm{I}$ for all integer indices
  $1 \leq j_1 < j_2 < \cdots < j_d$.

  The random variable $U_n$ is a function of $U_{1,n},\ldots,U_{K,n}$ and $I_n$.
  Since functions of independent random variables are independent, this implies
  that $(U_n\,|\,\sigma(\mathcal{E}_0,\ldots,\mathcal{E}_K))_{n=1}^{\infty}$ is
  a sequence of iid random variables
  $(\sigma(\mathcal{E}_0,\ldots,\mathcal{E}_K)$ being the $\sigma$-algebra
  generated by $\mathcal{E}_0,\dots,\mathcal{E}_K$). Hence
  $(U_n)_{n=1}^{\infty}$ is conditionally iid, so exchangeable. In order to show
  that every $d$-dimensional margin of $(U_n)_{n=1}^{\infty}$ is given by $C$,
  consider now indices $1 \leq j_1 < j_2 < \cdots < j_d$.  For a fixed vector
  $\bm{i}=(i_1,\ldots,i_d)$ and associated (deterministic) index-matrix
  $I^{\text{mat}}(\bm{i})$, we then have, as in the proof of
  Theorem~\ref{thm:index:mixed:form:C} (due to the independence of the
  individual stochastic processes), that
  \begin{align*}
    \P(U_{i_1,j_1}\leq u_1,\ldots,U_{i_d,j_d}\leq u_d)
    &= \prod_{k=1}^K \P(U_{k,j_1}\leq u_1^{I^{\text{mat}}(\bm{i})_{1,k}},\ldots,U_{k,j_d}\leq u_d^{I^{\text{mat}}(\bm{i})_{d,k}})\\
    &= \prod_{k=1}^K \P(U_{k,1}\leq u_1^{I^{\text{mat}}(\bm{i})_{1,k}},\ldots,U_{k,d}\leq u_d^{I^{\text{mat}}(\bm{i})_{d,k}})\\
    &= \prod_{k=1}^K C_k(u_1^{I^{\text{mat}}(\bm{i})_{1,k}},\ldots,u_d^{I^{\text{mat}}(\bm{i})_{d,k}}),
  \end{align*}
  where we used the exchangeability of $(U_{k,n})_{n=1}^{\infty}$,
  $1\leq k\leq K$, and the extendibility of $C_k$.  Due to exchangeability of
  $(I_n)_{n=1}^{\infty}$, we have the distributional equality
  $\bm{I}(\bm{j}):=(I_{j_1},\ldots,I_{j_d}) \stackrel{d}{=}
  (I_{1},\ldots,I_{d})=\bm{I}$. For $\bm{i}=(i_1,\ldots,i_d)$, we thus obtain
  that
  \begin{align*}
    \P(U_{j_1}\leq u_1,\ldots,U_{j_d}\leq u_d)&= \P(U_{I_{j_1},j_1}\leq u_1,\ldots,U_{I_{j_d},j_d}\leq u_d)\\ %
    &=\E_{\bm{I}(\bm{j})}\bigl[\P(U_{i_1,j_1}\leq u_1,\ldots,U_{i_d,j_d}\leq u_d)\,\big|\, \bm{I}(\bm{j})=\bm{i}\bigr]\\ %
    &=\E_{\bm{I}}\bigl[\P(U_{i_1,j_1}\leq u_1,\ldots,U_{i_d,j_d}\leq u_d)\,\big|\, \bm{I}=\bm{i}\bigr]\\ %
    &=\E_{\bm{I}}\biggl[\,\prod_{k=1}^K C_k(u_1^{I^{\text{mat}}(\bm{i})_{1,k}},\ldots,u_d^{I^{\text{mat}}(\bm{i})_{d,k}}) \,\bigg|\, \bm{I}=\bm{i}\biggr]\\ %
    &=\E_{\bm{I}}\biggl[\,\prod_{k=1}^K C_k(u_1^{I^{\text{mat}}_{1,k}},\ldots,u_d^{I^{\text{mat}}_{d,k}})\biggr]=C(u_1,\ldots,u_d), %
  \end{align*}
  showing that the index-mixed copula $C$ is extendible.
\end{proof}

\section{EFGM copulas}\label{sec:EFGM}
Multivariate distributions build on Eyraud--Farlie--Gumbel--Morgenstern (EFGM) copulas are
popular in the literature, see for example \cite{cambanis1977},
\cite{johnsonkott1975}, \cite[Section~44.13]{kotzbalakrishnanjohnson2000} or
\cite[Chapter~5]{kotzdrouet2001}, to name a few.
They originally date back to the work of
\cite{eyraud1936}, \cite{morgenstern1956}, \cite{farlie1960} and
\cite{gumbel1960}; see also
\cite[Section~1.6.3]{jaworskidurantehaerdlerychlik2010}. The popularity of EFGM
copulas mostly stems from their analytical tractability, oftentimes allowing for
explicit calculations of dependence-related quantities of interest. However, it
is well known that Spearman's rho and Kendall's tau of a bivariate EFGM copula
fall within the ranges $\left[-1/3,1/3 \right]$ and $\left[ -2/9,2/9 \right]$,
respectively; see
\cite[Example~5.7]{nelsen2006}. %
Several bivariate extensions of this family have been proposed in the literature
with the intention of relaxing this limitation; see for example
\cite{lai2000new}, \cite{amblard2002symmetry}, \cite{rodriguez2004new}, and
\cite{fischer2007constructing}. A four-parameter extension, with Spearman's rho
in $[-0.48, 0.5]$, has been studied by \cite{bairamov2001new}.  Another
extension is by \cite{amblard2009new}, for which Spearman's rho is within the
interval $[-3/4, 1]$. More recently, \cite{hurlimann2017comprehensive} developed
a comprehensive extension (called the Hoeffding-Fr{\'e}chet extension of EFGM
copulas) that allows for any Spearman's rho and Kendall's tau to be attained,
but it has not been generalized to higher dimensions.  Also, EFGM copulas are
tail independent. Because of these main limitations alone, the use of EFGM
copulas in typical copula applications in finance, insurance and risk management is
essentially non-existing. %

In this appendix, we consider a recent construction of EFGM copulas and reveal,
through the idea of index-mixing, properties of EFGM copulas such as why it only
has a limited range of concordance, no tail dependence, and why its construction
is so restricting in comparison to index-mixed copulas.

\subsection{Construction of EFGM copulas}\label{sec:EFGM:construct}
Let $\Pi(\bm{u})=\prod_{j=1}^du_j$ denote the independence copula. Furthermore, for $\bm{p}\in[0,1]^d$ let $\bm{B}\sim\B^{C_{\bm{B}}}_d(\bm{p})$ denote that the random vector $\bm{B}=(B_1,\ldots,B_d)$ has stochastic
representation $\bm{B}=(F_{\B(p_1)}^{-1}(V_1),\dots,F_{\B(p_d)}^{-1}(V_d))$,
where $\bm{V}=(V_1,\ldots,V_d)\sim C_{\bm{B}}$ for some $d$-dimensional copula $C_{\bm{B}}$ and
$F_{\B(p)}(x)=(1-p)\I_{[0,\infty)}(x)+p\I_{[1,\infty)}(x)$, $x\in\IR$, is
the distribution function of the Bernoulli distribution $\B(p)$ with success
probability $p$. Moreover, denote by $\mathcal{B}_d(\bm{p}):=\{B^{C_{\bm{B}}}_d(\bm{p})| C_{\bm{B}} \text{ is a $d$-copula}\}$ the Fr\'echet
class of all $d$-dimensional distributions with $j$th margin being $\B(p_j)$,
$j=1,\dots,d$.  If the components of $\bm{p}$ are all equal to $p\in[0,1]$, we
simply write $\B^{C_{\bm{B}}}_d(p)$ or $\mathcal{B}_d(p)$; for example, $\B^{C_{\bm{B}}}_d(p)$ means
$\B^{C_{\bm{B}}}_d(\bm{p})$ for $\bm{p}=(p,\dots,p)\in[0,1]^d$.

In contrast to the commonly adopted notion of a `stochastic representation' (see
our Section~\ref{sec:intro}), \cite{blierwongcossettemarceau2022a} call
`stochastic representation' a specific analytical form of a $d$-dimensional EFGM
copula, namely
\begin{align}
  C^{\text{EFGM}}(\bm{u})=\E_{\bm{B}}\biggl[\,\prod_{j=1}^dF_{\tilde{U}_{(B_j+1)}}(u_j)\biggr],\quad\bm{u}\in[0,1]^d,\label{eq:EFGM:cop:mix}
\end{align}
where $\bm{B}=(B_1,\dots,B_d)\in\mathcal{B}_d(1/2)$, $\tilde{U}_1,\tilde{U}_2\isim\U(0,1)$
and $F_{\tilde{U}_{(1)}}(x)=1-(1-x)^2=x(2-x)$, $x\in[0,1]$, is the distribution function of
$\tilde{U}_{(1)}=\min\{\tilde{U}_1,\tilde{U}_2\}\sim\Beta(1,2)$ and
$F_{\tilde{U}_{(2)}}(x)=x^2$, $x\in[0,1]$, is the distribution function of
$\tilde{U}_{(2)}=\max\{\tilde{U}_1,\tilde{U}_2\}\sim\Beta(2,1)$.

\cite[Theorem~2]{blierwongcossettemarceau2022c} present a stochastic
representation of $\bm{U}\sim C^{\text{EFGM}}$ in the sense of our Section~\ref{sec:intro} via
\begin{align}
  \bm{U}=(\bm{1}-\bm{B})\tilde{\bm{U}}_{(1)}+\bm{B}\tilde{\bm{U}}_{(2)},\label{eq:EFGM:stoch:rep}
\end{align}
where $\bm{1}=(1,\dots,1)$ and $\tilde{\bm{U}}_{(k)}=(\tilde{U}_{1,(k)},\dots,\tilde{U}_{d,(k)})$, $k=1,2$,
for $\tilde{U}_{j,k}\isim\U(0,1)$, $j=1,\dots,d$, $k=1,2$. For simplicity of the presented arguments,
we call
\begin{align}\label{eq:rowwise:sorted}
R := (\tilde{\bm{U}}_{(1)},\tilde{\bm{U}}_{(2)})
  &=\begin{pmatrix}
    \tilde{U}_{1,(1)} & \tilde{U}_{1,(2)}\\
    \vdots & \vdots\\
    \tilde{U}_{d,(1)} & \tilde{U}_{d,(2)}
  \end{pmatrix}
  =\begin{pmatrix}
    \min\{\tilde{U}_{1,1},\tilde{U}_{1,2}\} & \max\{\tilde{U}_{1,1},\tilde{U}_{1,2}\}\\
    \vdots & \vdots\\
    \min\{\tilde{U}_{d,1},\tilde{U}_{d,2}\} & \max\{\tilde{U}_{d,1},\tilde{U}_{d,2}\}
  \end{pmatrix}
\end{align}
the \emph{rowwise sorted uniform matrix} of the construction. Note that an alternative
representation for~\eqref{eq:EFGM:stoch:rep} is $\bm{U}=\tilde{\bm{U}}_{(1)}^{\bm{1}-\bm{B}}\tilde{\bm{U}}_{(2)}^{\bm{B}}$,
so the representation is not unique.
The following result establishes the correctness
of~\eqref{eq:EFGM:cop:mix} and \eqref{eq:EFGM:stoch:rep} directly, avoiding the detour
via exponential distributions and Sklar's Theorem as in \cite{blierwongcossettemarceau2022a}.
\begin{proposition}[EFGM copulas]
  Let $\tilde{\bm{U}}_1,\tilde{\bm{U}}_2\isim\Pi$ be independent of
  $\bm{B}\in\mathcal{B}_d(1/2)$.  Then
  $\bm{U}=(\bm{1}-\bm{B})\tilde{\bm{U}}_{(1)}+\bm{B}\tilde{\bm{U}}_{(2)}$ in~\eqref{eq:EFGM:stoch:rep} follows
  the (bona fide) copula $C^{\text{EFGM}}$ in~\eqref{eq:EFGM:cop:mix} which can alternatively be written as
  \begin{align}
    C^{\text{EFGM}}(\bm{u})=\E_{\bm{B}}\biggl[\,\prod_{j=1}^d\bigl(u_j(1+(1-u_j)(-1)^{B_j})\bigr)\biggr],\quad\bm{u}\in[0,1]^d.\label{eq:EFGM:cop:mix:expanded}%
  \end{align}
\end{proposition}
\begin{proof}
  Let $\tilde{\bm{U}}_{(k)}=(\tilde{U}_{1,(k)},\dots,\tilde{U}_{d,(k)})$,
  $k=1,2$, be defined as before and note that each random vector has copula $\Pi$ (and beta
  margins) since measurable functions of independent random variables are
  independent. Furthermore,
  $\P((1-B_j)\tilde{U}_{j,(1)}+B_j\tilde{U}_{j,(2)}\le
  u_j\,|\,B_j=b_j)=F_{\tilde{U}_{(b_j+1)}}(u_j)$, $b_j\in\{0,1\}$. Therefore,
  \begin{align*}
    \P(\bm{U}\le\bm{u})&=\E_{\bm{B}}\left[\P\left((\bm{1}-\bm{B})\tilde{\bm{U}}_{(1)}+\bm{B}\tilde{\bm{U}}_{(2)}\le\bm{u}\,|\,\bm{B}\right)\right]\\
                       &=\E_{\bm{B}}\left[\P\left((1-B_j)\tilde{U}_{j,(1)}+B_j\tilde{U}_{j,(2)}\le u_j\ \forall\,j\,|\,\bm{B}\right)\right]\\
                       &=\E_{\bm{B}}\biggl[\,\prod_{j=1}^d\P\left((1-B_j)\tilde{U}_{j,(1)}+B_j\tilde{U}_{j,(2)}\le u_j\,|\,B_j\right)\biggr]\\
    &=\E_{\bm{B}}\biggl[\,\prod_{j=1}^dF_{\tilde{U}_{(B_j+1)}}(u_j)\biggr],\quad\bm{u}\in[0,1]^d,
  \end{align*}
showing that the joint distribution of $\bm{U}$ is indeed equal to $C^{\text{EFGM}}$ given in \eqref{eq:EFGM:cop:mix}.
  Noting that
  \begin{align*}
    F_{\tilde{U}_{(b_j+1)}}(u_j)=u_j(1+(1-u_j)(-1)^{b_j}),\quad u_j\in[0,1],\ b_j\in\{0,1\},
  \end{align*}
  we obtain the form of $C^{\text{EFGM}}$ as in~\eqref{eq:EFGM:cop:mix:expanded}.
  With $u_k=1$ for all $k\neq j$ and noting that $\E[(-1)^{B_j}]=0$,
  $j=1,\dots,d$, we have the $j$th margin with $u_j\in[0,1]$ that
  \begin{align*}
    \P(U_j\le u_j)&=\E_{B_j}\bigl[F_{\tilde{U}_{(B_j+1)}}(u_j)\bigr]=\E_{B_j}\bigl[u_j(1+(1-u_j)(-1)^{B_j})\bigr]
=u_j\bigl(1+(1-u_j)\E[(-1)^{B_j}]\bigr)=u_j,
  \end{align*}
  so $C^{\text{EFGM}}$ is indeed a copula.
\end{proof}

The classical analytical form of EFGM copulas, see also
\cite[Section~1.6.3]{jaworskidurantehaerdlerychlik2010}, can be derived
from~\eqref{eq:EFGM:cop:mix:expanded} as follows.
\begin{proposition}[Classical analytical form of EFGM copulas]\label{prop:EFGM:cop:classical:form}
For $\bm{u}\in[0,1]^d$ the copula $C^{\text{EFGM}}$ in~\eqref{eq:EFGM:cop:mix:expanded} is given by
  \begin{align}
    C^{\text{EFGM}}(\bm{u})=\biggl(\,\prod_{j=1}^d u_j\biggr)\biggl(1+\sum_{k=2}^d\sum_{1\le j_1<\dots<j_k\le d}\!\!\!\!\theta_{j_1,\dots,j_k}\prod_{l=1}^k(1-u_{j_l})\biggr),\label{eq:EFGM:cop:classical:form}
  \end{align}
  where $\theta_{j_1,\dots,j_k}=\E_{\bm{B}}[(-1)^{B_{j_1}+\dots+B_{j_k}}]$ for all $1\le j_1<\dots<j_k\le d$, $k=2,\dots,d$.
\end{proposition}
\begin{proof}
  We have
  \begin{align*}
    C^{\text{EFGM}}(\bm{u})&=\E_{\bm{B}}\biggl[\,\prod_{j=1}^d\bigl(u_j(1+(1-u_j)(-1)^{B_j})\bigr)\biggr]=\biggl(\,\prod_{j=1}^du_j\biggr)\,\E_{\bm{B}}\biggl[\,\prod_{j=1}^d(1+(1-u_j)(-1)^{B_j})\biggr].
  \end{align*}
  By the multi-binomial
  theorem, %
  \begin{align*}
    \prod_{j=1}^d(1+a_j)=\sum_{\bm{j}\in\{0,1\}^d}\prod_{l=1}^da_j^{j_l}=1+\sum_{k=1}^da_k+\sum_{k=2}^d\,\sum_{1\le
    j_1<\dots<j_k\le d}\,\prod_{l=1}^ka_{j_l},
  \end{align*}
  and so, since $\E[(-1)^{B_k}]=0$ for all $k=1,\dots,d$,
  \begin{align*}
    \E_{\bm{B}}\biggl[\,\prod_{j=1}^d(1+(1-u_j)(-1)^{B_j})\biggr]
    =1+0+\sum_{k=2}^d\,\sum_{1\le
    j_1<\dots<j_k\le d}\!\!\!\!\E_{\bm{B}}[(-1)^{B_{j_1}+\dots+B_{j_k}}]\prod_{l=1}^k(1-u_{j_l}),
  \end{align*}
  from which the form of $C^{\text{EFGM}}$ as in~\eqref{eq:EFGM:cop:classical:form} follows
  when setting $\theta_{j_1,\dots,j_k}=\E_{\bm{B}}[(-1)^{B_{j_1}+\dots+B_{j_k}}]$ for all $1\le j_1<\dots<j_k\le d$, $k=2,\dots,d$.
\end{proof}

From Proposition~\ref{prop:EFGM:cop:classical:form} it follows that EFGM copulas are absolutely continuous with density
\begin{align*}
  c^{\text{EFGM}}(\bm{u})&=\E_{\bm{B}}\biggl[\,\prod_{j=1}^d(1+(1-2u_j)(-1)^{B_j})\biggr]
  = 1+\sum_{k=2}^d\sum_{1\le j_1<\dots<j_k\le d}\!\!\!\!\theta_{j_1,\dots,j_k}\prod_{l=1}^k(1-2u_{j_l}),\quad\bm{u}\in(0,1)^d,
\end{align*}
where the last equality follows as in the proof of
Proposition~\ref{prop:EFGM:cop:classical:form}. It follows that the $2^d-d-1$
EFGM copula parameters need to fulfill (the rather complicated set of inequalities)
\begin{align*}
  1+\sum_{k=2}^d\sum_{1\le j_1<\dots<j_k\le d}\theta_{j_1,\dots,j_k}\eps_{j_1}\dots\eps_{j_k}\ge 0
\end{align*}
for all $\eps_{j_1},\dots,\eps_{j_k}\in\{-1,1\}^d$. In turn, all parameters
satisfying this condition lead to valid EFGM copula densities; see also
\cite{cambanis1977}.

\subsection{Understanding the construction of EFGM copulas through index-mixing}\label{sec:EFGM:interpret}
A first observation is that $\bm{U}$ in~\eqref{eq:EFGM:stoch:rep} can be
written as
\begin{align}
  \bm{U} = R_{(\bm{I})}:=\begin{pmatrix} R_{1,I_1}\\\vdots\\ R_{d,I_d}\end{pmatrix}\quad\text{for}\ \bm{I}=\bm{1}+\bm{B}\ \text{with}\ \bm{B}\in\mathcal{B}_d(1/2),\label{eq:stoch:rep:reint}
\end{align}
where $R$ is the associated rowwise sorted uniform matrix; see \eqref{eq:rowwise:sorted}.
We see from~\eqref{eq:stoch:rep:reint} that the random index vector $\bm{I}$
selects, with its $j$th component $I_j$, componentwise order statistics (either
the minimum or the maximum) of two iid $\U(0,1)$ random variables.
In other terms, we can view the construction of $\bm{U}\sim C^{\text{EFGM}}$ as
randomly selecting one element of each row of the rowwise sorted uniform
matrix $R$ with equal probability via $\bm{I}$.
The adverbial phrase ``with equal probability'' in the last sentence is important.
Only because $\bm{B}$ has \emph{symmetric} Bernoulli margins $\B(1/2)$ and we
thus obtain that $I_j$ randomly selects with equal probability between
$\min\{\tilde{U}_{j,1},\tilde{U}_{j,2}\}$ and
$\max\{\tilde{U}_{j,1},\tilde{U}_{j,2}\}$ do we obtain that $U_j$ indeed
has $\U(0,1)$ margins and thus that $\bm{U}$ follows a copula, the EFGM copula.  As
such, EFGM copulas are not only limited in the sense that the construction starts
from iid $\U(0,1)$ variables, but also in the sense that the set of
distributions of $\bm{I}$ is restricted. Neither of these limitations exists
for index-mixed copulas.

\subsection{Understanding the strange and limited dependence of EFGM copulas}\label{sec:EFGM:limited:dep}
EFGM copulas exhibit the strongest concordance if the components of the index
vector $\bm{I}$ are comonotone, in which case we have the stochastic
representation $\bm{I}=\bm{1}+(F_{\B(1/2)}^{-1}(U),\dots,$ $F_{\B(1/2)}^{-1}(U))$
for $U\sim\U(0,1)$, so $\bm{I}\in\{\bm{1},\bm{2}\}$ with probability $1/2$ each.
All that comonotonicity of $\bm{I}$ does is enforcing that \emph{the same}
column of the rowwise sorted uniform matrix is returned as $\bm{U}$ (so
$\bm{U}=\tilde{\bm{U}}_{(1)}$ or $\bm{U}=\tilde{\bm{U}}_{(2)}$), and each of the
two columns is picked at random with probability $1/2$. However, for each fixed
$k=1,2$, the copula of $\tilde{\bm{U}}_{(k)}$ is the copula between $d$ $k$th
order statistics of iid $\U(0,1)$ random variables, which due to the
independence between the rows is indeed the independence copula! So the copula
of $\tilde{\bm{U}}_{(1)}$ is the independence copula (the margins are
$\Beta(1,2)$), and equally the copula of $\tilde{\bm{U}}_{(2)}$ is the
independence copula (the margins are $\Beta(2,1)$). Only due to the fact that we
\emph{randomly} select each of these two column vectors whose components are
independent do we obtain a random vector $\bm{U}$ with dependent components,
and, as mentioned before, is this random vector $\bm{U}$ guaranteed to have
$\U(0,1)$ margins. It is thus not surprising that the resulting dependence
(tail, concordance) is fairly limited.

\begin{example}[Samples from EFGM copulas]
  Figure~\ref{fig:EFGM} shows samples of size 10\,000 from two EFGM copulas. The
  sample on the left-hand side was constructed with independent $1+\B(1/2)$
  components of the index vector $\bm{I}$, so the resulting EFGM copula is the
  independence copula. The sample on the right-hand side was constructed with
  comonotone $1+\B(1/2)$ components of the index vector $\bm{I}$, so the
  resulting EFGM copula is the most concordant EFGM copula there is.  While
  differences between the scatter plots can be spotted (mostly in the corners of
  $[0,1]^2$), they are rather minute, and classical log-returns of financial
  data, say, exhibit substantially more concordance; see, for example,
  \cite{hofertoldford2018a}.  We know that Spearman's rho of this maximally
  concordant EFGM copula is only $\rho_{\text{S}}=1/3$, and Kendall's tau is only
  $\tau=2/9$.
  \begin{figure}[htbp]
    \includegraphics[width=0.48\textwidth]{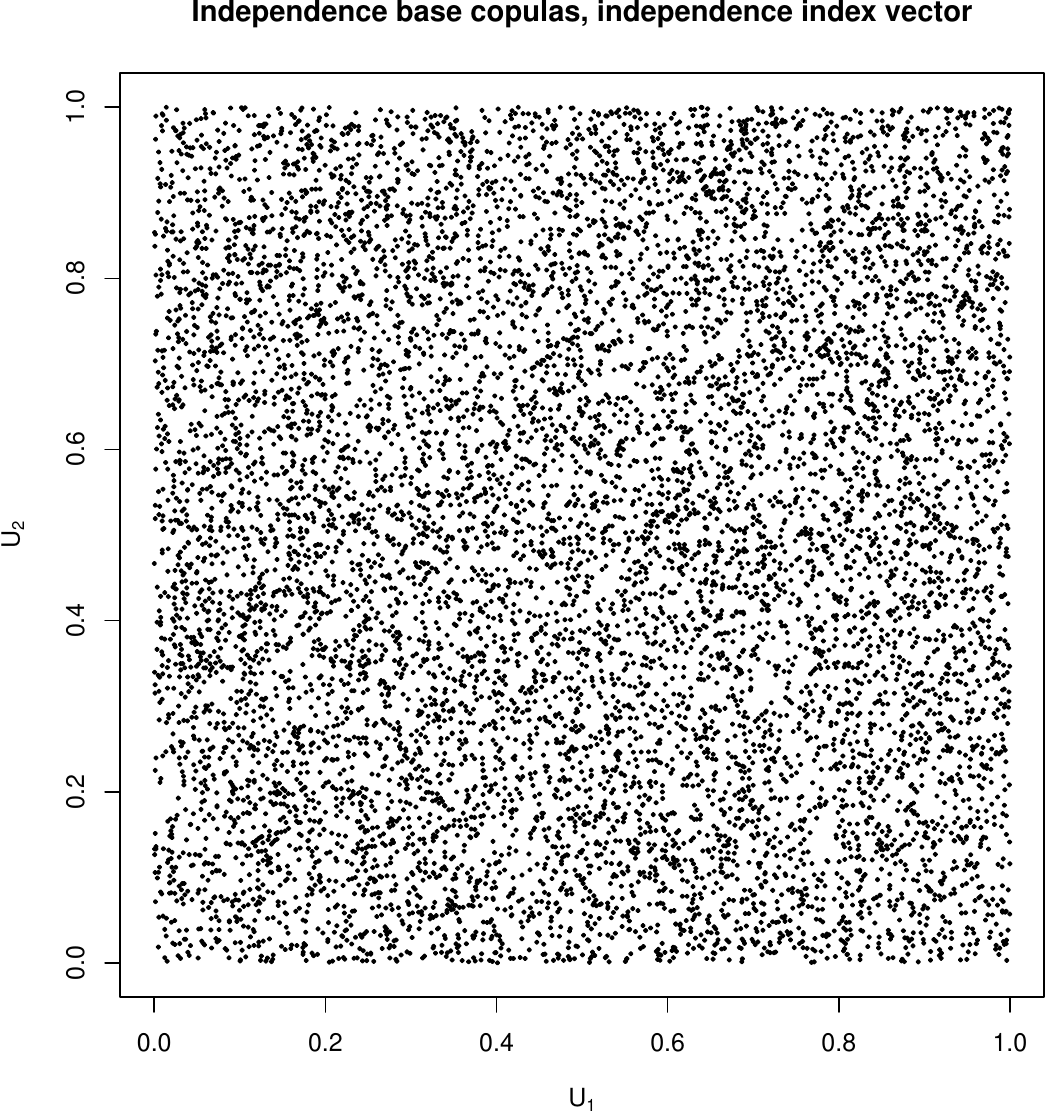}%
    \hfill
    \includegraphics[width=0.48\textwidth]{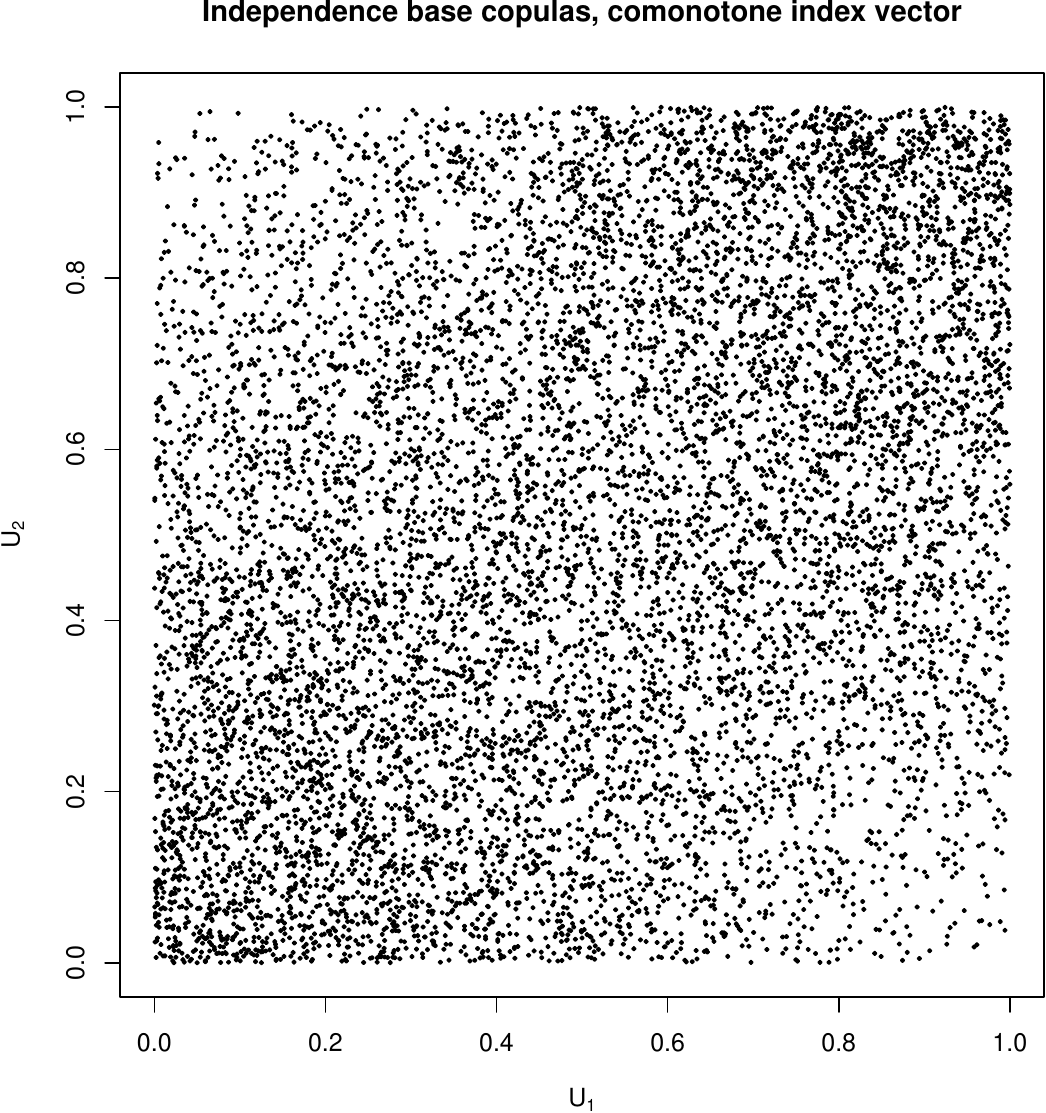}%
    \caption{Sample of size 10\,000 from EFGM copulas with independent (left) and
      comonotone (right) index vector $\bm{I}$ with $1+\B(1/2)$ margins. The
      resulting copulas are the independence copula and the most concordant EFGM
      copula, respectively.}
    \label{fig:EFGM}
  \end{figure}
\end{example}

\subsection{Limitations of sorting}\label{sec:unnatural:sort}
The rowwise sorting of the elements in $R$ is rather unnatural for two main
reasons.  First, consider the $j$th row of $R$. As mentioned before, the index
vector $\bm{I}$ selects one of the two elements in the $j$th row of the rowwise
sorted uniform matrix $R$ at random (which guarantees that $U_j\sim\U(0,1)$), so
sorting the values first does not seem to make a difference. However, this is
required to have any chance of obtaining dependent components in $\bm{U}$ since
if the rowwise sorted uniform matrix would just consist of iid $\U(0,1)$ random
variables without the sorting step, then we would always obtain $\bm{U}\sim\Pi$,
regardless of the dependence in $\bm{I}$ (be it independent or comonotone).

Second, the dependence between the elements of the $j$th row of the sorted
uniform matrix $R$ is rather a nuisance due to the fact that the $j$th
component $U_j$ of $\bm{U}$ is still required to be $\U(0,1)$ distributed for
all $j=1,\dots,d$. So the fact that the entries in each row of $R$
are sorted requires us to draw an element from each row with \emph{equal} probability.
This limits the flexibility in selecting an element from the $j$th row
of the rowwise sorted uniform matrix.

\subsection{Why not incorporating sorting into the construction of index-mixed copulas}
While Definition~\ref{def:index:mixed:copulas} purposely avoids the rowwise
sorting encountered in the EFGM copula construction (see
Section~\ref{sec:EFGM:construct} and Section~\ref{sec:EFGM:interpret}), it is in
principle possible to incorporate rowwise sorting of the copula matrix as an
additional step, which would then also imply that index-mixed copulas truly
generalize EFGM copulas. However, in order to retain the full flexibility of the
index distribution (compare with Section~\ref{sec:unnatural:sort}), it is
necessary to ensure that all entries of the copula matrix follow a $\U(0,1)$
distribution. In case of rowwise sorting this could be accomplished by applying
the $\Beta(k, K+1-k)$ distribution function to each entry in column $k$ of the
rowwise sorted copula matrix $\tilde{U}$. However, indexing $\tilde{U}$ by, say,
$\bm{I}=(k,\dots,k)$ (so selecting the $k$th column of $\tilde{U}$) makes it
harder to imagine and determine the resulting copula since the components in the
$k$th column of $\tilde{U}$ now (after rowwise sorting) can come from any of the
$K$ different base copulas, whereas without the rowwise sorting we know that the
$k$th column is distributed as the $k$th base copula $C_k$. In fact, for every
realization of a rowwise sorted copula matrix $\tilde{U}$, there is precisely
one index vector $\bm{I}$ that selects those $d$ elements from $\tilde{U}$ that
follow $C_k$.  Thus incorporating rowwise sorting into $\tilde{U}$ does not seem
to have any conceptual advantage, it rather convolutes the construction by
making it harder to grasp how dependence arises from index-mixing.  In contrast,
index-mixed copulas have several advantages when compared to EFGM copulas, such
as being able to capture the full range of concordance through dependent (as
opposed to independent) random vectors entering the construction, also $K>2$
such random vectors can enter the construction, the index vector not being
restricted to symmetric distributions, and, for example, hierarchies of
dependencies being easily incorporated by choosing hierarchical base
copulas.

\end{document}

%
%
%
%

%%% Local Variables:
%%% mode: LaTeX
%%% TeX-master: t
%%% End: